%% file: Paper_RoCTLstar_Express_Journal.tex
\documentclass[english,british,prodmode,acmtocl]{acmsmall}
\usepackage[T1]{fontenc}
\usepackage[latin9]{inputenc}
\usepackage{color}
\usepackage{prettyref}
\usepackage{amsmath}
\usepackage{amssymb}
\usepackage{graphicx}
\usepackage[authoryear]{natbib}
\PassOptionsToPackage{normalem}{ulem}
\usepackage{ulem}

\makeatletter

\newenvironment{lyxlist}[1]
{\begin{list}{}
{\settowidth{\labelwidth}{#1}
 \setlength{\leftmargin}{\labelwidth}
 \addtolength{\leftmargin}{\labelsep}
 }}
{\end{list}}

\usepackage{hyperref}
\usepackage{algorithm}
\usepackage{algorithmic}
\providecommand{uline}[1]{\emph{#1}}

\acmYear{Draft: \today{}}

\makeatother

\usepackage{babel}
\begin{document}
\global\long\def\stateX{x}
\renewcommand\marginpar{\adsfafdsdf}

\global\long\def\stateY{y}

\global\long\def\stateS{s_{\phi}}

\global\long\def\stateT{t_{\phi}}

\global\long\def\stateA{s_{\phi}}

\global\long\def\stateB{t_{\phi}}

\providecommand\thispaper{this paper}

\providecommand{\ifnoroot}[1]{}

\providecommand\UWA{\\The University of Western Australia\\Computer
Science and Software Engineering}

 \markboth{J. C. McCabe-Dansted et al.}{Specifying Robustness}

\title{Specifying Robustness}

\author{John C. McCabe-Dansted,
Tim French \and{}
Mark Reynolds \affil{University of Western Australia}
Sophie Pinchinat \affil{Campus Universitaire de Beaulieu}
}

\input{sprobIntroCut.tex}

\section{RoCTL{*}}

\label{sec:roctl}\label{sub:RoCTL-Structures} In this section
we define the RoCTL{*} logic. We first provide some basic definitions,
starting with our set of variables.
\begin{defi}
We let $\Var$ be our set of variables. The set $\Var$ contains a
special variable $\Viol$. A \uline{valuation} $\valuation$ is
a map from a set of worlds $\mfw$ to the power set of the variables.
The statement $p\in\valuation(\mfm)$ means roughly ``the variable
$p$ is true at world $\mfm$''.
\end{defi}

The $\Viol$ atom%
\footnote{A variant of RoCTL{*} was presented in~\citep{FrDaRe07book}, which
had two accessibility relations, a success and failure transition
and thus did not need the special atom $\Viol$. The definition we
use here was presented in \citep{DBLP:conf/jelia/McCabe-Dansted08}).
These definitions are equivalent if we disallow the RoCTL{*} formulas
from directly accessing the $\Viol$ atom \citep{MCD10}. All the
known results on RoCTL{*} apply equally well using either definition,
and no advantage is known to the definition in \citep{FrDaRe07book}.
Using the definition of the structures for RoCTL{*} that have a single
accessibility relation allows us to define both CTL{*} and RoCTL{*}
structures in the same way, greatly simplifying the definition of
the translations.%
} will be used to define failing transitions. Informally it may be
possible to enter a state labelled with $\Viol$, but it is forbidden
to do so; entering such a state will be considered a failure. 

As is normal we say a binary relation is serial if every element has
a successor.
\begin{defi}
We say that a binary relation $R$ on $S$ is \uline{serial} (total)
if for every $a$ in $S$ there exists $b$ in $S$ such that $aRb$.
\end{defi}

We now provide a definition of a structure. 
\begin{defi}
A \uline{structure} $M=\ctlstruct$ is a 3-tuple containing a
set of worlds $\mfw$, a serial binary relation $\ra$ on $\allworlds$,
a valuation $\valuation$ on the set of worlds $\mfw$.
\end{defi}

While in some logics the truth of formulas depends solely on the current
world, the truth of CTL{*} (and hence QCTL{*} and RoCTL{*}) may depend
on which future eventuates. These futures are represented as infinitely
long (full) paths through the structure. For this reason, we provide
a formal definition of fullpaths.
\begin{defi}
We call an $\omega$-sequence $\sigma=\left\langle \mfm_{0},\mfm_{1},\ldots\right\rangle $
of worlds a \uline{fullpath} iff for all non-negative integers
$i$ we have $w_{i}\ra w_{i+1}$. For all $i$ in $\natnum$ we define
$\patha_{\geq i}$ to be the fullpath $\left\langle \mfm_{i},\mfm_{i+1},\ldots\right\rangle $,
we define $\patha_{i}$ to be $w_{i}$ and we define $\patha_{\leq i}$
to be the sequence $\left\langle \jzseq wi\right\rangle $. 
\end{defi}

We now define the property of failure-freeness. This means that, in
the future, no failing transitions are taken. Informally, a failure-free
fullpath represents a perfect future. Whereas the Obligatory operator
in SDL quantifies over acceptable worlds, the Obligatory operator
we will define quantifies over failure-free fullpaths.

\begin{defi}
We say that a fullpath $\sigma$ is \uline{failure-free} iff for
all $i>0$ we have $\Viol\notin\valuation\left(\sigma_{i}\right)$.
We define $\SF(w)$ to be the set of all fullpaths starting with world
$\mfm$ and $S(w)$ to be the set of all failure-free fullpaths starting
with $\mfm$. We call a structure a RoCTL-structure iff $S(w)$ is
non-empty for every $w\in\allworlds$.
\end{defi}

We will now define deviations. Informally, these represent the possibility
of adding an additional failure to some step $i$ along a path. After
$i$ we follow a different path, and we allow only a single failure
not on the existing path so no failures occur after $i+1$. Deviations
are intended to represent possible failures we may wish to be able
to recover from, and if our system is robust to failures we also want
it to be robust in the face of correct transitions. For this reason
we allow the new transition added at step $i$ to be a success as
well as a failure. 
\begin{defi}
For two fullpaths $\patha$ and $\pathb$ we say that $\pathb$ is
an \uline{$i$-deviation} from $\patha$ iff $\patha_{\leq i}=\pathb_{\leq i}$
and $\pathb_{\geq i+1}\in S(\pathb_{i+1})$. We say that $\pathb$
is a \uline{deviation} from $\patha$ if there exists a non-negative
integer $i$ such that $\pathb$ is an $i$-deviation from $\patha$.
We define a function $\delta$ from fullpaths to sets of fullpaths
such that where $\patha$ and $\pathb$ are fullpaths, $\pathb$ is
a member of $\delta(\patha)$ iff $\pathb$ is a deviation from $\patha$. 
\end{defi}
We see that $S\left(\sigma_{0}\right)\subseteq\delta(\patha)\subseteq\SF(\sigma_{0})$.
 Where $p$ varies over $\Var$, we define RoCTL{*} \formulae{}
according to the following abstract syntax 
\begin{align*}
\forma & :=p\BNFbar\neg\forma\BNFbar\lAnd\forma\forma\BNFbar\lUntil\forma\phi\BNFbar N\phi\BNFbar A\forma\BNFbar O\forma\BNFbar\eA\phi\mbox{ .}
\end{align*}

A formula that begins with $A$, $\neg A$, $O$, $\neg O$, $p$
or $\neg p$ is called a \term{state formula}. For consistency with
\citep{FrDaRe07book}, we do not consider a formula that explicitly
contains $\Viol$ to be a RoCTL{*} formula, although our translation
into CTL{*} works equally well for such \formulae{}. The $\true,\,\neg,\,\wedge,\,\tN,\,\tU$
and $A$ are the familiar ``true'', ``not'', ``and'', ``next'',
``until'' and ``all paths'' operators from CTL\@. The abbreviations
$\false$, $\vee$, $F$, $G$, $W$, $E$ $\rightarrow$ and $\leftrightarrow$
are defined as in CTL{*} logic. As with Standard Deontic Logic (SDL)
logic, we define $\dP\equiv\neg\dO\neg$. Finally, we define the dual
$\prone$ of $\robust$ as the abbreviation $\eE\equiv\neg\eA\neg$.
 We call the $O$, $P$, $\eA$, $\eE$ operators Obligatory, Permissible,
Robustly and Prone respectively.

We define truth of a RoCTL{*} formula $\forma$ on a fullpath $\sigma=\left\langle \mfm_{0},\mfm_{1},\ldots\right\rangle $
in a RoCTL-structure $\mfk$ recursively as follows:
\begin{align*}
\mfk,\patha\forces\tN\forma & \tiff\mfk,\patha_{\geq1}\forces\forma\\
\mfk,\patha\forces\forma\tU\formb & \tiff\exists_{i\in\natnum}\st\mfk,\patha_{\geq i}\forces\formb\tand\\
 & \qquad\forall_{j\in\natnum}j<i\implies\mfk,\patha_{\geq j}\forces\phi\\
\mfk,\patha\forces A\forma & \tiff\forall_{\pathb\in\SF(\patha_{0})}\mfk,\pathb\forces\forma\\
\mfk,\patha\forces O\forma & \tiff\forall_{\pathb\in S(\sigma_{0})}\mfk,\pathb\forces\forma\\
\mfk,\patha\forces\eA\forma & \tiff\forall_{\pathb\in\delta(\patha)}\mfk,\pathb\forces\forma\mbox{ and }\mfk,\patha\forces\forma
\end{align*}
The definition for $\true$, $p$, $\neg$ and $\wedge$ is as we
would expect from classical logic.  The intuition behind the $\robust$
operator is that it quantifies over paths that could result if a single
error was introduced; the deviations only have at most one failure
not on the original path, and they are identical to the original path
until this failure occurs.
\begin{defi}
We say that a function $\tau$ from \formulae{} to \formulae{} is
\uline{truth-preserving} iff for all $M,\sigma$ and $\phi$ it
is the case that $M,\sigma\forces\phi\iff M,\sigma\forces\tau\left(\phi\right)$\textup{.}
\end{defi}
Given that traditional modal logics define truth at worlds, instead
of over paths, many important properties of modal logics assume such
a definition of truth. When dealing with those properties we can use
the following definition of truth of RoCTL{*} formulas at worlds. 
\begin{defi}
A RoCTL{*} formula is true at a world if it is true on any path leading
from that world, or more formally: 
\begin{align*}
\mfk,w\forces\forma & \tiff\Exists{\pathb\st\pi_{0}=w}\mfk,\pathb\forces\forma\mathfullstop
\end{align*}

\end{defi}

\section{Examples}

\label{sec:examples}\input{Body_Examples__v.tex}

\section{Technical Preliminaries }

\label{sec:machinery}\label{sec:Preliminaries}

In this section we will provide definitions and background results
that will be used in this paper.  In \prettyref{sub:CTL*-and-QCTL*}
we will define CTL{*} and its syntactic extension QCTL{*}. In \prettyref{sub:Automata}
we will define various forms of Automata. In \prettyref{sub:Bisimulations}
we will define Bisimulations. We will discuss expressive equivalences
\prettyref{sec:Expressive-Equivalences}, in particular between LTL
and automata.

\subsection{\label{sub:CTL*-and-QCTL*}Trees, LTL, CTL{*} and QCTL{*} }

In \thispaper{} we will also briefly consider Linear Temporal Logic
(LTL), CTL{*} and an extension QCTL{*} of CTL{*}. For the purposes
of this paper we will define CTL{*} to be a syntactic restriction
of RoCTL{*} excluding the $O$ and $\robust$ operator.
\begin{defi}
Where $p$ varies over $\Var$, we define CTL{*} \formulae{} according
to the following abstract syntax
\begin{align*}
\forma & :=p\BNFbar\neg\forma\BNFbar\lAnd{\forma}{\forma}\BNFbar\lUntil{\forma}{\phi}\BNFbar N\phi\BNFbar A\phi\mbox{ .}
\end{align*}

\end{defi}
We likewise define LTL to be the restriction of CTL{*} without the
$A$ operator.
\begin{defi}
Where $p$ varies over $\Var$, we define LTL \formulae{} according
to the following abstract syntax
\begin{align*}
\forma & :=p\BNFbar\neg\forma\BNFbar\lAnd{\forma}{\forma}\BNFbar\lUntil{\forma}{\phi}\BNFbar N\phi\mbox{ .}
\end{align*}

\end{defi}
In turn we define QCTL{*} as an extension of CTL{*} with a $\forall$
operator. 
\begin{defi}
\label{QCTL*-syntax}A\emph{ QCTL{*}}\logicindex{QCTL*} formula has
the following syntax: 
\begin{align*}
\forma & :=p\BNFbar\neg\forma\BNFbar\lAnd{\forma}{\forma}\BNFbar\lUntil{\forma}{\phi}\BNFbar N\phi\BNFbar A\phi\BNFbar\Qp\phi\mathfullstop
\end{align*}
The semantics of $p$, $\neg$, $\wedge$, $U$, $N$, and $A$ are
the same as in CTL{*} and RoCTL{*}. Before defining the Kripke semantics
for QCTL{*} we need to define the concept of a $p$-variant. Informally
a $p$-variant $M'$ of a structure $M$ is a structure that identical
except in the valuation of the $p$ atom. 
\end{defi}
 
\begin{defi}
\label{def:p-variant}Given some CTL-structure $M=(\allworlds,\access,\valuation)$
and some $p\in\Var$, a $p$-\uline{variant} of $M$ is some structure
$M=(\allworlds,\access,\valuation')$ where $\valuation'(w)\setdiff\{p\}=\valuation(w)\setdiff\{p\}$
for all $w\in\allworlds$.

Under the Kripke semantics for QCTL{*}, $\Qp\phi$ is defined as
\begin{align*}
M,b\forces\Qp\alpha\Longleftrightarrow & \textrm{ For every }p\textrm{-variant }M'\textrm{ of }M\\
 & \quad\textrm{ we have }M',b\forces\alpha\ .
\end{align*}

\end{defi}
In \thispaper{} we will use the \uline{tree-semantics} for QCTL{*}.
These semantics are the same as the Kripke semantics except that,
whereas the Kripke semantics evaluates satisfiability over the class
$\classC$ of CTL-structures, the tree semantics evaluate satisfiability
over the class $\classCt$ of CTL-structures which are trees (see
\prettyref{def:tree} below for a formal definition of trees). This
changes which \formulae{} are satisfiable in the logic as, unlike
CTL{*}~\citep{em83}, QCTL{*} is sensitive to unwinding into a tree
structure~\citep{kup95}. Note that the atom $p$ in $\Qp$ often
called a \uline{variable}.
\begin{theorem}
\label{thm:tree-QCTL-decidable}The tree-semantics for QCTL{*} are
decidable. \citep{Fr06}
\end{theorem}
We now provide the formal definition of a tree.
\begin{defi}
\label{def:tree}We say $T=\ctlstruct$ is a $\Var$-labelled \uline{tree},
for some set $\Var$, iff\end{defi}
\begin{enumerate}
\item $\allworlds$ is a non-empty set of nodes
\item for all $x,y,z\in\allworlds$ if $\tuple{x,z}\in\ra$ and $\tuple{y,z}\in\ra$
then $x=y$
\item there does not exist any cycle $x_{0}\ra x_{1}\cdots\ra x_{0}$ through
$\ra$
\item there exists a unique node $x$ such that for all $y\in\allworlds$,
if $y\ne x$ there exists a sequence $x\ra x_{1}\cdots\ra y$ through
$\ra$. We call the node $x$ the \uline{root} of the tree $T$
\item the valuation $g$ (or labelling) is a function from $\allworlds$
to $2^{\Var}$, that is for each $w\in\allworlds$, $g\left(w\right)\subseteq\Var$\end{enumerate}
\begin{defi}
We define the height of a finite tree $T=\treestruct$ as follows:
we let $\func{root}\left(T\right)$ be the root of the tree $T$.
We let $\heightT\left(T\right)=\height\left(\rootT\left(T\right)\right)$
where $\height$ is a function from $\allworlds$ to $\natnum$ such
that for all $x\in\allworlds$, we let $\height\left(x\right)$ be
the smallest non-negative integer such that $\height\left(x\right)>\height\left(y\right)$
for all $y$ such that $\tuple{x,y}\in\ra$. 
\end{defi}
For example, a leaf node has a height of 0 since 0 is the smallest
non-negative integer.
\begin{defi}
We say that $v$ is \uline{reachable} from $w$, with respect to
an accessibility relation $\access$, iff there is a path through
$R$ from $w$ to $v$.
\end{defi}

\begin{defi}
We say that a binary relation $R'$ is the \uline{fragment} of
another binary relation $R$ on a set $X$ iff 
\begin{align*}
\Forall{x,y}xR'y\iff x,y\in X\wedge xRy & \mathfullstop
\end{align*}

We say that a function $g'$ is the fragment of another function $g$
on a set $X$ iff $\func{range}\left(g\right)=X\subseteq\func{range}\left(g'\right)$
and $g\left(x\right)=g'\left(x\right)$ for all $x$ in $X$.
\end{defi}

\begin{defi}
We say $C=\left\langle \allworlds_{C},\ra_{C},\valuation_{C}\right\rangle $
is a \uline{subtree} of $T=\ctlstruct$ iff there exists $w\in\allworlds$
such that $\allworlds_{C}$ is the subset of $\allworlds$ reachable
from $w$ and $\ra_{C}$ and $g_{C}$ are the fragments of $\ra$
and $\valuation$ on $\allworlds_{C}$ respectively. We say $C$ is
a \uline{direct subtree} of $T=\treestruct$ if $C$ is a subtree
of $T$ and $\tuple{\textbf{root}\left(T\right),\textbf{root}\left(C\right)}\in\ra$.
\end{defi}

\subsection{\label{sub:Automata}Automata}

In this section we will define some basic terms and properties of
automata that will be used later in this paper. We focus on showing
that we can translate between counter-free automata and LTL \formulae{}.
\begin{defi}
A \uline{Finite State Automaton (FSA)} $\autom=(\alphabet,S,\initial,\transition,F)$
contains
\end{defi}
$\alphabet$: set of symbols (alphabet)

$S$: finite set of automaton states

$\initial$: set of initial states $\subseteq S$

$\transition$ : a transition relation $\subseteq$ $(S\times\alphabet\times S)$

$F$: A set of accepting states $\subseteq S$

We call the members of $\alphabet^{*}$ words. Each transition of
a path through an automaton is labelled with an element $e$ of $\Sigma$.
We say $s_{0}\overset{e_{0}}{\rightarrow}s_{1}\overset{e_{1}}{\rightarrow}\cdots\overset{e_{n-1}}{\rightarrow}s_{n}$
is a path of $\autom$ iff for all $0\leq i<n$ the tuple $\left\langle s_{i},e_{i},s_{i+1}\right\rangle $
is in $\transition$. The label of the path is the word $\left\langle e_{0},e_{1},\ldots,e_{n}\right\rangle $.
We say that a path through an automaton is a run iff $s_{0}\in Q_{0}$.
 A run of an FSA is called accepting if it ends in an accepting state.
We define the language $\lang(\autom)$ recognised by an automaton
to be the set of all words for which there is an accepting run.

\begin{defi}
\textup{We let $L_{p,q}\left(\autom\right)$ be the set of all labels
of paths through $\autom$ from $p$ to $q$.}
\end{defi}
Of particular importance to this paper are counter-free automata.
As will be discussed later we can translate LTL formulas to and from
counter-free automata.
\begin{defi}
A \uline{counter-free} automaton is an automaton such that for
all positive integers $m$, states $s\in S$ and words $u$ in $\alphabet^{*}$,
if $u^{m}\in L_{s,s}$ then $u\in L_{s,s}$~\citep{DiGa08Thomas}.
\end{defi}

\begin{defi}
\label{def:a-determinisation}We define a Deterministic Finite Automaton
(DFA) to be an FSA $\autom=(\alphabet,S,\initial,\transition,F)$
where $\left|Q_{0}\right|=1$ and for every $s$ in $S$ and $e$
in $\alphabet$ there exists exactly one $t\in S$ such that $\left(s,e,t\right)\in\transition$.
\end{defi}

\global\long\def\detm#1{\hat{#1}}

Having given the obvious definition of DFAs as a special case of FSAs,
we will now define a standard determinisation for FSAs.
\begin{defi}
Given an FSA $\autom=(\alphabet,S,\initial,\transition,F)$, we define
the determinisation of $\autom$ to be the DFA $\detm{\autom}=(\alphabet,\detm S,\left\{ \initial\right\} ,\detm{\transition},\detm F)$
with:\end{defi}
\begin{itemize}
\item $\detm S=2^{S}$. Each $\detm s\in\detm S$  represents the set of
states of $\autom$ that $\autom$ could be in now.
\item For each $\detm s,\detm t\in\detm S$, the tuple $\left\langle \detm s,e,\detm t\right\rangle $
is in $\detm{\transition}$ iff $\detm t$ is the maximal subset of
$S$ satisfying $\Forall{t\in\detm t}\Exists{s\in\detm s}\left\langle s,e,t\right\rangle \in\transition$.
\item $\detm s\in\detm F$ iff there is an $s\in\detm s$ such that $s\in F$. 
\end{itemize}
The reason for presenting the above determinisation is to so that
we can show that we can determinise FSA while preseriving counter-free
automata. While this intuitive, it is important to this paper so we
will provide a formal proof.
\begin{lemma}
\label{lem:preserve-counter-free}If $\autom$ is counter-free then
the determinisation $\detm{\autom}$ produced by the above procedure
is counter-free.\end{lemma}
\begin{proof}
Say that $\detm{\autom}$ is not counter-free. Thus there exists $u$,
$m$ and $\detm s$ such that $u^{m}\in\detm L_{\detm s,\detm s}$
but $u\notin\detm L_{\detm s,\detm s}$. 

Note that we have a cycle such that the word $u$ takes us from $\detm s_{0}=\detm s$
to $\detm s_{1}$, from $\detm s_{1}$ to $\detm s_{2}$ and so on
 back to $\detm s_{0}=\detm s$, or more formally: $u\in\bigcap_{i<m}L_{\detm s_{i},\detm s_{i+1}}$
and $u\in L_{\detm s_{m-1},\detm s_{0}}$. Note also that $\detm s\subseteq\mfw$,
and we see that $u^{m}\in L_{s,s}$ for all $s\in\detm s$. As $\autom$
is counter-free is it also the case that $u\in L_{s,s}$ for all $s\in\detm s$.
As $u\in L_{s,s}$ and $s\in\detm s_{0}$ it follows that $s\in\detm s_{1}$;
we may repeat this argument to show that as $s\in\detm s_{1}$ it
must also be the case that $s\in\detm s_{2}$ and so on. Thus $\detm s_{0}\subseteq\detm s_{1}\subseteq\cdots\subseteq\detm s_{0}$
and so $\detm s_{0}=\detm s_{1}=\cdots=\detm s_{0}$. We see $L_{\detm s_{0},\detm s_{1}}=L_{\detm s,\detm s}$
and since $u\in L_{\detm s_{0},\detm s_{1}}$ it follows that $u\in L_{\detm s,\detm s}$,
but we have assumed that $u\notin\detm L_{\detm s,\detm s}$. Hence
by contradiction, $\detm{\autom}$ is counter-free.
\end{proof}
We will use the fact that the determinisation is counter-free to generalise
the following theorem to non-deterministic automata.
\begin{theorem}
\label{thm:Translating-a-counter-free}Translating a counter-free
DFA into an LTL formula results in a formula of length at most $m2^{2^{\bigO\left(n\ln n\right)}}$
where $m$ is the size of the alphabet and $n$ is the number of states
\citep{Wilke99classifyingdiscrete}.
\end{theorem}
One minor note is that \citep{Wilke99classifyingdiscrete} uses stutter-free
operators so their $\left(\alpha U\beta\right)$ is equivalent to
our $N\left(\alpha U\beta\right)$; however, this is trivial to translate.

As the determinisation from \prettyref{def:a-determinisation} has
$2^{n}$ states where $n$ is the number of states in the original
FSA, \prettyref{cor:counter-free-FSA-to-LTL-3EXP} below follows from
\prettyref{lem:preserve-counter-free} and \prettyref{thm:Translating-a-counter-free}.
\begin{corollary}
\label{cor:counter-free-FSA-to-LTL-3EXP}Translating a counter-free
FSA into an LTL formula results in a formula of length at most $m2^{2^{\bigO\left(2^{n}n\right)}}$
where $m$ is the size of the alphabet and $n$ is the number of states.
\end{corollary}
We now define shorthand for discussing a variant of an automaton
starting at a different state.
\begin{defi}
\label{def:shorthand-As}Given an automaton $\autom=(\alphabet,S,\initial,\transition,F)$,
we use $\autom^{s}$ as shorthand for $(\alphabet,S,\left\{ s\right\} ,\transition,F)$
where $s\in S$. We say that an automaton $\autom$ accepts a word
\uline{from} state $s$ if the automata $\autom^{s}$ accepts the
word.
\end{defi}

\subsubsection{Automata on Infinite Words}

In \thispaper{} we use automata as an alternate representation of
temporal logic \formulae{}. LTL is interpreted over infinitely long
paths, and so we are particularly interested in automata that are
similarly interpreted over infinitely long runs. We will define an
infinite run now.
\begin{defi}
We call the members of $\alphabet^{\omega}$ infinite words. We say\textup{
$s_{0}\overset{e_{0}}{\rightarrow}s_{1}\overset{e_{1}}{\rightarrow}\cdots$
is an infinite path of $\autom$ iff for all }$i\geq0$ the tuple
$\left\langle s_{i},e_{i},s_{i+1}\right\rangle $ is in $\transition$\textup{.
The label of the path is $\left\langle e_{0},e_{1},\ldots\right\rangle $.
}An infinite run $\rho$ of $\autom$ is a path starting at a state
in $\initial$.  
\end{defi}
There are a number of different types of automata that can be interpreted
over infinite runs. These are largely similar to FSA, but have different
accepting conditions. \index{Buchi@\buchi{} automata}\emph{\buchi{}
automata} are extensions of finite state automata to infinite worlds.
A \buchi{} automaton is similar to an FSA, but we say an infinite
run is accepting iff a state in $F$ occurs infinitely often in the
run.

\begin{defi}
\label{def:path-to-word}For a fixed structure $M$, a fullpath $\sigma$
through $M$, and a set of state \formulae{} $\Phi$ we let $g_{\Phi}\left(\sigma_{\leq n}\right)=\left(w_{0},w_{1},\ldots,w_{n}\right)$
and $g_{\Phi}\left(\sigma_{\leq n}\right)=\left(w_{0},w_{1},\ldots\right)$
where $w_{i}=\left\{ \phi\colon\:\phi\in\Phi\wedge M,\sigma_{i}\forces\phi\right\} $
for each non-negative integer $i$.
\end{defi}
We are interested in counter-free automata because it is known that
a language $L$ is definable in \logicindex{LTL}LTL iff $L$ is accepted
by some counter-free \buchi{} automaton~\citep{DiGa08Thomas} (see
\prettyref{thm:LTL=00003Dcounter-free}).

It is well known that we can represent a CTL{*} formula as an LTL
formula over a path, where that path includes state formula as atoms; this is commonly used in model checking~\citep{ModelChecking,318620,clarke1986automatic}.
Recall that \prettyref{thm:LTL=00003Dcounter-free} states that a
language $L$ is definable in LTL iff $L$ is accepted by some counter-free
\buchi{} automaton~\citep{DiGa08Thomas}; thus we can also express
this LTL formula as a counter-free \buchi{} automaton. 

Formally, for any CTL{*} formula $\phi$ there exists a set of state
\formulae{} $\Phi$ and a counter-free automaton $\autom=(2^{\Phi},\States,\initial,\transition,F)$
such that $\autom$ accepts $g_{\Phi}\left(\sigma\right)$ iff $M,\sigma\forces\phi$. 
\begin{defi}
\label{def:equiv-autom2CTL}We say an automaton \textup{$\autom=(2^{\Phi},\States,\initial,\transition,F)$
is equivalent to a formula $\phi$ iff for all structures $M$ and
fullpaths $\sigma$ through $M$ we have: 
\begin{eqnarray*}
\left(\Forall{M,\sigma}M,\sigma\forces\phi\right)\iff\left(\autom\mbox{ accepts }g_{\Phi}\left(\sigma\right)\right)\mathfullstop
\end{eqnarray*}
}
\end{defi}

\subsubsection{Alternating Tree Automata}

 Our succinctness proof in \prettyref{sec:Succinctness} uses results
that show CTL{*} can be translated to tree automata. 

We will define a type of tree automata called \emph{symmetric alternating
automata} (SAA) (see for example~\citep{kupferman00automatatheoretic}),
these are a subclass of alternating automata, and can also be referred
to as just alternating automata (see for example~\citep{DBLP:journals/tcs/Dam94}).

Every node, in the run of an SAA on an input structure $M$, represents
a world of $M$. However, a world $w$ in the input structure $M$
may occur many times in a run. Where a non-deterministic automata
would non-deterministically pick a next state, an SAA non-deterministically
picks a conjunction of elements of the form $\left(\square,q\right)$
and $\left(\lozenge,q\right)$; alternatively we may define SAA as
deterministically picking a Boolean combination of requirements of
this form, see for example~\citep{kupferman00automatatheoretic}.
Alternating automata can also be thought of as a type of parity game,
see for example~\citep{GrThWi02}. An element of the form $\left(\square,q\right)$/$\left(\lozenge,q\right)$
indicates for every/some child $u$ of the current world $w$ of the
input structure $M$, a run on $M$ must have a branch which follows
$u$ and where $q$ is the next state. Before defining SAA we will
first define parity acceptance conditions.

\global\long\def\run{L}

\begin{defi}
A \uline{parity acceptance condition} $F$ of an automaton $(\alphabet,S,\initial,\transition,F)$
is a map from $S$ to $\natnum$. We say that a path satisfies the
parity condition $F$ iff the largest integer $n$, such that $F\left(q\right)=n$
for some $q$ that occurs infinitely often on the path, is even.
\end{defi}
We can now define SAA\@.
\begin{defi}
A \emph{symmetric alternating automata} (\uline{SAA}) is a tuple
\begin{align*}
(\alphabet,S,\initial,\transition,F)
\end{align*}
 where $\Sigma,S$ and $S_{0}$ are defined as in \buchi{} automata,
and 

$\transition$ : a transition function $\subseteq(S\times\alphabet\times2^{\left\{ \square,\lozenge\right\} \times S})$
\end{defi}
We define the acceptance condition $F$ of an SAA to be a parity acceptance
condition, but note that we can express \buchi{} parity conditions
as parity acceptance conditions. The SAA accepts a run iff every infinite
path through the run satisfies $F$.

A run $\run=\left\langle \allworlds_{\run},\ra_{\run},\valuation_{\run}\right\rangle $
of the SAA on a $\Var$-labelled pointed value structure $\pvs$ is
an $\allworlds\times S$-labelled tree structure satisfying the following
properties. Where $g_{\run}\left(\func{root}\left(\run\right)\right)=\tuple{w,q}$,
it is the case that $q\in S_{0}$ and $w=a$. For every $w_{\run}$
in $\allworlds_{\run}$, where $\tuple{w,q}=g_{\run}\left(w_{\run}\right)$
and $e=g\left(w\right)$, there exists some set $X\in2^{\left\{ \square,\lozenge\right\} \times S}$
such that $\tuple{q,e,X}\in\transition$ and
\begin{enumerate}
\item For each $\stateb\in S$ such that $\left(\square,r\right)\in X$,
for \emph{each} $u$ such that $w\ra u$ there must exist $u_{\run}$
such that $w_{\run}\ra_{L}u_{\run}$ and $\left(u,r\right)\in\valuation_{\run}\left(u_{\run}\right)$. 
\item For each $\stateb\in S$ such that $\left(\lozenge,r\right)\in X$,
for \emph{some} $u$ such that $w\ra u$ there must exist $u_{\run}$
such that $w_{\run}\ra_{L}u_{\run}$ and $\left(u,r\right)\in\valuation_{\run}\left(u_{\run}\right)$.\end{enumerate}

\begin{theorem}
\label{thm:CTL2AA-single-exponential}Given a CTL{*} formula $\psi$
we can construct an SAA $\autom_{\psi}$ with a number of states that
is singly exponential in the length of $\psi$. \end{theorem}
\begin{proof}
\citet{DBLP:journals/tcs/Dam94} provides a translation of CTL{*}
\formulae{} into equivalent $\mu$-calculus. The nodes are sets of
\formulae{}, so this is a singly exponential translation.

There are a number of translations of $\mu$-calculus into alternating
automata. Wilke gives a simple translation that does not assume that
the tree has any particular structure~\citep{Wilke00alternatingtree}.
The states in the resulting automata are sub\formulae{} of the $\mu$-calculus
formula. Hence the translation into alternating automata is linear.
\end{proof}
The translation via $\mu$-calculus above is sufficient for \thispaper{}.
There are translations that result in more optimised model checking
and decision procedure results~\citep{kupferman00automatatheoretic}.

\subsection{\label{sub:Bisimulations}Bisimulations}

An important concept relating to structures is bisimilarity, as two
bisimilar structures satisfy the same set of modal \formulae{}. 
We credit \citet{milner1980calculus} and \citet{ParkDavid1981} for
developing the concept of bisimulation.
\begin{defi}
\label{def:PVS}Where  $M=\ctlstruct$ is a structure and $a\in\allworlds$,
we say that $M_{a}$ is a Pointed Valued Structure (\uline{PVS}).

\end{defi}
We now provide a definition of a bisimulation.
\begin{defi}
\label{def:bisimulation}Given a PVS $\pvs$ and a PVS $\pvsB$ we
say that a relation $\bisimulation$ from $S$ to $\hat{\allworlds}$
is a \uline{bisimulation} from $\pvs$ to $\pvsB$ iff\end{defi}
\begin{enumerate}
\item $\left(\pvsW,\pvsWB\right)\in\bisimulation$ 
\item for all $\left(u,\hat{u}\right)\in\bisimulation$ we have $g\left(u\right)=\hat{g}\left(\hat{u}\right)$.
\item for all $\left(u,\hat{u}\right)\in\bisimulation$ and $v\in uR$ there
is some $\hat{v}\in\hat{u}\hat{R}$ such that $\left(v,\hat{v}\right)\in\bisimulation$.
\item for all $\left(u,\hat{u}\right)\in\bisimulation$ and $\hat{v}\in\hat{u}\hat{R}$
there is some $v\in uR$ such that $\left(v,\hat{v}\right)\in\bisimulation$.
\end{enumerate}
Bisimulations can be used to define bisimilarity.
\begin{defi}
\label{def:bisimilar}We say that $\pvs$ and $\pvsB$ are \uline{bisimilar}
iff there exists a bisimulation from $\pvs$ to $\pvsB$.
\end{defi}

\begin{defi}
\label{def:bisimulation-invariant}We say that a formula $\phi$ of
some logic L is \uline{bisimulation invariant} iff for every bisimilar
pair of PVS's $\left(M,w\right)$ and $(\hat{M},\hat{w})$ where $M$
and $\hat{M}$ are structures that L is interpreted over, we have
$M,w\forces\phi$ iff $\hat{M},\hat{w}\forces\phi$. We say the logic
L is bisimulation invariant iff every formula $\phi$ of L is bisimulation
invariant.
\end{defi}
Knowing that a logic is bisimulation invariant is useful because we
can take the tree-unwinding of a structure without changing the set
of \formulae{} that it satisfies.

\input{Thesis_Background_Equiv__.tex}

\global\long\def\funcC{\beta}

\section{Bisimulation Invariance}

\label{sec:bisim}\label{sub:Bisimulation-invariance-of-RoCTL*}

\renewcommand{\CTLv}{CTL}\renewcommand\RoCTLvstruct{RoCTL-structure}\renewcommand\RoCTLvstarLogic{RoCTL*} 

Recall that bisimulation invariance was defined in \prettyref{def:bisimulation-invariant}.
We shall now begin to prove some basic lemmas necessary to show that
RoCTL{*} is bisimulation invariant. First we will prove that \RoCTLvstarLogic{}
is bisimulation invariant, and define bisimulations on RoCTL-structures.
Before reading the following definition recall the definition of a
PVS, or pointed valued structure, from \prettyref{def:PVS}.
\begin{defi}
Let $\bisimulation$ be any bisimulation from some PVS $M_{w}$ to
another PVS $\hat{M}_{\hat{w}}$. We define $\bisimulation^{\omega}$
to be a binary relation from fullpaths through $M$ to fullpaths though
$\hat{M}$ such that $\left(\sigma,\hat{\sigma}\right)\in\bisimulation^{\omega}$
iff $\left(\sigma_{i},\hat{\sigma}_{i}\right)\in\bisimulation$ for
all $i\in\natnum$. We say that a PVS $M_{w}$ is a Ro\CTLv{}-model
iff $M$ is a \RoCTLvstruct{}.
\end{defi}
It is important that for a path $\sigma$ though $M$ we can find
a similar path $\hat{\sigma}$ through $\hat{M}$. We will now show
that this is the case.
\begin{lemma}
\label{lem:bi-construct-path}Let $\bisimulation$ be any bisimulation
from some Ro\CTLv{}-model $M_{w}$ to another Ro\CTLv{}-model $\hat{M}_{\hat{w}}$.
For any fullpath $\sigma$ where $\sigma_{0}=w$ through $M$ there
exists a fullpath $\hat{\sigma}$ through $\hat{M}$ such that $\left(\sigma,\hat{\sigma}\right)\in\bisimulation^{\omega}$
and $\hat{\sigma}_{0}=\hat{w}$; likewise for any fullpath $\hat{\sigma}$
where $\hat{\sigma}_{0}=\hat{w}$ through $\hat{M}$ there exists
a fullpath $\sigma$ through $M$ such that $\left(\sigma,\hat{\sigma}\right)\in\bisimulation^{\omega}$
and $\sigma_{0}=w$.\end{lemma}
\begin{proof}
We construct $\hat{\sigma}$ from $\sigma$ as follows: let $\hat{\sigma}_{0}=\hat{w}$.
Once we have chosen $\hat{\sigma}_{i}$ we choose $\hat{\sigma}_{i+1}$
as follows: since $\left(\sigma_{i},\hat{\sigma}_{i}\right)\in\bisimulation$
and $\sigma_{i+1}\in\sigma_{i}R$ there is some $\hat{v}\in\hat{\sigma}_{i}\hat{R}$
such that $\left(\sigma_{i+1},\hat{v}\right)\in\bisimulation$; we
let $\hat{\sigma}_{i+1}=\hat{v}$. We may construct $\sigma$ from
$\hat{\sigma}$ likewise.
\end{proof}
The following lemma is similar; however, we are specifically attempting
to construct deviations.
\begin{lemma}
\label{lem:bi-construct-deviation}Let $\bisimulation$ be a bisimulation
from some Ro\CTLv{}-model $M_{w}$ to another Ro\CTLv{}-model $\hat{M}_{\hat{w}}$.
Let $\left(\sigma,\hat{\sigma}\right)\in\bisimulation^{\omega}$.
Given a deviation $\hat{\pi}$ from $\hat{\sigma}$ we can construct
a fullpath $\pi$ such that $\pi$ is a deviation from $\sigma$ and
$\left(\pi,\hat{\pi}\right)\in\bisimulation^{\omega}$.\end{lemma}
\begin{proof}
As $\hat{\pi}$ is a deviation from $\hat{\sigma}$, it is the case
that $\hat{\pi}$ is an $i$-deviation from $\hat{\sigma}$ for some
non-negative integer $i$. Since $\left(\sigma_{i},\hat{\pi}_{i}\right)\in\bisimulation$
we can construct a fullpath $\tau$ such that $\left(\tau,\hat{\pi}_{\geq i}\right)\in\bisimulation^{\omega}$
and $\tau_{0}=\sigma_{i}$. We see that $\sigma_{\leq i-1}\cdot\tau$
is a fullpath through $M$, we call this fullpath $\pi$. Since $\hat{\pi}_{\geq i+1}$
is failure-free $\tau_{\geq1}$ is failure-free and thus $\pi_{\geq i+1}$
is failure-free. Thus $\pi$ is a deviation from $\sigma$.
\end{proof}
We will now state and prove the truth lemma.
\begin{lemma}
\label{lem:bi-path-invariant}Let $M_{w}$ and $\hat{M}_{\hat{w}}$
be a pair of arbitrary Ro\CTLv{}-models and let $\bisimulation$
be a bisimulation from $M_{w}$ to $\hat{M}_{\hat{w}}$. Then for
any $\left(\sigma,\hat{\sigma}\right)\in\bisimulation^{\omega}$,
and for any formula $\phi$ it is the case that $M,\sigma\forces\phi\iff\hat{M},\hat{\sigma}\forces\phi$. \end{lemma}
\begin{proof}
For contradiction, let $\phi$ be the shortest formula such that there
exists a pair $\sigma$, $\hat{\sigma}$ of fullpaths in $\bisimulation^{\omega}$
not satisfying $M,\sigma\forces\phi\iff\hat{M},\hat{\sigma}\forces\phi$.
Without loss of generality we can assume that $M,\sigma\forces\phi$
but $\hat{M},\hat{\sigma}\nforces\phi$. Consider the possible forms
of $\phi$.

$\phi=p$: Since $M,\sigma\forces p$ we know that $p\in g\left(\sigma_{0}\right)$.
As $\bisimulation$ is a bisimulation and $\left(\sigma_{0},\hat{\sigma}_{0}\right)\in\bisimulation$
we know that $p\in g\left(\sigma_{0}\right)$. Hence $\hat{M},\hat{\sigma}\forces p$.
 This contradicts our assumption that $\hat{M},\hat{\sigma}\nforces\phi$. 

$\phi=N\psi$: Since $M,\sigma\forces N\psi$, we know that $M,\sigma_{\geq1}\forces\psi$,
and since $\phi$ is the shortest counter example, we know that $\hat{M},\hat{\sigma}_{\geq1}\forces\psi$.
We see that $\hat{M},\hat{\sigma}\forces\phi$.

$\phi=\formc U\psi$: Since $M,\sigma\forces\formc U\psi$, we know
that a non-negative integer $i$ such that $\mfk,\patha_{\geq i}\forces\formb$
and for all non-negative $j$ less than $i$ we have $\mfk,\patha_{\geq j}\forces\theta$.
As $\psi$ and $\formc$ are shorter than $\phi$ we know $\hat{\mfk},\hat{\patha}_{\geq i}\forces\formb$
and $\hat{\mfk},\hat{\patha}_{\geq j}\forces\theta$. Thus $\hat{\mfk},\hat{\patha}\forces\theta U\psi$.

$\phi=A\psi$: Since $\hat{M},\hat{\sigma}\nforces A\psi$ we know
there exists $\hat{\pi}$ such that $\hat{M},\hat{\pi}\nforces\psi$.
From \prettyref{lem:bi-construct-path} we know that there exists
a path $\pi$ such that $\left(\pi,\hat{\pi}\right)\in\bisimulation$.
As $\psi$ is shorter than $\phi$ we know that $M,\pi\nforces\psi$.
Hence $M,\sigma\nforces A\psi$.

$\phi=O\psi$: Since $\hat{M},\hat{\sigma}\nforces O\psi$ we know
there exists $\hat{\pi}$ such that $\hat{M},\hat{\pi}\nforces\psi$
and $\hat{\pi}$ is failure-free. From \prettyref{lem:bi-construct-path}
we know that there exists a path $\pi$ such that $\left(\pi,\hat{\pi}\right)\in\bisimulation$.
As $\psi$ is shorter than $\phi$ we know that $M,\pi\nforces\psi$.
As $\hat{\pi}$ is failure-free, for all $i>0$ we know $\Viol\notin g\left(\hat{\pi}_{i}\right)$,
from the definition of a bisimulation we know that $\Viol\notin g\left(\pi_{i}\right)$.
Hence $\pi$ is failure-free and $M,\sigma\nforces O\psi$.

$\phi=\eA\psi$: Since $\hat{M},\hat{\sigma}\nforces\eA\psi$ we know
there exists $\hat{\pi}$ such that $\hat{M},\hat{\pi}\nforces\psi$
and $\hat{\pi}$ is either $\hat{\sigma}$ or a deviation from $\hat{\sigma}$.
If $\hat{\pi}=\hat{\sigma}$ then $M,\sigma\nforces\psi$ and $M,\sigma\nforces\eA\psi$.
If $\hat{\pi}$ is an $i$-deviation from $\hat{\sigma}$ then from
\prettyref{lem:bi-construct-deviation} we know there is a deviation
$\pi$ from $\sigma$ such that $(\pi,\hat{\pi})\in\bisimulation^{\omega}$.
We see that $M,\pi\nforces\psi$ and thus $M,\sigma\nforces\robust\psi$.

By contradiction we know that no such $\phi$ exists. 
\end{proof}

\begin{lemma}
\label{lem:RoCTL*v-is-bisimulation-invariant}\label{lem:RoCTL*-is-bisimulation-invariant}\RoCTLvstarLogic{}
is \index{bisimulation invariant}bisimulation invariant. \end{lemma}
\begin{proof}
Consider any RoCTL{*} formula $\phi$. Let $\bisimulation$ be a bisimulation
from some pair of PVS's $\left(M,w\right)$ and $(\hat{M},\hat{w})$,
and say that $M,w\forces\phi$ but $\hat{M},\hat{w}\nforces\phi$.
Recall that under RoCTL{*} we define truth at a world as follows:
\begin{align*}
\mfk,w\forces\forma & \tiff\Exists{\pathb\st\pi_{0}=w}\mfk,\pathb\forces\forma\mathfullstop
\end{align*}
From \prettyref{lem:bi-path-invariant} we know that there exists
a fullpath $\hat{\pi}$ through $\hat{M}$ such that $\hat{\pi}=\hat{w}$
and $\hat{M},\hat{\pi}\forces\phi$. Hence $\hat{M},\hat{w}\forces\phi$.
Thus we see that for any bisimilar pair of PVS's $\left(M,w\right)$
and $(\hat{M},\hat{w})$ we have 
\begin{eqnarray*}
\left(M,w\right)\forces\phi & \iff & (\hat{M},\hat{w})\forces\phi\mathfullstop
\end{eqnarray*}
By definition we see that $\phi$ is bisimulation invariant. Since
$\phi$ is an arbitrary RoCTL{*} formula, we see that RoCTL{*} is
bisimulation invariant.
\end{proof}

\section{Reduction into QCTL{*}}

\label{sec:qctl}\label{sec:A-Linear-reduction-into-QCTL*}\input{Body_ALinearReductionIntoQCTLstar__.tex}

\subsection{\label{sec:Hybrid-Logic}A Comment on Hybrid Logic}

Even the tree-semantics of QCTL{*} is non-elementary to decide and
no translation into CTL{*} is elementary in length. For this reason
we also investigated other logics to translate RoCTL{*} into. We know
that we can represent RoCTL{*} with a single variable fragment of
a hybrid extension of CTL{*}, by translating $\prone\phi$ into a
formula such as the following: 
\begin{align*}
\phi\vee\exists x.\left(Fx\wedge E\left(\phi\wedge F\left(x\wedge NNG\neg\Viol\right)\right)\right) & \mathcomma
\end{align*}
where $\exists x.\psi$ is the hybrid logic formula indicating that
the exists a valuation of $x$ such that $x$ is true at exactly one
node on the tree and $\psi$ is true. This is still a way away from
producing a decision procedure for RoCTL{*}. There has been considerable
research into single variable fragments of Hybrid Logic recently (see
for example \citep{DBLP:conf/mfcs/KaraWLS09} for a good overview
of the results in this area). However, these fragments do not contain
the $\exists$ operator as a base operator. Although $\exists x.\psi$
can be defined as an abbreviation, this requires two variables. Even
adding a single variable hybrid logic to CTL{*} leads to a non-elementary
decision procedure (see for example \citep{DBLP:conf/mfcs/KaraWLS09}),
and adding two variables to an otherwise quite weak temporal logic
again gives a non-elementary decision procedure \citep{SW07}.  A
potential avenue of research is investigating the complexity of deciding
the fragment of Hybrid CTL{*} (HCTL{*}) where the hybrid component
consists solely of the $\exists$ operator over a single variable,
as the translation of RoCTL{*} into HCTL{*} falls inside this fragment.
Although we have also given a linear translation into the tree-semantics
of QCTL{*} logic, this single variable fragment of HCTL{*} seems much
more restricted than QCTL{*}. Additionally this fragment of HCTL{*}
seems to have a closer relationship with pebble automata. Never-the-less
this avenue does not seem likely to result in an elementary decision
procedure for RoCTL{*}.

\section{$\ATL$}

\label{sec:ALTL}\label{sec:altl}

Here we define a possible extension of LTL allowing automata to be
used as operators, and briefly show to convert an $\ATL$ formula
$\phi$ into an automaton $\autom_{\phi}$. 
\begin{defi}
We define $\ATL$ \formulae{} recursively according to the following
BNF notation, 
\begin{align*}
\forma & ::=p\BNFbar\neg\forma\BNFbar\lAnd{\forma}{\forma}\BNFbar\lUntil{\forma}{\phi}\BNFbar N\phi\BNFbar\autom\mathcomma
\end{align*}
where $p$ varies over $\Var$ and $\autom$ can be any counter-free
FSA that accepts $2^{\Var}$ as input, that is $\alphabet=2^{\Var}$.
Recall from \prettyref{def:path-to-word} that $g_{\Phi}$ is a simple
conversion from fullpaths to words of an automaton. In this section
we will assume that the special atoms required for the translation
are members of $\Var$, and so will use $\Var$ as $\Phi$. The semantics
of $\ATL$  are defined similarly to LTL, with the addition that
$M,\sigma\forces\autom$ iff the automata $\autom$ accepts $g_{\Var}\left(\sigma\right)$,
or in other words. 
\begin{eqnarray*}
\mfk,\patha\forces\autom & \tiff & \exists_{i}\st g_{\Var}\left(\sigma_{\leq i}\right)\in\lang\left(\autom\right)
\end{eqnarray*}

\end{defi}
Note that since automata can be $\ATL$ \formulae{}, the following
definition also gives us a definition of equivalence between \formulae{}
and automata.
\begin{defi}
We say that a pair of \formulae{} $\phi$, $\psi$ are equivalent
($\phi\equiv\psi$) iff for all structures $M$ and paths $\sigma$
through $M$: 
\begin{align*}
M,\sigma\forces\phi\iff & M,\sigma\forces\psi\mathfullstop
\end{align*}

\end{defi}
We will now give a partial translation from $\ATL$ \formulae{} into
automata; we will not define the acceptance condition $F$ since $F$
is discarded when we produce $\autom_{\devi\phi}$ from $\autom_{\phi}$.
\begin{defi}
\label{def:length-ATL}We define the length of an $\ATL$ formula
recursively as follows:
\begin{eqnarray*}
\left|\phi\wedge\psi\right|=\left|\phi U\psi\right| & = & \left|\phi\right|+\left|\psi\right|\\
\left|\neg\phi\right|=\left|N\phi\right| & = & \left|\phi\right|+1\\
\left|p\right| & = & 1\\
\left|(\alphabet,S,\initial,\transition,F)\right| & = & \left|S\right|
\end{eqnarray*}

In some translations we encode state-\formulae{} (e.g. $A\psi$)
into atoms (labelled $p_{A\psi}$). We define the complexity $\left|\phi\right|^{\star}$
of an $\ATL$ formula $\phi$ similarly, except that we define the
complexity $\left|p_{\psi}\right|^{\star}$ of an atom labelled $p_{\psi}$
to be $\left|\psi\right|^{\star}$.
\end{defi}

\begin{lemma}
\label{lem:ALTL-decidable}The satisfiability problem for $\ATL$
is decidable.\end{lemma}
\begin{proof}
From \prettyref{cor:counter-free-FSA-to-LTL-3EXP} we can replace
each automata in a $\ATL$ formula $\phi$ with an equivalent LTL
formula. This will result in an LTL formula $\phi$' equivalent to
$\phi$. We can then use any decision procedure for LTL to decide
$\phi$.
\end{proof}

\subsection{A partial translation from $\ATL$ into automata}

\label{sub:ALTL-to-automata}Here we define a translation of an $\ATL$
formula $\phi$ into an automaton $\autom_{\phi}$. However, we do
not define a traditional acceptance condition as this is not required
when constructing $\autom_{\devi\phi}$ from $\autom_{\phi}$. In
this section we will use $\stateA$ and $\stateT$ to represent arbitrary
states of $\autom_{\phi}$; we use $\stateX$ and $\stateY$ to represent
arbitrary states of automata in $\phi$. 
\begin{defi}
The closure $\cl\phi$ of the formula $\phi$ is defined as the smallest
set that satisfies the four following requirements:\end{defi}
\begin{enumerate}
\item $\phi\in\cl\phi$
\item For all $\psi\in\cl\phi$, if $\delta\leq\psi$ then $\delta\in\cl\phi$.
\item For all $\psi\in\cl\phi$, $\neg\psi\in\cl\phi$ or there exists $\delta$
such that $\psi=\neg\delta$ and $\delta\in\cl\phi$.
\item If $\autom\in\cl\phi$ then $\autom^{\stateX}\in\cl\phi$ for all
states $\stateX$ of $\autom$.
\end{enumerate}
The states of $\autom_{\phi}$ are sets of \formulae{} that could
hold along a single fullpath. 
\begin{proposition}
\label{prop:The-size-of}The size of the closure set is linear in
$\left|\phi\right|$. 
\end{proposition}
\begin{defi}
\label{defn:MPC-Tableau-3}[MPC] We say that $\stateA\subseteq\closure$
is Maximally Propositionally Consistent (\uline{MPC)} iff for all
$\alpha,\beta\in\stateA$\end{defi}
\begin{description}
\item [{(M1)}] if $\beta=\neg\alpha$ then $\beta\in a$ iff $\alpha\notin\stateA$,
\item [{(M2)}] if $\alpha\wedge\beta\in\closure$ then $\left(\alpha\wedge\beta\right)\in\stateA\leftrightarrow\left(\alpha\in\stateA\tand\beta\in\stateA\right)$ \end{description}

\begin{defi}
\label{def:set-of-states-Sphi} The set of states $S_{\phi}$ is
the set of all subsets $\stateA\subseteq\closure$ satisfying: \end{defi}
\begin{description}
\item [{(S1)}] $\stateA$ is MPC\@
\item [{(S2)}] if $\alpha U\beta\in\stateA$ then $\alpha\in\stateA$ or
$\beta\in\stateA$
\item [{(S3)}] if $\neg\left(\alpha U\beta\right)\in\stateA$ then $\beta\notin\stateA$
\item [{(S4)}] $\stateA$ is non-contradictory, i.e. $\bigwedge\stateA$
is satisfiable.
\end{description}
Note that $\ATL$ is decidable, so we can compute whether $\stateA$
is contradictory. We now define a standard temporal successor relation
for LTL formula.

\global\long\def\hues{H_{\phi}}
\global\long\def\tmpsucc{r_{N}}
\global\long\def\rA{r_{A}}
\global\long\def\rX{\tmpsucc}
\global\long\def\myPhi{\Var}
\global\long\def\myatom{p}

\begin{defi}
\label{defn:rX-Tableau}[$\rX$] The temporal successor $\tmpsucc$
relation on states is defined as follows: for all states $\stateA$,
$\stateB$ put $\left(\stateA,\stateB\right)$ in $\rX$ iff the following
conditions are satisfied:\end{defi}
\begin{description}
\item [{(R1)}] $N\alpha\in\stateA$ implies $\alpha\in\stateB$
\item [{(R2)}] $\neg N\alpha\in\stateA$ implies $\alpha\notin\stateB$
\item [{(R3)}] $\alpha U\beta\in\stateA$ and $\beta\notin\stateA$ implies
$\alpha U\beta\in\stateB$
\item [{(R4)}] $\neg(\alpha U\beta)\in\stateA$ and $\alpha\in\stateA$
implies $\neg(\alpha U\beta)\in\stateB$\end{description}
\begin{defi}
We define the transition relation $\transition_{\phi}\subseteq S_{\phi}\times\alphabet\times S_{\phi}$
as follows: a member $\left\langle \stateS,e,\stateT\right\rangle $
of $S_{\phi}\times\alphabet\times S_{\phi}$ is a member of $\transition_{\phi}$
iff \end{defi}
\begin{description}
\item [{(T1)}] $\left\langle \stateS,\stateT\right\rangle \in\rX$
\item [{(T2)}] For each $\myatom\in\Var$, it is the case that $\myatom\in e$
iff $\myatom\in\stateS$ 
\item [{(T3)}] If $\autom^{\stateX}\in\stateS$, and $\stateX$ is not
an accepting state of $\autom^{\stateX}$, then there must exist a
state $\stateY$ of $\autom^{\stateX}$ such that $\autom^{\stateY}\in\stateT$
and $\left\langle \stateX,e,\stateY\right\rangle $ is a transition
of $\autom^{\stateX}$.
\item [{(T4)}] If $\neg\autom^{\stateX}\in\stateS$, then for each state
$\stateY$ of $\autom^{\stateX}$ such that $\left\langle \stateX,e,\stateY\right\rangle $
is a transition of $\autom^{\stateX}$ it must be the case that $\autom^{\stateY}\notin\stateT$.
\end{description}

\begin{defi}
The automata $\autom_{\phi}$ is the tuple $(\alphabet,S_{\phi},Q_{0},\transition_{\phi})$,
where $Q_{0}$ is the set $\left\{ a\colon a\in S_{\phi}\wedge\phi\in a\right\} $.
\end{defi}
Note that the tuple above does not include an acceptance condition.
The automata $\autom_{\phi}$ is used only to generate the automata
$\autom_{\devi\phi}$. The automata $\autom_{\devi\phi}$ reads the
finite prefix $\sigma_{\leq i}=\pi_{\leq i}$ of an $i$-deviation
$\pi$ from $\sigma$ and then reads a state formula indicating that
we can deviate. This in essence splits the deviation into a prefix
and suffix.  For this reason we do not define a standard acceptance
condition, instead we will say that $\autom_{\phi}$ accepts a pair
$\left(\pi,i\right)$ iff there exists a state $\stateS\in S_{\phi}$
such that the automaton can reach state $\stateS$ after reading the
prefix $\pi_{\leq i-1}$, and $\pi_{\geq i}\forces\bigwedge\stateS$.
Or formally:
\begin{defi}
Given a fullpath $\pi$ though some structure $M$, and non-negative
integer $i$, we say that $\autom_{\phi}$ accepts a pair $\left(\pi,i\right)$
iff there exists a state $\stateS\in S_{\phi}$ such that $\pi_{\geq i}\forces\bigwedge\stateS$,
and there exists a path of $\autom_{\phi}$ labelled $g_{\Var}\left(\pi_{\leq i-1}\right)$
which ends in the state $\stateS$. \end{defi}
\begin{lemma}
\label{lem:s-iff-t}For any fullpath $\pi$, integer $j$, pair of
states $\stateS,\stateT$ such that 
\begin{eqnarray*}
\left\langle \stateS,g_{\Var}\left(\pi_{j}\right),\stateT\right\rangle  & \in & \delta_{\phi}
\end{eqnarray*}
 we have $\pi_{\geq j+1}\forces\bigwedge\stateT\implies\pi_{\geq j}\forces\bigwedge\stateS$.\end{lemma}
\begin{proof}
For contradiction assume that this lemma is false. Then $\pi_{\geq j+1}\forces\bigwedge\stateT$
and $\pi_{\geq j}\nforces\bigwedge\stateS$. Since $\pi_{\geq j}\nforces\bigwedge\stateS$
then there exists some $\psi\in\stateS$ such that $\pi_{\geq j}\nforces\psi$.
We assume without loss of generality that $\psi$ is the shortest
such formula. We now consider each possible form of $\psi$, in each
case recall that $\psi\in\stateS$, $\pi_{\geq j+1}\forces\bigwedge\stateT$
and $\left\langle \stateS,g_{\Var}\left(\pi_{j}\right),\stateT\right\rangle \in\delta_{\phi}$.

$\psi=\neg\neg\alpha$ From M1 and $\psi\in\stateS$ we get $\alpha\in\stateS$
and since $\alpha$ is shorter than $\neg\neg\alpha$ it follows that
$\pi_{\geq j}\forces\alpha$ and so $\pi_{\geq j}\forces\neg\neg\alpha$.
However, by assumption $\pi_{\geq j}\nforces\psi$.

$\psi=p$: From T2 we know that as $p\in\stateS$, we have $p\in g_{\Var}\left(\pi_{j}\right)$
and so $\pi_{\geq j}\forces p$. But by assumption $\pi_{\geq j}\nforces\psi$. 

$\psi=\neg p$: From M1 we know that $p\notin\stateS$, and from T2
we have $p\notin g_{\Var}\left(\pi_{j}\right)$ and so $\pi_{\geq j}\forces\neg p$.

$\psi=\alpha\wedge\beta$: As $\stateS$ is MPC we see that $\alpha,\beta\in\stateS$.
As we have assumed that $\psi$ is the shortest formula that provides
a counterexample we see that $\pi_{\geq j}\forces\alpha$ and $\pi_{\geq j}\forces\beta$.
Hence $\pi_{\geq j}\forces\alpha\wedge\beta$.

$\psi=\neg\left(\alpha\wedge\beta\right)$: As $\stateS$ is MPC we
see that $\alpha\wedge\beta\notin\stateS$. It follows that $\alpha\notin\stateS$
or $\beta\notin\stateS$. Without loss of generality, assume $\alpha\notin\stateS$.
Thus $\pi_{\geq j}\nforces\alpha$ and $\pi_{\geq j}\nforces\left(\alpha\wedge\beta\right)$.
Hence $\pi_{\geq j}\forces\neg\left(\alpha\wedge\beta\right)$.

$\psi=N\theta$: We see that if $\pi_{\geq j}\nforces N\theta$ then
$\pi_{\geq j+1}\nforces\theta$, but we see that from R1 that $\theta\in\stateT$.
By contradiction $\theta$ cannot be of the form $N\theta$.

$\psi=\neg N\theta$: We see that if $\pi_{\geq j}\nforces\neg N\theta$
then $\pi_{\geq j+1}\forces\theta$, but we see that from R2 that
$\theta\notin\stateT$. 

$\psi=\alpha U\beta$: We see that if $\alpha U\beta\in\stateS$ then
from S2 either $\alpha\in\stateS$ or $\beta\in\stateS$. Since $\pi_{\geq j}\nforces\alpha U\beta$
it follows that $\pi_{\geq j}\forces\neg\beta$. As $\neg\beta$ is
shorter than $\psi$ we have $\neg\beta\in\stateS$ and so $\beta\notin\stateS$.
Since $\beta\notin\stateS$, from R3 we have $\alpha U\beta\in\stateT$
and so $\pi_{\geq j+1}\forces\alpha U\beta$. As $\alpha\in\stateS$
and $\alpha$ is shorter than $\psi$ we see that $\pi_{\geq j}\forces\alpha$.
As $\pi_{\geq j}\forces\alpha$ and $\pi_{\geq j+1}\forces\alpha U\beta$
we see that $\pi_{\geq j}\forces\alpha U\beta$.

$\psi=\neg\left(\alpha U\beta\right)$ We see that if $\neg\left(\alpha U\beta\right)\in\stateS$
then from S3 we have $\beta\notin\stateS$ and thus $\neg\beta\in\stateS$.
As $\neg\beta$ is shorter than $\psi$ we have $\pi_{\geq j}\forces\neg\beta$.
Since $\pi_{\leq j}\nforces\neg\left(\mbox{\ensuremath{\alpha}}U\beta\right)$
we have $\pi_{\geq j}\forces\alpha U\beta$; as $\pi_{\geq j}\forces\neg\beta$
it follows that $\pi_{\geq j}\forces\alpha$. Thus $\alpha\in\stateS$,
and from R4 we know $\neg\left(\alpha U\beta\right)\in\stateT$ and
hence $\pi_{\geq j+1}\forces\neg\left(\alpha U\beta\right)$. As $\pi_{\geq j}\forces\neg\beta$
it follows that $\pi_{\geq j}\forces\neg\left(\alpha U\beta\right)$.
By contradiction, $\psi$ cannot be of the form $\neg\left(\alpha U\beta\right)$.

$\psi=\autom^{x}$:  If $x$ is an accepting state of $\autom^{x}$,
then we see that $\autom^{x}$ is satisfied on all fullpaths through
$M$, including $\pi_{\geq j}$ and so $x$ is not an accepting state.
We see from T3 that there exists a state $y$ of $\autom^{x}$ such
that $\autom^{y}\in\stateT$ and $\left\langle x,g_{\Var}\left(\pi_{j}\right),y\right\rangle $
is a transition of $\autom^{x}$. As $\pi_{\geq j+1}\forces\bigwedge\stateT$
we see $\pi_{\geq j+1}\forces\autom^{y}$. We can prepend the state
$x$ and the symbol $g_{\Var}\left(\pi_{j}\right)$ to the accepting
path for $\autom^{y}$ to construct an accepting path for $\autom^{x}$,
so we see that $\pi_{\geq j}\forces\autom^{x}$.

$\psi=\neg\autom^{x}$:  Since $\pi_{\geq j}\nforces\psi$ we see
$\pi_{\geq j}\forces\autom^{x}$. Thus there must exist a state $y$
of $\autom^{x}$ such that $\left\langle x,g_{\Var}\left(\pi_{j}\right),y\right\rangle $
is in the transition relation of $\autom^{x}$ and $\pi_{\geq j+1}\forces\autom^{y}$.
However, from T4 and M1, we see that $\neg\autom^{y}\in\stateT$,
and since $\pi_{\geq j+1}\forces\bigwedge\stateT$, we have $\pi_{\geq j+1}\forces\neg\autom^{y}$\@. 

We have considered all possible forms of $\psi$ and in each case
got a contradiction. By contradiction this lemma must be true.
\end{proof}
We will now state the lemma demonstrating the correctness of the translation.
\begin{lemma}
\label{lem:accepts,pi,i}For any fullpath $\pi$ though $M$, and
non-negative integer $i$, the automata $\autom_{\phi}$ accepts the
pair $\left(\pi,i\right)$ iff $\pi\forces\phi$. 
\end{lemma}

\begin{proof}
We first show that this lemma holds for $i=0$. 

$\left(\Longrightarrow\right)$ We let $\stateS$ be the maximal subset
of $\cl\phi$ such that for each $\psi\in\stateS$ we have $\pi\forces\psi$.
We see that $\stateS$ satisfies S1--4 and so $\stateS\in S_{\phi}$.
We see $\phi\in\stateS$ and so $\stateS\in Q_{0}$. Clearly $\pi\forces\bigwedge\stateS$.

$\left(\Longleftarrow\right)$ By definition each $\stateS\in Q_{0}$
includes $\phi$ and so clearly if $\pi\nforces\phi$ then $\pi\nforces\bigwedge\stateS$. 

Say that the lemma holds for $i=j$, where $j$ is some non-negative
integer. We now show that the lemma holds for $i=j+1$.

$\left(\Longleftarrow\right)$ Say that $\pi\forces\phi$. Since the
lemma holds for $i=j$, we see that there exists a state $\stateS\in S_{\phi}$
such that $\pi_{\geq j}\forces\bigwedge\stateS$ and there exists
a path of $\autom_{\phi}$ labelled $g_{\Var}\left(\pi_{\leq j-1}\right)$
which ends in the state $\stateS$. We let $\stateT$ be the maximal
subset of $\cl\phi$ such that for each $\psi\in\stateT$ we have
$\pi_{\geq j+1}\forces\psi$. Again, we see that $\stateT$ satisfies
S1--4 and so $\stateT\in S_{\phi}$.  We now show that $\left\langle \stateS,g_{\Var}\left(\pi_{j+1}\right),\stateT\right\rangle \in\transition_{\phi}$.

\begin{description}
\item [{T1}] Say that $N\alpha\in\stateS$, then since $\pi_{\geq j+1}\forces\bigwedge\stateS$
it is clear that $\pi_{\geq j+1}\forces\alpha$. Since either $\alpha$
or its negation is in $\stateT$ and $\pi_{\geq j+1}\forces\bigwedge\stateT$,
we see that $\alpha$ in $\stateT$. We see that $\tmpsucc$ is a
standard temporal successor function, and so a similar argument can
be made for R2--4.
\item [{T2}] We see from the semantics that, for each atom $p$, we have
$\pi_{\geq j+1}\forces p$ iff $p\in g_{\Var}\left(\pi_{j+1}\right)$.
Since either $p$ or $\neg p$ in $\stateT$, we again see that $p\in\stateT$
iff $p\in g_{\Var}\left(\pi_{j+1}\right)$. 
\item [{T3}] It is clear that if an automaton accepts a word ``abcd$\cdots$z''
starting at a state $x$ then there must exist state $y$ from which
the automata accepts the word ``bcd$\cdots$z'', and such that $\left\langle \stateX,a,\stateY\right\rangle $
is in the transition relation. Again, $\stateT$ contains either $\autom^{y}$
or $\neg\autom^{y}$. Since $\pi_{\geq j+1}\forces\bigwedge\stateT$
and $\pi_{\geq j+1}\forces\autom^{y}$ it is clear that $\autom^{y}\in\stateT$.
\item [{T4}] This is the converse of T3. We see that if there exists a
state $y$ from which the automata accepts the word ``bcd$\cdots$z'',
and $\left\langle \stateX,a,\stateY\right\rangle $ is in the transition
relation then the automata accepts a word ``abcd$\cdots$z''. Say
that $\neg\autom^{x}\in\stateS$, then $\pi_{\geq j}\forces\neg\autom^{x}$
and so $\pi_{\geq j+1}\nforces\autom^{y}$ for any $y$ reachable
from $x$ in $\autom^{x}$ by reading the symbol $g_{\Var}\left(\pi_{j+1}\right)$.
Yet again, $\stateT$ contains either $\autom^{y}$ or $\neg\autom^{y}$.
Since $\pi_{\geq j+1}\forces\bigwedge\stateT$ and $\pi_{\geq j+1}\nforces\autom^{y}$
it is clear that $\neg\autom^{y}\in\stateT$.\end{description}

$\left(\Longrightarrow\right)$ Say that $\pi\nforces\phi$, but that
the automata $\autom_{\phi}$ accepts the pair $\left(\pi,j+1\right)$.
Then there exists a path through $\autom_{\phi}$ labelled $g_{\Phi}\left(\pi_{\leq j}\right)$
ending at  a state $\stateT$ such that $\pi_{\geq j+1}\forces\bigwedge\stateT$;
let $\stateS$ be the state immediately preceding $\stateT$ along
that path. Since $\pi\nforces\phi$ and the lemma holds for $i=j$
we see that $\pi_{\geq j}\nforces\bigwedge\stateS$. From \prettyref{lem:s-iff-t},
we get a contradiction. 
\end{proof}

\begin{defi}
We say that an $\ATL$ formula is counter-free if all automata contained
in the formula are counter-free.
\end{defi}
Although we know that every LTL formula is equivalent  to some counter-free
automata in that they accept precisely the same strings/paths \citep{DiGa08Thomas},
note that it is not the case that no non-counter free automata is
equivalent to an LTL formula. For example, the following automata
accepts the same paths that satisfy $Gp$, yet it is not counter free
as $pp\in L_{a,a}$ but $p\notin L_{a,a}$.

\begin{center}
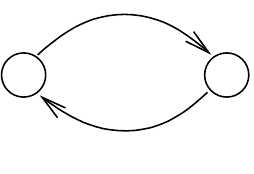 
\par\end{center}

We cannot assume that $\autom_{\phi}$ is counter free simply because
$\phi$ is equivalent to an LTL formula. We will now prove that $\autom_{\phi}$
is counter-free. Although we have not defined a traditional acceptance
condition for $\autom_{\phi}$, for the purposes of the next lemma
we will say that the automata accepts a word $g_{\Var}\left(\pi\right)$
iff $\autom_{\phi}$ accepts $\left(\pi,i\right)$ for all $i\geq0$. 
\begin{lemma}
\label{lem:powerset-is-counter-free}If $\phi$ is counter-free then
the automata $\autom_{\phi}$ is counter-free.\end{lemma}
\begin{proof}
Each state is a set of $\ATL$ formula, by taking the conjunction
of these \formulae{} we get an $\ATL$ formula $\psi$. Each automata
$\autom^{2}$ in $\psi$ comes from some automata $\autom^{1}$ in
$\phi$, and $\autom^{1}$ differs from $\autom^{2}$ only in the
initial states. Since $\autom^{1}$ is counter-free we see that $\autom^{2}$
is counter-free. Since each automata in $\ATL$ is counter-free we
can find an equivalent LTL formula, and so $\psi$ is equivalent to
some LTL formula $\psi'$.

If $\autom_{\phi}$ is not counter-free then there exists a positive
integer $m$, state $\stateS\in S_{\Phi}$ and word $u$ in $\alphabet^{*}$
such that $u^{m}\in L_{\stateS,\stateS}$ and $u\notin L_{\stateS,\stateS}$.
Since the states are non-contradictory we know that $\autom_{\phi}^{\stateS}$
accepts some word $w$. For any state $\stateT$ there exists some
formula $\formc$ such that $\neg\formc\in\stateS$ and $\formc\in\stateT$
or visa-versa. As such $\autom_{\phi}^{\stateT}$ does not accept
the word $w$. Since $u\notin L_{\stateS},_{\stateS}$ and $u^{m}\in L_{\stateS,\stateS}$
we see that $\autom_{\phi}^{\stateS}$ does not accept $u\cdot w$
but it does accept $u^{m}\cdot w$. By induction we discover that
for all non-negative $i$ the automaton $\autom_{\phi}^{\stateS}$
does not accept $u^{im+1}\cdot w$ but it does accept $u^{im}\cdot w$.
We see that any automaton that accepts this language must have a counter,
yet $\autom_{\phi}^{\stateS}$ is equivalent to an LTL formula and
so the language must be accepted by some counter-free automata. By
contradiction we know that $\autom_{\phi}$ is counter-free.
\end{proof}

\section{Translation into CTL{*}}

\label{sec:ctlstar}\foreignlanguage{english}{\label{sec:Translation-RoCTL*->CTL*.}}We
now present a translation from RoCTL{*} into CTL{*}. Note that $\prone\phi$
indicates that $\phi$ holds on the current path or a deviation. As
a convenience we use a psuedo-operator $\devi$ which indicates that
$\phi$ holds on a deviation. In \prettyref{sub:ALTL-to-automata}
we presented a translation from $\ATL$ into an automaton $\autom_{\phi}$;
In \prettyref{sub:devi} we will show how to construct an automaton
$\autom_{\devi\phi}$ which accepts iff $\autom_{\phi}$ would accept
on a deviation from the current path, and then translate $\prone\phi$
into $\phi\vee\autom_{\devi\phi}$. In \prettyref{sub:RoCTL*-to-ALTL-and-CTL*}
we combine these translations to provide a translation of RoCTL{*}
into $\ATL$ and then into CTL{*}.

\subsection{$\autom_{\phi}$ to $\autom_{\devi\phi}$}

\label{sub:devi}

\input{sub,devi.tex}

\input{sub,devi,proof.tex}

\begin{lemma}
The automaton $\autom_{\devi\phi}$ is counter-free.

\begin{proof}
Recall that a counter-free automaton is an automaton such that for
all states $s\in S$ and words $u$ in $\alphabet^{*}$, if $u^{m}\in L_{s,s}$
then $u\in L_{s,s}$.

If $s=\sF$ then every word $u$ is in $L_{s,s}$. If $s\ne\sF$ then
every path from $s$ to $s$ in $\autom_{\devi\phi}$ is also a path
from $s$ to $s$ in $\autom_{\phi}$, and $\autom_{\phi}$ is counter-free.\end{proof}

\end{lemma}

\subsection{RoCTL{*} to $\ATL$ and CTL{*}}

\label{sub:RoCTL*-to-ALTL-and-CTL*}

\global\long\def\RA{\varrho}

Here we define a translation $\RA$ from RoCTL{*} into $\ATL$. It is well known that we can express a CTL{*} formula as an LTL formula
over a path, where that path includes state formula as atoms; this is commonly used in model checking, see for example \citep{ModelChecking,318620,clarke1986automatic}.
This translation likewise replaces state \formulae{} with atoms.
It uses the standard translation of the $O$ operator found in \citep{FrDaRe07book},
and the $f_{\prone}$ translation from \prettyref{def:fprone}. The
translation $\RA$ is defined recursively as follows:
\begin{eqnarray*}
\RA(\phi\wedge\psi) & = & \RA\left(\phi\right)\wedge\RA\left(\psi\right)\\
\RA(\neg\phi) & = & \neg\RA(\phi)\\
\RA(A\phi) & = & p_{A\RA\left(\phi\right)}\\
\RA(O\phi) & = & p_{A\left(NG\neg\Viol\rightarrow\RA\left(\phi\right)\right)}\\
\RA(\robust\phi) & = & \neg f_{\prone}\left(\neg\RA\left(\phi\right)\right)\\
\RA(N\phi) & = & N\RA(\phi)\\
\RA(\phi U\psi) & = & \RA(\phi)U\RA(\psi)
\end{eqnarray*}

\begin{defi}
\label{def:fprone}For any $\ATL$ formula $\phi$, we define $f_{\prone}\left(\phi\right)$
to be $\phi\vee\autom_{\devi\phi}$.\end{defi}
\begin{theorem}
\label{thm:ALTL-truth}The translation $\RA$ of RoCTL{*} into $\ATL$
is truth-preserving if the atoms of the form $p_{A\psi}$ are assumed
to hold precisely at those worlds where $A\psi$ holds.\end{theorem}
\begin{proof}
It is easy to see from \prettyref{lem:devi-correct} that $\sigma\forces f_{\prone}\left(\phi\right)$
iff $\sigma\forces\prone\phi$. It is clear that $\sigma\forces O\phi$
iff $\sigma\forces A\left(NG\neg\Viol\rightarrow\phi\right)$ as $NG\neg\Viol$
satisfied precisely on the failure-free paths, this was proven more
formally in \citep{FrDaRe07book,MCD10}. From these facts it is easy
to see that $\RA$ is truth-preserving.\end{proof}
\begin{lemma}
\label{lem:fprone-1exp}The complexity of $f_{\prone}\left(\phi\right)$
is singly exponential in $\left|\phi\right|$.\end{lemma}
\begin{proof}
We see from \prettyref{def:set-of-states-Sphi} that the translation
of $\phi$ into $\autom_{\phi}$ results in an automaton that has
a number of states singly exponential in $\left|\phi\right|$. The
automaton $\autom_{\devi\phi}$ has exactly one more state than the
automata $\autom_{\phi}$, and so the number of states in $\autom_{\devi\phi}$
is also singly exponential in $\left|\phi\right|$.  From \prettyref{def:length-ATL},
the length of the $\ATL$ formula $\autom_{\devi\phi}$ is the number
of states in $\autom_{\devi\phi}$, and so $\left|\autom_{\devi\phi}\right|$
is singly exponential in $\left|\phi\right|$. As $f_{\prone}\left(\phi\right)=\autom_{\devi\phi}\vee\phi$
we see that $\left|f_{\prone}\left(\phi\right)\right|^{\star}$ is
singly exponential in $\left|\phi\right|$.\end{proof}
\begin{corollary}
The translation into $\ATL$ is at most $i$-exponential in length,
for \formulae{} with at most $i$ nested $\eA$ operators.

\global\long\def\RC{\mathbf{RC}}
\end{corollary}
\begin{defi}
We define a translation $\RC$ from RoCTL{*} into CTL{*} such that
for each RoCTL{*} formula $\phi$ we let $\RC\left(\phi\right)$ be
the $\ATL$ formula $\RA\left(\phi\right)$ with each atom of the
form $p_{A\psi}$ replaced with $A\psi$, and each automata in $\RA\left(\phi\right)$
replaced with the translation into an equivalent LTL formula referenced
in \prettyref{cor:Translating-a-counter-free}. 
\end{defi}
The following theorem follows from \prettyref{thm:ALTL-truth}. 

\begin{lemma}
\label{lem:Truth->Sat}Where $\tau\left(\phi\right)$ is a truth-preserving
translation from RoCTL{*} to CTL{*}, $\Gamma\left(\phi\right)$ is
both truth and satisfiability preserving, where $\Gamma\left(\phi\right)\equiv\tau\left(\phi\right)\wedge AGEN\neg\Viol$.\end{lemma}
\begin{proof}
Consider some RoCTL-structure $M$. Since $\SP\left(w\right)$ is
non-empty for any world $w$ of $M$, there exists some fullpath $\sigma\in\AP\left(w\right)$
such that $M,\sigma\forces N\neg\Viol$. Hence $M,w\forces EN\neg\Viol$.
Since this is true for any arbitrary $w$ we also see that $M,w\forces AGEN\neg\Viol$.
Thus for all fullpaths $\pi$ we have $M,\pi\forces\tau\left(\phi\right)\iff M,\pi\forces\Gamma\left(\phi\right)$,
and so $\Gamma$ is truth-preserving.

If $\phi$ is satisfiable we see that there exists a RoCTL-structure
$M$ and fullpath $\sigma$ through $M$ such that $M,\sigma\forces\phi$.
Hence $M,\sigma\forces\tau\left(\phi\right)$, and as before $M,\sigma\forces\Gamma\left(\phi\right)$.
Thus $\Gamma\left(\phi\right)$ is satisfiable.

Say $\Gamma\left(\phi\right)$ is satisfiable in CTL{*}. Then there
exists some CTL-structure $M$ and fullpath $\sigma$ through $M$
such that $M,\sigma\forces\Gamma\left(\phi\right)$. We can assume
without loss of generality that all worlds in $M$ are reachable from
$\sigma_{0}$, and so for every world $w$ in we have $M,w\forces EN\neg\Viol$.
Thus for every world $w$ we can pick a fullpath $\sigma$ starting
at $w$ such that $\sigma\forces GN\neg\Viol$, and so $\SP\left(w\right)$
is non-empty. By definition $M$ is a RoCTL-structure, and as $M,\sigma\forces\Gamma\left(\phi\right)$
we have $M,\sigma\forces\tau\left(\phi\right)$. Finally, $M,\sigma\forces\phi$,
and so $\phi$ is satisfiable in RoCTL{*}.\end{proof}

\begin{theorem}
\label{thm:RoCTL*->CTL*}The translation $\RC$ into CTL{*} is truth-preserving.
\end{theorem}
As the RoCTL-structures are precisely those structures where $\SP\left(w\right)$
is non-empty for each world $w$ (see \prettyref{lem:Truth->Sat} for more detail),
we have the following corollary. 
\begin{corollary}
The translation $\RC_{SAT}$ is satisfaction preserving (and truth
preserving) where $\RC_{SAT}\left(\phi\right)\equiv\RC\left(\phi\right)\wedge AGEN\neg\Viol$.\end{corollary}
\begin{theorem}
The translation $\RC$ is at most $\left(i+3\right)$-exponential
in the length, for \formulae{} with at most $i$ nested $\eA$ operators.\end{theorem}
\begin{proof}
From \prettyref{lem:fprone-1exp}, we see that there is at most a
singly exponential blowup per $\eA$ operator. Once we have translated
the whole formula into an $\ATL$ formula $\psi$, we know from \prettyref{cor:counter-free-FSA-to-LTL-3EXP}
that we can translate the automata into LTL \formulae{} with a 3-exponential
blowup. 

The automata are translated into LTL recursively, but the blowup remains
3-exponential. Say $\phi$ is the formula being translated. We see
that the number states in each automaton is no more than the complexity
$\left|\RA\left(\phi\right)\right|^{\star}$ of $\RA\left(\phi\right)$.
Thus with each recursion we multiply the length of the translated
formula by a number 3-exponential in $\left|\RA\left(\phi\right)\right|^{\star}$
which together still results in a 3-exponential blowup (note, for
example the formula $\left(2^{^{^{n}}}\right)^{i}$ is singly exponential
in $n$, not $i$-exponential in $n$).
\end{proof}

\section{Optimality of reduction into CTL{*}}

\label{sec:optimal}\label{sec:Succinctness}\input{Body_Succinctness__.tex}

\subsection{\label{sec:Easily-Translatable-Fragments}Easily Translatable Fragments
of RoCTL{*}}

Although the translation is non-elementary in the worst case we note
that real world{} formulas often fall into an easily translatable
fragment of RoCTL{*}. The most common use for nested Robustly operators
is to directly chain $n$ Robustly operators together to express the
statement ``Even with $n$ additional failures''. We also note that
when describing the behaviour of a system, the specification of the
system takes the form of a number of clauses each of which are reasonably
short, see  Example \ref{exa:coordinated-attack}. We will now show
that such formulas are easy to translate into CTL{*}, and that it
is easy to use CTL{*} decision procedures on such formulas.

It is easy to represent the statement ``$\Viol$ occurs at most $n$
times in future worlds'' in LTL, we will call this statement $\gamma^{n}$.
So for example, $\gamma^{0}\equiv NG\neg\Viol$, $\gamma^{1}\equiv N\left(\neg\Viol UNG\neg\Viol\right)$,
and so forth. Note that $\left|\gamma^{n}\right|\in\bigO\left(n\right)$.
We see that translating $\devi^{n}\phi$ is no more complex than translating
$\devi\phi$; we can translate $\devi^{n}\phi$ the same way as we
translated $\devi\phi$ as above, but we replace $\psi_{s_{i}}$ with
\begin{eqnarray*}
E\left(\bigwedge s_{i}\wedge N\gamma^{n-1}\right)\mbox{ .}
\end{eqnarray*}

We see that $\prone\phi$ means $\phi$ holds on the original fullpath
or a deviation, $\prone\prone\phi$ means that $\phi$ holds on the
original path or a deviation, or a deviation from a deviation. In
general $\prone^{n}\phi$ means that $\phi$ holds on some path at
most $n$ deviations from the current path. Thus:
\begin{eqnarray*}
\prone^{n}\phi & \equiv & \phi\vee\devi\phi\vee\cdots\vee\devi^{n}\phi\mathfullstop
\end{eqnarray*}

Thus we see that the length of the translation of $\prone^{n}\phi$
is linear in $n$, and thus has no overall effect on the order of
complexity. Note that $\robust^{n}\phi\equiv\neg\prone^{n}\neg\phi$,
so $\robust^{n}$ is also no harder to translate than a single $\robust$
operator.  This is significant because one of the motivations of
RoCTL{*} was to be able to express the statements of the form ``If
less than $n$ additional failures occur then $\phi$''. The related
statement ``If $n$ failures occur then $\phi$'' is ever easier
to translate into CTL{*} as $O\robust^{n}\phi\equiv A\left(\gamma^{n}\rightarrow\phi\right)$.

Let the $\robust$-complexity of a formula $\phi$ be defined as follows:
\begin{eqnarray*}
\left|\phi\right|_{\robust} & = & \max_{\robust\psi\leq\phi}\left|\psi\right|\mbox{ .}
\end{eqnarray*}

It is clear that there exists some function $f$ such that for all
RoCTL{*} formula $\phi$ of length $n$ the translation of $\phi$
into CTL{*} is of length $f\left(n\right)$ or less. As the translation
of $\prone$ does not look inside state \formulae{} it is clear that
$\left|c\left(\phi\right)\right|\in\bigO\left(f\left(\left|\phi\right|_{\robust}\right)\left|\phi\right|\right)$.
In other words, for any fragment of RoCTL{*} where the length $\left|\phi\right|_{\robust}$
of path-\formulae{} contained within a $\robust$ operator is bounded
there is a linear translation from this fragment to CTL{*}. As a result
the complexity properties of RoCTL{*} \formulae{} with  bounded
$\left|\phi\right|_{\robust}$ are similar to CTL{*}; we can decide
the satisfiability problem in doubly exponential time and the model
checking problem in time singly exponential in the length of the formula
and linear in the size of the model, see \citep{ModelChecking} for
an example of a model checker for CTL{*}.

We can refine both above results by noting that the construction of
$\autom_{\devi\phi}$ does not look inside state \formulae{}. Thus
a fragment of RoCTL{*} which has a bounded number of $\robust^{n}$
nested within a path-formula (unbroken by $A$ or $O$) has an elementary
translation into CTL{*}.

In \citep{two_seconds}, we discussed a fragment of RoCTL{*} called
State-RoCTL\@. This fragment could naturally express many interesting
robustness properties, but had a linear satisfaction preserving translation
into CTL\@. The truth-preserving translation of State-RoCTL into
RoCTL{*} was technically exponential, but had a linear number of unique
sub-formulas and so has a natural and efficient compressed format;
for example, the truth-preserving translation provided a polynomial-time
model checking procedure.

\section{Conclusion}

\label{sec:concl}\input{sprobConcl.tex}

\bibliographystyle{ACM-Reference-Format-Journals}
\bibliography{csbibtex-acm}

\end{document}

%% file: sprobIntroCut.tex
\markboth{J. C. McCabe-Dansted et al.}{Specifying Robustness}
\title{Specifying Robustness}
\author{John C. McCabe-Dansted,
Tim French \and{}
Mark Reynolds
\affil{University of Western Australia}
Sophie Pinchinat
\affil{Campus Universitaire de Beaulieu}}
\begin{abstract}
This paper proposes a new logic RoCTL{*} to model robustness in concurrent
systems. RoCTL{*} extends CTL{*} with the addition of Obligatory and
Robustly operators, which quantify over failure-free paths and paths
with one more failure respectively. We present a number of examples
of problems to which RoCTL{*} can be applied.

The core result of this paper is to show that RoCTL{*} is expressively
equivalent to CTL{*} but is non-elementarily more succinct. We present
a translation from RoCTL{*} into CTL{*} that preserves truth
but may result in non-elementary growth in the length
of the translated formula as each nested Robustly operator may result
in an extra exponential blowup. 
However, we show that this translation is optimal
in the sense that any equivalence preserving translation will require
an extra exponential growth per nested Robustly. 
We also compare RoCTL{*}
to Quantified CTL{*} (QCTL{*}) and hybrid logics. 
\end{abstract}


\category{F.4.1}{Mathematical Logic and Formal Languages}{Temporal Logics}

\terms{Algorithms, Languages, Reliability, Theory, Verification}

\keywords{Robustness, succinctness, branching time}

\acmformat{John C. McCabe-Dansted, Tim French, Mark Reynolds and Sophie Pinchinat, 2013. Specifying Robustness}

\begin{bottomstuff}
This project is supported by the Australian Government's International Science Linkages program and the Australian Research Council.

Author's addresses: J. C. McCabe-Dansted, T. French {and} M. Reynolds,
School of Computer Science and Software Engineering,
University of Western Australia;
Sophie Pinchinat,
IRISA,
Campus Universitaire de Beaulieu
\end{bottomstuff}

\providecommand{\uline}[1]{\term{#1}}

\renewcommand{\uline}{\textbf}

\input{csmacros.tex}
\renewcommand\formulae{formulas}

\maketitle
\section{Introduction }

We introduce the Robust Full Computation Tree Logic (RoCTL{*})
as
an extension of 
the branching time temporal logic CTL{*} to represent issues relating
to robustness and reliability in systems. It does this by adding an
Obligatory operator and a Robustly operator. The Obligatory operator
specifies how the systems should ideally behave by quantifying over paths
in which no failures occur. The Robustly operator specifies that something
must be true on the current path and similar paths that ``deviate''
from the current path, having at most one more failure occurring.
This notation allows phrases such as ``even with $n$ additional
failures'' to be built up by chaining $n$ simple unary Robustly
operators together.

RoCTL{*} is a particular combination of 
temporal and deontic logics allowing reasoning about
how requirements on behaviour 
are progressed and change with time,
and the unfolding of actual events.
The RoCTL{*} Obligatory operator is similar to the Obligatory operator
in Standard Deontic Logic (SDL), although in RoCTL{*} the operator
quantifies over paths rather than worlds. 
However, it is the Robustly operator
which gives RoCTL{*} many advantages
over a simple 
combination of temporal logic and
deontic logic as
in \cite{vandertorre98temporal}.
SDL has many paradoxes and 
some of these, such as the ``Gentle Murderer'' paradox
(``if you murder, you must murder gently''~\cite{Fo84}), spring from
the inadequacy of SDL to deal with obligations caused by acting
contrary to duty. 
Contrary-to-Duty (CtD) obligations are important for modeling a robust
system, as it is often important to state that the system should achieve
some goal and also that, if it fails, then it should 
act to mitigate or in some way recover
from the failure. 

RoCTL{*} can represent CtD obligations by specifying
that the agent must ensure that the CtD obligation is met even if
a failure occurs. 
SDL
is able to distinguish what ought to be true from what is true, but
is unable to specify obligations that come into force only when we
behave incorrectly. 
Addition of temporal operators to deontic logic allows us to specify
correct responses to failures that have occurred in the past~\cite{vandertorre98temporal}.
However, this approach alone is not sufficient~\cite{vandertorre98temporal}
to represent obligations such as ``You must assist your neighbour,
and you must warn them iff you will not assist them''. In RoCTL{*}
these obligations can be represented if the obligation to warn your
neighbour is robust but the obligation to assist them is not. 

\newcommand{\ignore}[1]{}
\ignore{
Other approaches to dealing with Contrary-to-Duty obligations exist.
Defeasible logic is often used~\cite{dblp:journals/fuin/mccarty94},
and logics of agency, such as STIT~\cite{be91}, can be useful as
they can allow obligations to be conditional on the agent's ability
to carry out the obligation.  
}

A number of other extensions of temporal logics have been proposed
to deal with deontic or robustness issues~\cite{jan_broersen_designing_2004,w_long_quantification_2000,hansson94logic,huib_aldewereld_designing_2005,agerri_rodrigo_normative_2005}.
Each of these logics are substantially different from RoCTL{*}. Some
of these logics are designed specifically to deal with deadlines~\cite{jan_broersen_designing_2004,hansson94logic}.
The Agent Communication Language was formed by adding deontic and other
modal operators to CTL~\cite{agerri_rodrigo_normative_2005}; this
language does not explicitly deal with robustness or failures. 
\citet{hansson94logic} 
proposed an extension of CTL
to deal with reliability. However, as well as being intended to deal with deadlines,
their logic reasons about reliability using probabilities rather than
numbers of failures, and their paper does not contain any discussion
of the relationship of their logic to deontic logics. Like our embedding
into QCTL{*}, 
\citet{huib_aldewereld_designing_2005}
uses a Viol atom to represent failure. However, their logic also uses
probability instead of failure counts and is thus suited to a different
class of problems than RoCTL{*}. Another formalisation of robustness
is representing the robustness of Metric Temporal Logic (MTL) formulas
to perturbations in timings 
\cite{bouyer-robust}.
None of these logics appear to have an operator that is substantially
similar to the Robustly operator of RoCTL{*}.

In the last few years there has been considerable interest in logics
for reasoning about systems that are robust to partial
non-compliance with the norms. One approach has been to define 
{\em robustness}
in terms of the ability of a multi-agent system to deal with having
some subset of agents that are unwilling or unable to comply with
the norms \cite{HRW08,AHW10}. Like RoCTL{*} they consider socially
acceptable behaviours to be a subset of physically possible behaviours.
A logic that like can discuss numbers of faults was suggested by \cite{FNP10},
though this logic extended ATL instead of CTL{*} and defined fault-tolerance
in terms of numbers of winning strategies. More recently the Deontic
Computation Tree Logic (dCTL) was proposed \cite{castro2011dctl}.
Like RoCTL{*} the logic divides states into normal and abnormal states,
but avoids capturing the full expressivity of CTL{*} to allow the
model checking property to be polynomial like the simpler CTL logic.
There is a restriction of RoCTL{*} that can be easily translated into
CTL \cite{two_seconds}, allowing this restriction to be reasoned about
as efficiently as CTL; however, dCTL is more expressive than CTL \cite{castro2011dctl}.
Finally, a Propositional Deontic Logic was proposed by \cite{AKC12}
than divided events into allowable and non-allowable depending on
the current state.

Diagnosis problems in control theory~\cite{jmpc,avw} also deals
with failures of systems. Diagnosis is in some sense the dual of the
purpose of the RoCTL{*} logic, as diagnosis requires that failure
cause something (detection of the failure) whereas robustness involves
showing that failure will \emph{not} cause something.

This paper provides some examples of robust systems that can be effectively
represented in RoCTL{*}. It is easy to solve the coordinated attack
problem if our protocol is allowed to assume that only $n$ messages
will be lost. The logic may also be useful to represent the resilience
of some economy to temporary failures to acquire or send some resource.
For example, a remote mining colony may have interacting requirements
for communications, food, electricity 
and fuel. RoCTL{*} may be more
suitable than Resource Logics (see for example~\cite{weerdt01resource})
for representing systems where a failure may cause a resource to become
temporarily unavailable. 
This paper presents a simple example where
the only requirement is to provide a cat with food when it is hungry.

The Obligatory operator,
as well as some uses of the Robustly operator, are easy to translate
into CTL{*}~\cite{DBLP:conf/jelia/McCabe-Dansted08}
but a general 
way to achieve a translation to CTL{*} is not obvious.
The first translation
in our paper
is of RoCTL{*} into the tree semantics
of Quantified CTL{*} (QCTL{*}). We note that a similar translation
can be made into a fragment of Hybrid temporal logic. Although QCTL{*}
is strictly more expressive than CTL{*} the translation of RoCTL{*}
into QCTL{*} will be given for two reasons. Firstly the translation
into QCTL{*} is very simple and thus is well suited as an introduction
to reasoning with RoCTL{*}. Secondly, even this weak result is sufficient
to demonstrate that RoCTL{*} is decidable. Finally, the translation
into QCTL{*} is linear, while it will be shown that any translation
to CTL{*} must be non-elementary in the worst case. 

We then give a translation of RoCTL{*} formulas into CTL{*}. 
This results in a formula
that is satisfied on a model iff the original formula is satisfied
on the same model. This means that we can use all the CTL{*} model
checkers, decision procedures and so forth for RoCTL{*}.
Unfortunately,
the translation can be quite long.
We show that although all RoCTL{*} \formulae{} can be translated
into CTL{*}, the length of the CTL{*} formula is not elementary in
the length of the RoCTL{*} formula. Hence some properties can be represented
much more succinctly in RoCTL{*} than CTL{*}.
This
translation requires roughly one extra exponential per nested robustly
operator. We will show that no translation can do better than this,
so although RoCTL{*} is no more expressive than CTL it is very succinct
in the sense that any translation of RoCTL{*} into either CTL{*} or
tree automata will result in a non-elementary blowup in the length
of the formula.

We can summarise the contributions of this paper as follows.
Firstly, it defines a new intuitive and expressive logic, RoCTL{*}, for specifying robustness in systems.
The logic seems to combine temporal and deontic notions in
a way that captures the important contrary-to-duty obligations
without the usual paradoxes.
Secondly, it provides a proof that the logic can be translated in a truth-preserving manner
into the existing CTL* logic.
Thirdly, it provides a proof that RoCTL{*} is non-elementarily more succinct than CTL*
for specifying some properties.

This paper extends results from the conference papers \cite{FrDaRe07book,DBLP:conf/time/McCabe-DanstedFRP09,Da11Cats}. 
There is further discussion and more details in the thesis
\cite{MCD10}.

The structure of the paper is as follows.
RoCTL{*}
is introduced in the next section
before we
show that the new logic can be applied across
a wide variety of examples,
practical, theoretical and philosophical.
In section~\ref{sec:machinery}, we
revise a large collection of
existing machinery that
we will need in the subsequent 
expressivity and succinctness proofs.
In section~\ref{sec:bisim}, we
show that RoCTL{*} is preserved under bisimulations:
needed for some unwinding proofs, but also interesting to have.
In section~\ref{sec:qctl}, we
show the fairly straightforward translation
of RoCTL{*} into QCTL{*}.
Section~\ref{sec:altl} presents
some useful conversions between automata.
Section~\ref{sec:ctlstar} contains the
translation of RoCTL{*} into CTL{*}.
In section~\ref{sec:optimal}, we
show that this translation is optimal.

%% file: csmacros.tex
\providecommand{\uline}{\emph}

\global\long\def\formulae{formul\ae{}}

\global\long\def\lemmas{lemmas}

\providecommand{\logicindex}[1]{\index{Logics!#1}} 

\global\long\def\underbar{\term}

\providecommand{\term}{}

\global\long\def\doi#1{\href{http://dx.doi.org/#1}{doi:#1}}
 \global\long\def\retrieved#1#2{Retrieved on #1: \url{#2}}


\newrefformat{exa}{Example \ref{#1}}

\newrefformat{apx}{Appendix \ref{#1}} 
\newrefformat{sec}{Section \ref{#1}} 
\newrefformat{fig}{Figure \ref{#1}} 
\newrefformat{def}{Definition~\ref{#1}} \newrefformat{defn}{Definition~\ref{#1}}
\newrefformat{cor}{Corollary~\ref{#1}} \newrefformat{sub}{Section~\ref{#1}}
\newrefformat{prop}{Proposition \ref{#1}}

\global\long\def\aform{\phi}

\global\long\def\access{R}

\global\long\def\allworlds{S}

\global\long\def\ATL{\mathcal{A}LTL}

\global\long\def\axOiP{O5}
 \global\long\def\axONiNO{O6}
 \global\long\def\axNOiN{O7}

\global\long\def\bctlstruct{\left(\mfw,\ra,\valuation,B\right)}

\global\long\def\bisimulation{\mathfrak{B}}

\global\long\def\Forall#1{\forall#1:\,}

\global\long\def\fsX#1{\mathsf{#1}}

\global\long\def\fsNOAR{\fsX{NORA}}

\global\long\def\Exists#1{\exists#1:\,}

\global\long\def\myth{\scriptsize{\mbox{th}}}

\global\long\def\ith{i^{\myth}}
 \global\long\def\mth{m^{\myth}}

\global\long\def\liff{\leftrightarrow}

\global\long\def\mathcomma{\mbox{ ,}}
 \global\long\def\mathfullstop{\mbox{ .}}

\global\long\def\Na{N_{a}^{-1}}

\global\long\def\negU{\overset{\neg}{U}}

\global\long\def\node{\mathbf{n}}

\global\long\def\normf{\Xi}

\global\long\def\pvsW{w}
 \global\long\def\pvsWB{\hat{w}}

\global\long\def\pvsOf#1#2{#1_{#2}}

\global\long\def\pvs{\pvsOf{(\allworlds,\access,g)}{\pvsW}}
 \global\long\def\pvsB{\pvsOf{(\hat{\allworlds},\hat{R},\hat{g})}{\pvsWB}}

\global\long\def\pvrs{\pvsOf{(\allworlds,\access^{s},R^{f},g)}{\pvsW}}
 \global\long\def\pvrsB{\pvsOf{(\hat{\allworlds},\hat{R}^{s},\hat{R}^{f},\hat{g})}{\pvsWB}}

\global\long\def\setdiff{-}

\global\long\def\suffix{\mathbf{\mbox{suffix}}}

\global\long\def\pathc{\rho}

\global\long\def\QForall#1{\forall#1}

\global\long\def\Qp{\QForall p}
 \global\long\def\Qy{\QForall y}

\global\long\def\tuple#1{\left(#1\right)}

\global\long\def\States{Q}

\global\long\def\initial{\States_{0}}

\global\long\def\devi{{\bf \Lambda}}

\global\long\def\true{\top}

\global\long\def\false{\bot}

\global\long\def\fgap{\,}
 \global\long\def\fstop{\fgap.}
 \global\long\def\fcomma{\fgap,}
 \global\long\def\BNFbar{\,|\,}
 \global\long\def\BNFeq{::=}

\global\long\def\tN{N}
 \global\long\def\tU{U}
 \global\long\def\eA{\blacktriangle}
 \global\long\def\eE{\triangle}
 \global\long\def\dO{O}
 \global\long\def\dP{P}
 \global\long\def\AP{\func{ap}}
 \global\long\def\SF{\func{ap}}
 \global\long\def\tG{G}
 \global\long\def\Union#1#2{\dyadic#1\union#2}
 \global\long\def\dyadic#1#2#3{\left(#1#2#3\right)}
 \global\long\def\lImplies#1#2{\dyadic{#1}{\rightarrow}{#2}}
 \global\long\def\stateforma{A\forma}
 \global\long\def\stateformb{A\formb}
 \global\long\def\aform{\phi}
 \global\long\def\lUntil#1#2{\dyadic{#1}U{#2}}
 \global\long\def\lWeak#1#2{\dyadic{#1}W{#2}}

\global\long\def\classM{\classvs}
 \global\long\def\forma{\aform}
 \global\long\def\formb{\psi}
 \global\long\def\formc{\theta}
 \global\long\def\formd{\gamma}
 \global\long\def\classvs{\mathbb{M}}
 \global\long\def\lAnd#1#2{\dyadic{#1}{\wedge}{#2}}

\global\long\def\valuation{g}

\global\long\def\jzseq#1#2{#1_{0},#1_{1},\ldots,#1_{#2}}

\global\long\def\alphabet{\Sigma}

\global\long\def\transition{\delta}

\global\long\def\accepting{F}

\global\long\def\buchi{B\"uchi}

\global\long\def\autom{\mathcal{A}}

\global\long\def\robust{\blacktriangle}

\global\long\def\prone{\triangle}

\global\long\def\Var{\mathcal{V}}

\global\long\def\LTLof#1{\CTLf(#1)}

\global\long\def\CTLf{\buchif^{-1}}

\global\long\def\buchif{\mathfrak{A}}

\global\long\def\Viol{{\bf v}}

\global\long\def\union{\cup}

\global\long\def\forces{\vDash}
 \global\long\def\forcesw{\forces^{E}}

\global\long\def\nforces{\nvDash}

\global\long\def\lang{\mathcal{L}}

\global\long\def\vP{\valuation_{\Phi}}

\global\long\def\vPp{g_{\Phi_{\prone}}}

\global\long\def\natnum{\mathbb{N}}

\global\long\def\negrass#1#2{#1\nrightarrow^{s}#2}
 \global\long\def\mfm{w}
 \global\long\def\failmap{\epsilon}
 \global\long\def\patha{\sigma}
 \global\long\def\pathb{\pi}
 \global\long\def\tiff{\text{ iff }}
 \global\long\def\st{\text{ s.t. }}
 \global\long\def\tand{\text{ and }}
 \global\long\def\mfkf{(\mfw,\ras,\raf,\valuation)}
 \global\long\def\pathdf{\tau}

\global\long\def\ras{\access^{s}}
 \global\long\def\raf{\access^{f}}
 \global\long\def\mfw{\allworlds}

\global\long\def\MCTL{M}
 \global\long\def\mfk{M}
 \global\long\def\mfs{\mathcal{A}}
 \global\long\def\transframe{(\mfw,\ra)}
 \global\long\def\mfsf{XXXXX}
 \global\long\def\classf{\classM}
 \global\long\def\classC{\mathbb{C}}
 \global\long\def\classCt{\classC_{t}}
 \global\long\def\rac{\access}
 \global\long\def\ra{\access}
 \global\long\def\ctlstruct{\tuple{\mfw,\access,\valuation}}
 \global\long\def\treestruct{\ctlstruct}
 \global\long\def\sF{s_{F}}
 \global\long\def\bigO{\mathcal{O}}
 \global\long\def\prev{\mbox{prev}}

\global\long\def\cl{\textbf{cl}}
 \global\long\def\closure{\cl\phi}

\global\long\def\func#1{\mbox{\textbf{#1}}}
 \global\long\def\height{\func{\heightT}_{\ra}}
 \global\long\def\heightT{\func{height}}
 \global\long\def\rootT{\func{root}}

\global\long\def\statea{q}
 \global\long\def\stateb{r}

\global\long\def\jSeqOneTo#1#2{#1_{1},#1_{2},\ldots,#1_{#2}}

\global\long\def\Lstar{\square\raisebox{-2pt}{\hspace{-1.2ex}*}\hspace{0.36ex}}

\global\long\def\classK{\mathbb{K}}
 \global\long\def\fsK{\mathsf{K}}
 \global\long\def\fsOA{\mathsf{OA}}
 \global\long\def\fsO{\mathsf{O}}
 \global\long\def\fsOAR{\mathsf{OA\eA}}
 \global\long\def\classOAR{\mathbb{OA\eA}}
 \global\long\def\fsX#1{\mathsf{#1}}

\global\long\def\aformalsystem{\mathsf{S}}
 \global\long\def\fsS{\aformalsystem}
 \global\long\def\fsPC{\fsX{PC}}
 \global\long\def\conseqrel{\vdash}
 \global\long\def\crS{\conseqrel_{\aformalsystem}}
 \global\long\def\structure{\mathcal{A}}
 \global\long\def\valustruct{(\structure,\valuation)}

\global\long\def\classOA{\mathbb{OA}}
 \global\long\def\li{\rightarrow}
 \global\long\def\lIFF{\leftrightarrow}

\global\long\def\forceso{\forces_{\square}}
 \global\long\def\forcesr{\forces_{\star}}

\global\long\def\rasf{\access^{sf}}

\providecommand{\ctls}{CTL*}\providecommand{\ctl}{CTL}

\global\long\def\DP{\func{dp}}
 \global\long\def\SP{\func{sp}}

\global\long\def\CTLv{CTL{}$_{\Viol}${}}

\global\long\def\RoCTLvstarLogic{RoCTL{}*\hspace{-1ex}$_{\Viol}$}

\global\long\def\RoBCTLvstarLogic{RoBCTL{}*\hspace{-1ex}$_{\Viol}$}

\global\long\def\CTLvstar{CTL\makebox[0cm][l]{*}$_{\Viol}$}

\global\long\def\CTLvstruct{CTL{}$_{\Viol}${} structure}

\global\long\def\CTLvStruct{CTL{}$_{\Viol}${} Structure}

\global\long\def\RoBCTLvstar{RoB\CTLvstar}

\global\long\def\RoBCTLv{RoB\CTLv}

\global\long\def\RoBCTLvstruct{RoB\CTLvstruct}

\global\long\def\RoCTLvstar{Ro\CTLvstar}

\global\long\def\RoCTLv{Ro\CTLv}

\global\long\def\RoCTLvstruct{Ro\CTLvstruct}

\global\long\def\RoCTLvStruct{Ro\CTLvStruct}

%% file: Body_Examples__v.tex
\providecommand{\ifnoroot}[1]{#1} 

\ifnoroot{\input{csmacros.tex}}

\global\long\def\tA{A}

In this section a number of examples are presented. These examples
will demonstrate how combinations of RoCTL{*} operators can be used,
and contrast the meaning of apparently similar combinations. A number
of problem domains will be touched on briefly.

The first example will show how a variant of Chisholm's paradox can
be represented in RoCTL{*}. Example~\ref{exa:ONNO} examines the
difference between the formula $NO\forma$ and the formula $ON\phi$,
and shows how this combination of operators can be used to represent
a contrary-to-duty obligation that is triggered by a failure in the
past. Example~\ref{exa:coordinated-attack} shows how RoCTL{*} may
be used to specify a robust network protocol, in this case relating
to the coordinated attack problem. Example~\ref{exa:cat} uses the
feeding of a cat to show how we can reason about consequences of policies
in RoCTL{*}. These examples frequently use the $\eA$/$\eE$ operator
to form the pair $O\eA$;  Example~\ref{exa:bitflip} exhibits the
simple formula $O(\prone Fe\rightarrow Fw)$ which nests $\eA$/$\eE$
in a less trivial way.  Example~\ref{exa:fuse} also nests $\eA$
in a less trivial way, as it is used to compare the meaning of $\eA G$
with the meaning of $G\eA$. 

In each of these examples, an informal English requirement will be
listed with formal specification as a RoCTL{*} formula. The informal
requirements will have flavor and explanation that may not be expressed
in the formal specification, and thus should not be interpreted as
simple translations from RoCTL{*} to English.
\begin{exam}
\label{exa:Chisholm}We may represent a variant of Chisholm's paradox
\cite[p34--5]{chisholm1963contrary}  as follows:

$OFh$: You must help your neighbour (eventually)

$O\eA\left(\neg Fh\liff Fw\right)$: You must warn your neighbour
that you will not help them iff you will not help them, even if a
single failure occurs.
\end{exam}
Note that $O\left(\neg Fh\li Fw\right)$ would be redundant given
$OFh$, as all failure-free paths would satisfy $Fh$ and thus $O\left(\neg Fh\li Fw\right)$
would be vacuously true. However, $O\eA\left(\neg Fh\li Fw\right)$
is not redundant, as this indicates that even if a single failure
occurs. As with similar defeasible representations of this problem,
the obligation to warn the neighbour is meaningful as it is stronger
than the obligation to help the neighbour.

It may seem that the obligation to eventually help your neighbour
is vacuous, as one could aways claim that they will help their neighbour
sometime later. In RoCTL{*} the obligation is not vacuous, as following
a path where you never help the neighbour violates the norm. A common
sense interpretation of this is, if you plan to never help your neighbour,
then lying about that plan does not satisfy the first obligation,
rather it also violates the second. RoCTL{*} focuses on modelling
and verifying systems. It is reasonable and meaningful to state that
a task must complete in finite time without specifying a deadline.
Additionally, we note that if we have had multiple perfect opportunities
to help our neighbour, and did not do so, the neighbour may become
rightfully suspicious that we do not plan to help them; however, diagnosing
systems on the basis of behaviour is outside the scope of \thispaper{}. 
\begin{exam}
\label{exa:ONNO}Here is an example of a simple Contrary-to-Duty obligation.
This provides a counter example to both $ON\forma\rightarrow NO\phi$
and $NO\forma\rightarrow ON\phi$.
\end{exam}
In some case decisiveness may be more important than making the right
decision. For example, when avoiding collision with an object we may
have the choice of veering right or left. In this case it may be more
efficient to veer to the right, and so we should make this decision.
However, changing our mind could cause a collision, so it is best
to stay with the inferior decision once chosen. We show how we may
formalise such a decision to demonstrate the difference in the meaning
of $ON$ and $NO$ in RoCTL{*}.
\begin{description}
\item [{$ON(Gp)$}] You should commit to the proper decision. (It is obligatory
that by the next step, you will always ``act according to proper
decision'' {[}$p${]})
\item [{$NO\left(G\neg p\vee Gp\right)$}] Once you have made your decision,
you should stick with it. (at the next step it is obligatory that
you will always not $p$ or always $p$)
\end{description}
It is logically consistent with both the above that we do not make
the proper decision ($N\neg p$), as the above only specifies what
should happen not what actually will happen. Once we have made the
wrong decision we cannot satisfy $Gp$, so we should stick with the
wrong decision $G\neg p$. Hence, in this case, both $ON(Gp)$ and
$NO(G\neg p)$ are true.  Likewise $ON(G\neg p)$ and $NO(Gp)$ are
false. This demonstrates how obligations can change with time in RoCTL{*}. 

We will now give an example of a structure $\mfk=\ctlstruct$ that
satisfies these \formulae{}:%
\begin{minipage}[t]{0.55\columnwidth}%
\begin{align*}
\allworlds & =\{u,v,w,w'\}\mbox{,}\\
\ra & =\left\{ (u,v),(v,v),\left(u,w'\right),\left(w',w\right)(w,w)\right\} \mbox{,}\ \\
\valuation(v) & =\left\{ p\right\} ,\quad\valuation(w)=\valuation\left(u\right)=\emptyset,\quad\valuation(w')=\left\{ \Viol\right\} \mathfullstop
\end{align*}
\end{minipage}%
\begin{minipage}[t]{0.3\columnwidth}%
\raisebox{-1in}{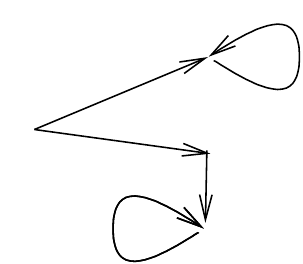}%
\end{minipage}

Let  $\patha$ be the fullpath $\left\langle u,w,w',w',\dots\right\rangle $
corresponding to making the wrong decision. We see that $M,\patha_{\geq1}\forces\neg p$,
so for every failure-free path $\pi$ starting at $\sigma_{1}$ we
have $M,\pi\forces\neg p$ and hence $M,\patha_{\geq1}\forces O\neg p\wedge\neg Op$.
Thus $M,\patha\forces NO\neg p\wedge N\neg Op$. As $N$ is its own
dual it follows that $M,\patha\forces\neg NOp$. 

Let $\pathb=\left\langle v,v,\ldots\right\rangle $. We see that $M,\pathb\forces p$.
We see that $\SP(u)=\left\{ \left\langle u,v,v,\ldots\right\rangle \right\} $.
Hence $M,\patha\forces ONp$ and it follows that $M,\patha\forces\neg O\neg Np$
and so $M,\patha\forces\neg ON\neg p$.

Hence $M,\patha\forces\lAnd{ONp}{\neg NOp}$ and so $M,\patha\nforces\lImplies{ON\forma}{NO\forma}$
where $\forma=p$. Likewise $M,\patha\forces\lAnd{NO\neg p}{\neg ON\neg p}$,
so $M,\patha\nforces\lImplies{NO\forma}{ON\forma}$ where $\forma=\neg p$.

It is well known that simple combinations of deontic and temporal
logics can represent contrary-to-duty obligations of the form ``If
you have previously done $\phi$, you should do $\psi$''. We now
give an example of a contrary-to-duty obligation RoCTL{*} can express
where time is not central to the obligation.
\begin{exam}
\label{exa:coordinated-attack}In the coordinated attack problem we
have two generals $X$ and $Y$. General $X$ wants to organise an
attack with $Y$. A communication protocol will be presented such
that  a coordinated attack will occur if no more than one message
is lost. 
\end{exam}
The coordinated attack problem requires that the both generals know
that the other will attack despite the possibility that any message
could be lost. This is known to be impossible. We will show how we
can specify a policy on RoCTL{*} that specifies a weaker variant of
the coordinated attack problem where we can achieve a coordinated
attack provided no more than one message is intercepted (and both
generals are willing assume that no more than one message will be
lost). 
\begin{description}
\item [{$AG\left(s_{X}\rightarrow ONr_{Y}\right)$:}] If $X$ sends a message,
$Y$ should (in an ideal world) receive it at the next step. Note
that it may not actually be the case that the message arrives as it
may be intercepted.
\item [{$AG\left(\neg s_{X}\rightarrow\neg Nr_{Y}\right)$:}] If $X$ does
not send a message now, $Y$ will not receive a message at the next
step.
\item [{$AG(f_{X}\rightarrow AGf_{X})$:}] If $X$ commits to an attack,
$X$ cannot withdraw.
\item [{$AG(f_{X}\rightarrow\neg s_{X})$:}] If $X$ has committed to an
attack, it is too late to send messages.
\item [{$A\left(\neg f_{X}Wr_{X}\right)$:}] $X$ cannot commit to an attack
until $X$ has received a message (which would contain plans from
$Y$).
\item [{$A\left(\neg r_{X}Ws_{Y}\right)$:}] $X$ will not receive a message
until $Y$ sends one.
\end{description}
Similar constraints to the above also apply to $Y$. Below we add
a constraint requiring $X$ to be the general planning the attack
\begin{description}
\item [{$A\left(\neg s_{Y}Wr_{Y}\right)$:}] General $Y$ will not send
a message until $Y$ has received a message. 
\end{description}
No protocol exists to satisfy the original coordination problem, since
an unbounded number of messages can be lost. Here we only attempt
to ensure correct behaviour if one or fewer messages are lost. 
\begin{description}
\item [{$A\left(s_{X}Ur_{X}\right)$:}] General $X$ will send plans until
a response is received.
\item [{$AG\left(r_{X}\rightarrow f_{X}\right)$:}] Once general $X$ receives
a response, $X$ will commit to an attack.
\item [{$A\left(\neg r_{Y}W\left(r_{Y}\wedge\left(s_{Y}\wedge Ns_{Y}\wedge NNf_{Y}\right)\right)\right)$:}] Once
general $Y$ receives plans, $Y$ will send two messages to $X$ and
then commit to an attack.
\end{description}
Having the formal statement of the policy above and the semantics
of RoCTL{*} we may prove that the policy $\hat{\phi}$ is consistent
and that it implies correct behaviour even if a single failure occurs:
\begin{align*}
\hat{\forma}\rightarrow O\eA F(f_{X}\wedge f_{Y})\fstop
\end{align*}

Indeed, we will shown that such issues can be decided in finite time
in Section~\ref{sec:A-Linear-reduction-into-QCTL*}.

For a more thorough specification of the Coordinated Attack problem,
see for example \cite{halpern1990knowledge}. 
\begin{exam}
\label{exa:cat}We have a cat that does not eat the hour after it
has eaten. If the cat bowl is empty we might forget to fill it. We
must ensure that the cat never goes hungry, even if we forget to fill
the cat bowl one hour. At the beginning of the first hour, the cat
bowl is full. We have the following atoms:\end{exam}
\begin{description}
\item [{$b$}] ``The cat bowl is full at the beginning of this hour''
\item [{$d$}] ``This hour is feeding time''
\end{description}
We can translate the statements above into RoCTL{*} statements:
\begin{enumerate}
\item $\tA G(d\rightarrow\neg\tN d)$: If this hour is feeding time, the
next is not.
\item \label{enu:might-fail}$AG((d\vee\neg b)\rightarrow\eE\tN\neg b)$:
If it is feeding time or the cat bowl was empty, a single failure
may result in an empty bowl at the next step
\item $AG((\neg d\wedge b)\rightarrow\tN b)$: If the bowl is full and it
is not feeding time, the bowl will be full at the beginning of the
next hour.
\item \label{enu:must-feed}$\dO\eA G\left(d\rightarrow b\right)$: It is
obligatory that, even if a single failure occurs, it is always the
case that the bowl must be full at feeding time. 
\item $b$: The cat bowl starts full.
\end{enumerate}
Having formalised the specification it can be proven that the specification
is consistent and that the policy implies $O\eA GONb$, indicating
that the bowl must be filled at every step (in case we forget at the
next step), unless we have already failed twice. The formula $AGONb\rightarrow O\eA G\left(d\rightarrow b\right)$
can also be derived, indicating that following a policy requiring
us to always attempt to fill the cat bowl ensures that we will not
starve the cat even if we make a single mistake. Thus following this
simpler policy is sufficient to discharge our original obligation.\begin{exam}
\label{exa:bitflip}Say that a bit ought to flip at every step, but
might fail to flip at any particular step. This may be represented
with the RoCTL{*} statement 
\begin{eqnarray*}
AGO\left(b\leftrightarrow\neg Nb\right)\wedge AG\prone\left(b\leftrightarrow Nb\right)\mathcomma
\end{eqnarray*}
which is satisfied by the following model: 
\end{exam}
\begin{center}
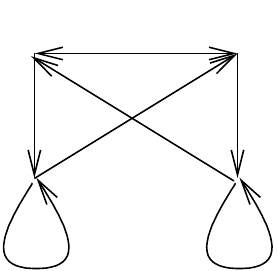
\par\end{center}

Then we may derive the following statements:
\begin{description}
\item [{$O\eA(\left(b\wedge Nb\right)\rightarrow NG\left(b\leftrightarrow\neg Nb\right))$}] If
a single failure occurs, and the bit fails to flip at the next step,
it will flip continuously from then on.
\item [{$O\eA FG\left(b\leftrightarrow\neg Nb\right)$}] Even if a single
failure occurs, there will be time at which the bit will flip correctly
from then on.
\end{description}
However, we will not be able to derive $OF\eA G\left(b\leftrightarrow\neg Nb\right)$,
as this would mean that there was a time at which a failure could
not cause the bit to miss a step.

\begin{exam}
\label{exa:fuse} Say a system has a battery that can sustain the
system for a single step, even if a failure occurs (the fuse blows).
Let $\phi$ represent ``the system has power now and at the next
step''. Then, even if a single failure occurs, it will always be
the case that even if a deviating event occurs the system will have
power now and at the next step ($O\tG\eA\phi$). It would not follow
that even if a single failure occurred the system would always have
power ($O\eA\tG\phi$); the battery power would only last one step
after the fuse blew. If we also specified that the fuse was an electronic
fuse that automatically reset, then if a single failure occurs, the
system would only have to rely on battery power for one step. Then,
if the fuse only blows once then system will always have power ($\eA\tG\forma$).
As with the $A$ operator in CTL{*}, $\eA G\forma\rightarrow G\eA\forma$
is valid in RoCTL{*} but $G\eA\forma\rightarrow\eA G\forma$ is not.\end{exam}

%% file: ONNOv.pdftex_t
\begin{picture}(0,0)%
\includegraphics{ONNOv}%
\end{picture}%
\setlength{\unitlength}{3947sp}%
\begingroup\makeatletter\ifx\SetFigFont\undefined%
\gdef\SetFigFont#1#2#3#4#5{%
  \reset@font\fontsize{#1}{#2pt}%
  \fontfamily{#3}\fontseries{#4}\fontshape{#5}%
  \selectfont}%
\fi\endgroup%
\begin{picture}(1449,1266)(736,-6481)
\put(751,-5911){\makebox(0,0)[lb]{\smash{{\SetFigFont{12}{14.4}{\rmdefault}{\mddefault}{\updefault}{\color[rgb]{0,0,0}$u$}%
}}}}
\put(1748,-6006){\makebox(0,0)[lb]{\smash{{\SetFigFont{12}{14.4}{\rmdefault}{\mddefault}{\updefault}{\color[rgb]{0,0,0}$w'\{\Viol\}$}%
}}}}
\put(1801,-6361){\makebox(0,0)[lb]{\smash{{\SetFigFont{12}{14.4}{\rmdefault}{\mddefault}{\updefault}{\color[rgb]{0,0,0}$w$}%
}}}}
\put(1351,-5386){\makebox(0,0)[lb]{\smash{{\SetFigFont{12}{14.4}{\rmdefault}{\mddefault}{\updefault}{\color[rgb]{0,0,0}$v\{p\}$}%
}}}}
\end{picture}%

%% file: BitFlipv.pdftex_t
\begin{picture}(0,0)%
\includegraphics{BitFlipv}%
\end{picture}%
\setlength{\unitlength}{3947sp}%
\begingroup\makeatletter\ifx\SetFigFont\undefined%
\gdef\SetFigFont#1#2#3#4#5{%
  \reset@font\fontsize{#1}{#2pt}%
  \fontfamily{#3}\fontseries{#4}\fontshape{#5}%
  \selectfont}%
\fi\endgroup%
\begin{picture}(1318,1290)(3436,-5530)
\put(3451,-4411){\makebox(0,0)[lb]{\smash{{\SetFigFont{12}{14.4}{\rmdefault}{\mddefault}{\updefault}{\color[rgb]{0,0,0}$\{\}$}%
}}}}
\put(4351,-4411){\makebox(0,0)[lb]{\smash{{\SetFigFont{12}{14.4}{\rmdefault}{\mddefault}{\updefault}{\color[rgb]{0,0,0}$\{b\}$}%
}}}}
\put(3751,-5161){\makebox(0,0)[lb]{\smash{{\SetFigFont{12}{14.4}{\rmdefault}{\mddefault}{\updefault}{\color[rgb]{0,0,0}$\{\Viol\}$}%
}}}}
\put(4726,-5161){\makebox(0,0)[lb]{\smash{{\SetFigFont{12}{14.4}{\rmdefault}{\mddefault}{\updefault}{\color[rgb]{0,0,0}$\{b,\Viol\}$}%
}}}}
\end{picture}%

%% file: Thesis_Background_Equiv__.tex
\providecommand{\ifnoroot}[1]{#1} 

\ifnoroot{\input{csmacros.tex}}

\providecommand{\thispaper}{this thesis}

\subsection{\label{sec:Expressive-Equivalences}Expressive Equivalences}

While \thispaper{} focuses on temporal logic, there are many ways
of defining the languages expressible by LTL\@.  This is very useful,
as it provides us with many ways to model the expressivity of temporal
logics. We are particularly interested in the expressive equivalence
of LTL with counter-free \buchi{} automata. 

In \prettyref{sub:First-Order-Definable-Lanuages}, we will outline
some important results relating to expressive equivalences, focusing
on those presented in \cite{DiGa08Thomas}. There are a number of
reasons we present these here. Firstly, by showing the many results
that \cite{DiGa08Thomas} builds upon we hope to give the reader a
feel for the complexity of attempting to follow approach of \cite{DiGa08Thomas}
in proving that LTL and counter-free \buchi{} automata have the same
expressive power. Secondly, since \cite{DiGa08Thomas} uses many results,
having a map of those results and where to find them in the paper
can be of assistance in following the work of \cite{DiGa08Thomas}.

In \prettyref{sub:Finite-DFAs-to}, we outline the proof of \cite{Wilke99classifyingdiscrete}
that any language recognised by a finite counter-free DFA can be represented
in LTL\@. We note that this result is much weaker than the theorem
of \cite{DiGa08Thomas}. However, this result  is simple and constructive.
This allows us to get an idea as to what the \formulae{} translated
from DFAs might look like, as well as an indication of the length
of the translated \formulae{}.

\subsection{\label{sub:First-Order-Definable-Lanuages}First-Order Definable
Languages}

We here present a summary of some significant results in first order
definable languages. We focus on the survey paper of \cite{DiGa08Thomas},
which provides a very powerful equivalence theorem.
\begin{theorem}
\label{thm:LTL=00003Dcounter-free}For any language L, the following
statements are all equivalent. \cite{DiGa08Thomas}\end{theorem}
\begin{enumerate}
\item L is first-order definable
\item L is star-free
\item L is aperiodic
\item L is definable in LTL
\item L is first-order definable with at most 3 names for variables
\item L is accepted by a counter-free \buchi{} automata
\item L is accepted by some aperiodic automata
\item L is accepted by some very weak automata
\end{enumerate}
Below we summarise the results that provide the basis for this theorem.
Given that the proofs are numerous and frequently complex we will
not reproduce them here. Further, since we are only interested in
counter-free \buchi{} automata and LTL we do not define the other
terms used in the theorem. Readers are invited to read \cite{DiGa08Thomas}
if they are interested in this detail.

\begin{figure}
\begin{centering}
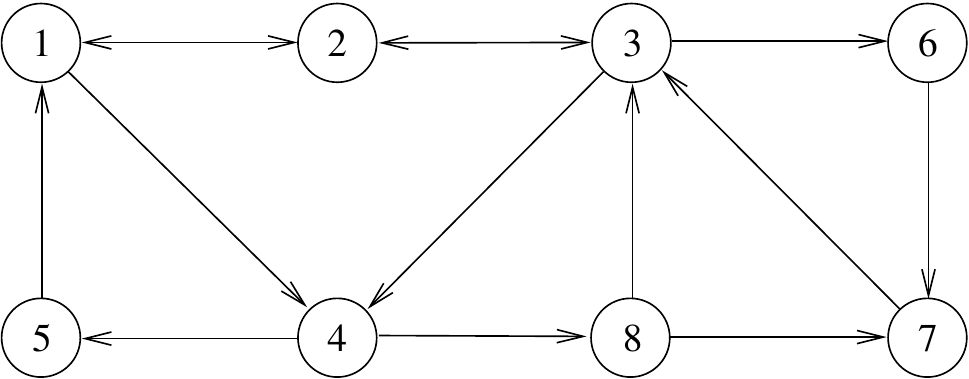
\par\end{centering}

\caption{\label{fig:Visual-Summary-of}Visual Summary of Equivalence Results
in \prettyref{thm:LTL=00003Dcounter-free}}
\end{figure}

{[}1{]}$\implies${[}4{]}: This is in essence Kamp's Theorem \cite{phd-kamp}.
Note that Kamp focuses on translating into a temporal logic with past-time
operators; however, this can be translated back into LTL \cite{567462,gabbay1994temporal}.

{[}1{]}$\iff${[}2{]}: \cite{DiGa08Thomas} cites \cite{Perrin:1986:FLS:9118.9126},
and presents a proof in their Section 4, as well as an alternate
proof of {[}1{]}$\implies${[}2{]} in their section 10.2. 

{[}2{]}$\iff${[}3{]}: \cite{Perrin:1984:RRA:645716.665310}, and
\cite[Section 6]{DiGa08Thomas} 

{[}3{]}$\implies${[}4{]}:  This is one of the more complex proofs
of this paper \cite[Section 8]{DiGa08Thomas}. It serves a similar
purpose to Kamp's theorem.

{[}3{]}$\implies${[}6{]}$\implies${[}7{]}$\implies${[}3{]}: This
is their Proposition 34, \cite[p27]{DiGa08Thomas}. This builds on
a number of results discussed in the paper. For example, {[}6{]}$\implies${[}7{]}
is trivial because since any counter-free \buchi{} automaton is periodic,
which is Lemma 29 of \cite[p25]{DiGa08Thomas}.  

{[}4{]}$\implies${[}8{]}$\implies${[}7{]}: This is mentioned at
the top of page 4. {[}4{]}$\implies${[}8{]} is Proposition 41 of
\cite[p35]{DiGa08Thomas}. The proof takes LTL \formulae{} in positive
normal form and provides a simple construction of the corresponding
weak alternating automata. {[}8{]}$\implies${[}7{]} does not appear
to be explicitly stated in the text, but a translation into \buchi{}
automata is given and in the proof of Proposition 43 \cite[p36]{DiGa08Thomas},
it is mentioned that the automata has an aperiodic transition monoid,
and so by definition is an aperiodic automata.

{[}4{]}$\implies${[}5{]}: \cite{DiGa08Thomas} describes this as
trivial and presents a simple proof (Section 7 p12--13).

{[}5{]}$\implies${[}1{]}: Obvious as {[}5{]} is a restriction of
{[}1{]}.

{[}8{]}$\implies${[}3{]}: Proposition 43 of \cite[p36]{DiGa08Thomas}.

We now present a brief outline of the path from counter-free automata
to LTL, and where they are found in \cite{DiGa08Thomas}. First it
is shown that counter-free automata are aperiodic {[}p25, lemma 29{]}.
Translating aperiodic automata into aperiodic monoids is discussed
{[}p28{]}. The most substantial part of the proof is the translation
from aperiodic monoids (or homomorphisms). The set of words and the
concatenation operator can be considered an infinite monoid {[}p13{]}.
We can choose a homomorphism from that infinite monoid to a finite
monoid. They present a factorisation of the words, and we can factorise
words of a language to produce a simplified language. The translation
into LTL has two major steps, translating the simplified language
into LTL, and showing that the existence of an LTL formula for the
simplified language demonstrates the existence of an LTL formula for
the original language.

Translating LTL to counter-free \buchi{} automata would seem significantly
more simple. The obvious powerset construction is counter-free, though
it has a Streett acceptance condition rather than \buchi{}. Note
that \cite{DiGa08Thomas} is used in \thispaper{} only for an existence
result, and so the details are not important to \thispaper{}; following
\prettyref{fig:Visual-Summary-of} counter-clockwise from {[}4{]}
to {[}6{]} is sufficient, even though this is presumably not cleanest
or simplest route possible.

\subsection{\label{sub:Finite-DFAs-to}Finite Counter-free DFAs to LTL}

\selectlanguage{english}%
\global\long\def\preautom{\autom}

\global\long\def\QtoQ{\alpha}

\selectlanguage{british}%
We here outline the proof of \cite{Wilke99classifyingdiscrete}, showing
how we may translate a counter-free DFA into an LTL formula.

For any automaton (or pre-automaton) $\preautom$, word $u$ and state
$q$. We use $u^{\preautom}\left(q\right)$ to represent the current
state of the automaton after starting at state $q$, and reading the
word $u$. For any function $\QtoQ\colon\: Q\rightarrow Q$, we let
the language $L_{\QtoQ}^{\preautom}$ be the set of words $u$ such
that $u^{\preautom}=\QtoQ$. For any set $S$, we let $u^{\preautom}\left[S\right]=\left\{ u^{\preautom}(q)\colon\: q\in S\right\} $.
\begin{theorem}
The language recognised by any counter-free DFA $\autom$ can be expressed
in LTL\@. \cite{Wilke99classifyingdiscrete}
\end{theorem}
Due to the importance of this result to \prettyref{sub:RoCTL*-to-ALTL-and-CTL*},
we will briefly outline their proof. They prove that for all words
$u$ the language $L_{\QtoQ}^{\preautom}$ can be expressed in LTL\@.
It is then clear that the language recognised by $\autom$ can be
expressed by the LTL formula:
\begin{align*}
\bigvee_{\QtoQ\st\alpha[Q_{0}]\cap F\ne\emptyset} & \textrm{LTL}\left(L_{\QtoQ}^{\preautom}\right),
\end{align*}
where $\textrm{LTL}\left(L_{\QtoQ}^{\preautom}\right)$ is the LTL
formula that defines the language $L_{\QtoQ}^{\preautom}$.  

The proof that $L_{\QtoQ}^{\preautom}$ can be expressed in LTL works
by induction, either reducing the state space at the expense of increasing
the alphabet, or shrinking the alphabet without increasing the state
space. 

They note that, since $\preautom$ is counter-free, if $u^{\preautom}\left[Q\right]=Q$
then $u^{\preautom}$ is the identity (that is $u^{\preautom}\left(q\right)=q$
for all $q\in Q$). Hence if $u^{\autom}[Q]=Q$ for all $u$ then
it is trivial to express $L_{\QtoQ}^{\preautom}$ in LTL\@. Otherwise
there is some input symbol $b$ such that $b^{\autom}[Q]$ is a strict
subset of $Q$.

They then define three languages based on $b$; $L_{0}$ the restriction
of $L_{\QtoQ}^{\preautom}$ where $b$ does not occur; $L_{1}$ the
restriction of $L_{\QtoQ}^{\preautom}$ where $b$ occurs precisely
once; and $L_{2}$ the restriction where $b$ occurs at least twice.
Let $B$ be the obvious restriction of $\autom$ such that $b$ is
removed from the input language, and let $\tilde{L}_{\QtoQ}^{B}$
be $L_{\alpha}^{B}\union\{\epsilon\}$. They also define $C$ such
that the language recognised by $C$ is similar to that of $\preautom$
except that the input symbols of $C$ are in essence words that end
in $b$, and so we can restrict the states of $C$ to be $b^{\autom}[Q]$.
Recall that $b^{\autom}[Q]$ is a strict subset of $Q$ and so we
have reduced the number of states. They define a function $h$ to
translate the words of $\preautom$ into words of $C$, and likewise
$h^{-1}$ translates the words of $C$ into words of $\autom$. They
provide the following equalities:
\begin{align*}
L_{0} & =L_{\alpha}^{B},\quad\bigcup_{\alpha=\beta b^{\autom}\beta'}\overset{L_{\beta,\beta'}}{\overbrace{\tilde{L}_{\beta}^{B}b\tilde{L}_{\beta'}^{B}}},\quad L_{2}=\bigcup\overset{L_{\beta,\gamma,\beta'}}{\overbrace{\tilde{L}_{\beta}^{B}bh^{-1}\left(L_{\gamma}^{C}\right)\tilde{L}_{\beta'}^{B}}}
\end{align*}
They let $\Gamma=\Sigma\setdiff\left\{ b\right\} $, and note that
\begin{align*}
L_{\beta,\beta'}=\tilde{L}_{\beta}^{B}b\Gamma^{*}\cap\Gamma^{*}b\tilde{L}_{\beta'}^{B} & \quad L_{\beta,\gamma,\beta'}=\Sigma^{*}b\tilde{L}_{\beta'}^{B}\cap\Gamma^{*}bh^{-1}\left(L_{\gamma}^{C}\right)\Gamma^{*}\cap\tilde{L}_{\beta}^{B}b\Sigma^{*}\mathfullstop
\end{align*}
Since $B$ has a smaller input language, and $C$ has a smaller state
space, we can assume by way of induction that $L_{\alpha}^{B}$, $\tilde{L}_{\beta}^{B}$,
$\tilde{L}_{\beta'}^{B}$ and $L_{\gamma}^{C}$ can be expressed in
LTL\@. It follows that $L_{\QtoQ}^{\preautom}$ can be expressed
in LTL\@. The result then follows from induction. 
\begin{corollary}
\label{cor:Translating-a-counter-free}Translating a counter-free
DFA into an LTL formula results in a formula of length at most $m2^{2^{\bigO\left(n\ln n\right)}}$
where $m$ is the size of the alphabet and $n$ is the number of states.
\cite{Wilke99classifyingdiscrete}
\end{corollary}
One minor note is that \cite{Wilke99classifyingdiscrete} uses stutter-free
operators so their $\left(\alpha U\beta\right)$ is equivalent to
our $N\left(\alpha U\beta\right)$; however, this is trivial to translate.

\ifnoroot{\bibliographystyle{alpha}
\bibliography{/home/john/svn/PhD/csbibtex}
}

%% file: bits/FirstOrder.pdftex_t
\begin{picture}(0,0)%
\includegraphics{bits/FirstOrder}%
\end{picture}%
\setlength{\unitlength}{4144sp}%
\begingroup\makeatletter\ifx\SetFigFont\undefined%
\gdef\SetFigFont#1#2#3#4#5{%
  \reset@font\fontsize{#1}{#2pt}%
  \fontfamily{#3}\fontseries{#4}\fontshape{#5}%
  \selectfont}%
\fi\endgroup%
\begin{picture}(4429,1722)(1653,-3783)
\end{picture}%

%% file: Body_ALinearReductionIntoQCTLstar__.tex
\providecommand{\thispaper}{this thesis}

\ifdefined\ifnoroot
\else

\providecommand\ifnoroot[1]{#1}

\ifnoroot{

\input{csmacros.tex}
}

\fi

\global\long\def\tP{\tau_{\prone}}
\global\long\def\tX{\tau}
\global\long\def\suffixQCTL{\star}
\global\long\def\QCTLof#1{#1^{\suffixQCTL}}
\global\long\def\Qforces{\QCTLof{\forces}}

\global\long\def\deviatenow{\formd}
\global\long\def\subforma{S_{\phi}}
\global\long\def\deviationtype{\epsilon}
\global\long\def\tprone{\tX^{\prone}}

\global\long\def\dthat{\hat{\deviatenow\deviationtype}}
\global\long\def\devnow{d}

\global\long\def\MQ{M}
\global\long\def\pathQCTL{\sigma}
\global\long\def\pathbQCTL{\pathb}

\global\long\def\fQ{\phi^{\suffixQCTL}}
 In this section we will present a translation of RoCTL{*} (and \RoCTLvstarLogic{})
\formulae{} into QCTL{*} such that the \formulae{} are satisfiable
in the tree semantics of QCTL{*} iff they are satisfiable in RoCTL{*}.
As we have shown that RoCTL{*} is bisimulation invariant in \prettyref{lem:RoCTL*v-is-bisimulation-invariant},
in this section we will assume that all structures are tree structures.
We will use $\Qforces$ to indicate $\forces$ is being interpreted
according to the semantics of tree QCTL{*}. 
\begin{defi}
We define a translation function $\tau^{O}$ from QCTL{*} \formulae{}
to QCTL{*} \formulae{} such that for any formula $\fQ$ 
\begin{align*}
\tX^{O}(\fQ)=A\left(NG\neg\Viol\rightarrow\fQ\right)
\end{align*}
\end{defi}
\begin{lemma}
\label{lem:O-to-QCTL}Say that $\forma$ is a RoCTL{*} formula and
$\fQ$ is a QCTL{*} formula such that for all $M$ and $\patha$ it
is the case that $M,\patha\forces\forma$ iff $\MQ,\pathQCTL\Qforces\fQ$.
Then, for all $M$ and $\patha$ it is the case that $M,\patha\forces O\forma$
iff $\MQ,\pathQCTL\Qforces\tau^{O}\left(\fQ\right)$.\end{lemma}
\begin{proof}
$\left(\Longrightarrow\right)$ Say that $M,\patha\forces O\forma$.
Then for all failure-free paths $\pi$ starting at $\sigma_{0}$,
$M,\pi\forces\phi$ and so $M,\pi\Qforces\fQ$. By definition, a path
is failure-free iff for all $i>0$ we have $\Viol\notin\valuation\left(\patha_{i}\right)$.
Since every path that satisfies $NG\neg\Viol$ is failure-free we
see that every path that starts at $\sigma_{0}$ satisfies $NG\neg\Viol\rightarrow\fQ$.
Hence $M,\sigma\Qforces A\left(NG\neg\Viol\rightarrow\fQ\right)$.

$\left(\Longleftarrow\right)$ Say that $M,\patha\Qforces A\left(NG\neg\Viol\rightarrow\fQ\right)$.
Then every path starting at $\sigma_{0}$ satisfies $NG\neg\Viol\rightarrow\fQ$.
A path that satisfies $NG\neg\Viol$ is failure-free, so every failure-free
path starting at $\sigma_{0}$ satisfies $\fQ$, and hence $\phi$.
Thus $M,\sigma_{0}\forces O\forma$. 
\end{proof}
We let $\deviatenow$ be the (Q)CTL{*} formula $NNG\neg\Viol$. Thus $\deviatenow$ does not specify whether the previous or next
transitions are failures, but requires that all transitions after
the next one be successes. The $\deviatenow$ formula is used to represent the requirement that
all transitions after a deviation must be successes. 

We define a translation function $\tprone$ from QCTL{*} \formulae{}
to QCTL{*} \formulae{} such that for any formula $\fQ$ and for some
atom $y$ not in $\fQ$: 
\begin{align*}
\tprone\left(\fQ\right)= & \Qy\left[Gy\rightarrow E\left[\left(Gy\vee F\left(y\wedge\deviatenow\right)\right)\wedge\fQ\right]\right].
\end{align*}
Note that for $\tprone\left(\fQ\right)$ to hold, $E\left[\left(Gy\vee F\left(y\wedge\deviatenow\right)\right)\wedge\fQ\right]$
must hold for all possible atoms $y$ that satisfy $Gy$, including
the case where $y$ is true only along the current fullpath $\pathQCTL$.
The diagram below shows a fullpath $\pathbQCTL$ that satisfies $F\left(y\wedge\deviatenow\right)$
for all such $y$.

\begin{center}
\begin{center}
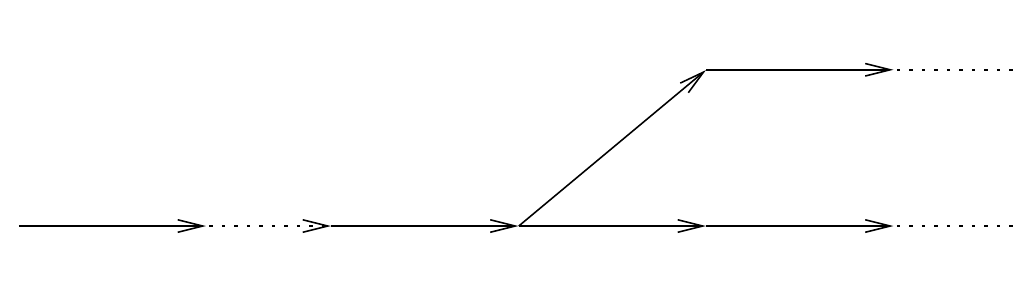
\par\end{center}
\par\end{center}

Recall from \prettyref{def:p-variant} that a $p$-variant of a structure
$M$ is a structure $M^{p}$ which values the atom $p$ differently
but is otherwise similar.
\begin{lemma}
\label{lem:Prone-to-QCTL}Say that $\forma$ is a \RoCTLvstarLogic{}
formula and $\fQ$ is a QCTL{*} formula such that for all $M$ and
$\patha$ it is the case that $M,\patha\forces\forma$ iff $\MQ,\pathQCTL\forces\fQ$.
Then, for all $M$ and $\patha$ it is the case that $M,\patha\forces\prone\forma$
iff $\MQ,\pathQCTL\forces\tprone\left(\fQ\right)$.\end{lemma}
\begin{proof}
$\left(\Longrightarrow\right)$ Say that $M,\patha\forces\prone\forma$.
Then $M,\patha\forces\forma$ or there exists a deviation $\pathb$
from $\patha$ such that $M,\pathb\forces\forma$. If $M,\patha\forces\forma$
then $\MQ,\pathQCTL\forces\fQ$ and so
\begin{align*}
\MQ,\pathQCTL\forces & \forall y\left[Gy\rightarrow E\left[Gy\wedge\fQ\right]\right],
\end{align*}
thus $\MQ,\pathQCTL\forces\tprone\left(\fQ\right)$. 

On the other hand, if $M,\patha\nforces\forma$ then, for some $i$,
there exists an $i$-deviation $\pathb$ from $\patha$ such that
$M,\pathb\forces\forma$.  If $Gy$ holds along $\pathQCTL$ then
$y$ holds at $\pathbQCTL_{i}=\pathQCTL_{i}$. As $\pathb$ is an
$i$-deviation, all transitions following $\pathb_{i+1}$ are success
transitions, so $\MQ,\pathbQCTL_{\geq i}\forces\deviatenow$ and $\MQ,\pathbQCTL\forces F\left(y\wedge\deviatenow\right)\wedge\fQ$
from which it follows that $\MQ,\pathQCTL\forces\tprone\left(\fQ\right)$.

$\left(\Longleftarrow\right)$ Say that $\MQ,\pathQCTL\forces\tprone\left(\fQ\right)$.
Then 
\begin{align*}
M^{y},\pathQCTL\forces\left[Gy\rightarrow E\left[\left(Gy\vee F\left(y\wedge\deviatenow\right)\right)\wedge\fQ\right]\right]\,,
\end{align*}
 where $M^{y}$ is any $y$-variant of $\MQ$. Consider an $M^{y}$
for which $y$ is true at a state $w$ iff $w\in\pathQCTL$. Then
$M^{y},\pathQCTL\forces E\left[\left(Gy\vee F\left(y\wedge\deviatenow\right)\right)\wedge\fQ\right]$.
Thus there exists some fullpath $\pi$ such that $\pi_{0}=\pathQCTL_{0}$
and $M^{y},\pi\forces F\left(y\wedge\deviatenow\right)\wedge\fQ$
or $M^{y},\pi\forces Gy\wedge\fQ$. 

If $M^{y},\patha^{y}\forces Gy\wedge\fQ$ then $\pi=\pathQCTL$, so
$\MQ,\pathQCTL\forces\fQ$ and $M,\patha\forces\forma$. If $M^{y},\pi\forces F\left(y\wedge\deviatenow\right)\wedge\fQ$
then there exists a non-negative integer $i$ such that $M^{y},\pi_{\geq i}\forces y\wedge\deviatenow$.
Since $y$ only occurs on the current path $\pi_{\leq i}=\sigma_{\leq i}$
and recall that the formula $\deviatenow$ indicates that we deviate
here. Thus $\pi$ is an $i$-deviation from $\sigma$ and so $M,\patha\forces\prone\forma$.
\end{proof}
We will now combine $\tX^{O}$ and $\tX^{\prone}$ to provide a translation
of \RoCTLvstarLogic{} into QCTL{*}.
\begin{defi}
\label{def:RoCTL-to-QCTL*}We let $\tau$ be a function from formulas
to formulas defined recursively as follows:

\begin{align*}
\tau\left(p\right) & =p\\
\tau\left(\neg\phi\right) & =\neg\tau\left(\phi\right)\\
\tau\left(\phi\wedge\psi\right) & =\tau\left(\phi\right)\wedge\tau\left(\psi\right)\\
\tau\left(\phi U\psi\right) & =\tau\left(\phi\right)U\tau\left(\psi\right)\\
\tau\left(N\phi\right) & =N\tau\left(\phi\right)\\
\tau\left(A\phi\right) & =A\tau\left(\phi\right)\\
\tau\left(O\phi\right) & =\tau^{O}\left(\tau\left(\phi\right)\right)\\
\tau\left(\eA\phi\right) & =\neg\tau^{\prone}\left(\neg\tau\left(\phi\right)\right)\mathfullstop
\end{align*}
\end{defi}
\begin{theorem}
\label{thm:MFinQCTL}For any \index{Logics!RoCTLv@\RoCTLvstarLogic{}}\RoCTLvstarLogic{}
formula $\forma$ of length $n$ we can produce a \logicindex{QCTL*}QCTL{*}
formula $\fQ$ of length $\bigO(n)$ by simple recursive translation
such that for any tree \RoCTLvstruct{} $M$ and fullpath $\sigma$
though $M$ we have $M,\sigma\forces\phi$ iff $M,\sigma\Qforces\fQ$.\end{theorem}
\begin{proof}
From~\prettyref{lem:O-to-QCTL} and~\prettyref{lem:Prone-to-QCTL}
above we see that $M,\sigma\Qforces\tau\left(\phi\right)$ iff $M,\sigma\forces\phi$
where $\tau$ is the translation function from RoCTL{*} \formulae{}
to QCTL{*} \formulae{} from Definition~\ref{def:RoCTL-to-QCTL*}.
\end{proof}
We will also use the above translation to show that it is possible
to decide the satisfiability of RoCTL{*} \formulae{}.
\begin{lemma}
Each \RoCTLvstarLogic{} formula $\phi$ is satisfiable in \RoCTLvstarLogic{}
iff $AGEN\neg\Viol\wedge\tau\left(\phi\right)$ is satisfiable in
the tree semantics of \logicindex{QCTL*}QCTL{*}.\end{lemma}
\begin{proof}
Recall that a valued structure is a \RoCTLvstruct{} iff $\SP\left(w\right)$
is non-empty for each world $w$ in the valued structure. The subformula
$AGEN\neg\Viol$ ensures that the translated formula is satisfiable
on a path $\sigma$ through $M$ only if $\SP\left(w\right)$ is non-empty
on all worlds $w$ reachable from $\sigma_{0}$. It is trivial to
show that removing all worlds not reachable from $\sigma_{0}$ from
$M$ does not affect whether $M,\sigma\forces\phi$. As such this
result follows simply from \prettyref{thm:MFinQCTL}.\end{proof}
\begin{theorem}
\label{thm:RoCTL*-decidable}RoCTL{*}  are decidable.\end{theorem}
\begin{proof}
Recall that every RoCTL{*} formula is a \RoCTLvstarLogic{} formula.
As RoCTL{*} is bisimulation-invariant (\prettyref{lem:RoCTL*-is-bisimulation-invariant})
we can limit our selves to tree-structures without affecting the set
of valid \formulae{}. When we limit ourselves to tree-structures
RoCTL{*} operates over the same structures as QCTL{*} and we see that
for each such structure $M$, and from the previous lemma for every
path $\sigma$ through $M$ we have $M,\sigma\forces\phi$ iff $M,\sigma\Qforces\tau\left(\phi\right)$.
Thus $\phi$ is satisfiable iff $\tau\left(\phi\right)$ is satisfiable. 

As the tree semantics for QCTL{*} are decidable \cite{EmSi84,678409},
it is obvious from \prettyref{thm:MFinQCTL} that RoCTL{*} is decidable.
\end{proof}
We can show that the above translation is also truth-preserving when
using the amorphous semantics for QCTL{*}. The argument is similar
to above, the $\forall$ operator in the amorphous semantics quantifies
over all bisimulations, and some bisimulations are tree unwindings.
These tree-unwindings will have a $y$-variant where $y$ is true
only along the current path $\sigma$ as so the $\left(\Longleftarrow\right)$
direction of the proof in \prettyref{lem:Prone-to-QCTL} works similarly
for the amorphous semantics of QCTL{*}. In the $\left(\Longrightarrow\right)$
direction we have to consider arbitrarily bisimulations under the
amorphous semantics; however, since RoCTL{*} is bisimulation invariant
this does not cause problems.

The amorphous semantics provide a model-checking procedure for RoCTL{*}.
Note that since the models are serial, all tree models have an infinite
number of worlds. On the other hand the amorphous semantics can be
model-checked; for example, by reduction to amorphous automata \cite{FrenchTim2003}.

We will not present the full proof that the above translation is also
truth preserving when the amorphous semantics are used. Firstly it
would be repetitive. The proof for the amorphous semantics is notationally
more complex as it requires bisimulations, but this merely obfuscates
the ideas central to the translation without introducing new fundamental
ideas. Secondly we will get the model checking result for free when
we introduce the translation into CTL{*} presented in \prettyref{thm:RoCTL*->CTL*}.

%% file: deviation_QCTL.pdftex_t
\begin{picture}(0,0)%
\includegraphics{deviation_QCTL}%
\end{picture}%
\setlength{\unitlength}{3947sp}%
\begingroup\makeatletter\ifx\SetFigFont\undefined%
\gdef\SetFigFont#1#2#3#4#5{%
  \reset@font\fontsize{#1}{#2pt}%
  \fontfamily{#3}\fontseries{#4}\fontshape{#5}%
  \selectfont}%
\fi\endgroup%
\begin{picture}(4902,1369)(736,-7559)
\put(4501,-6436){\makebox(0,0)[lb]{\smash{{\SetFigFont{12}{14.4}{\rmdefault}{\mddefault}{\updefault}{\color[rgb]{0,0,0}$s$}%
}}}}
\put(5251,-6436){\makebox(0,0)[lb]{\smash{{\SetFigFont{12}{14.4}{\rmdefault}{\mddefault}{\updefault}{\color[rgb]{0,0,0}$s \cdots$}%
}}}}
\put(3976,-7486){\makebox(0,0)[lb]{\smash{{\SetFigFont{12}{14.4}{\rmdefault}{\mddefault}{\updefault}{\color[rgb]{0,0,0}$\pathQCTL_{i+1}$}%
}}}}
\put(3076,-7486){\makebox(0,0)[lb]{\smash{{\SetFigFont{12}{14.4}{\rmdefault}{\mddefault}{\updefault}{\color[rgb]{0,0,0}$\pathQCTL_{i}$}%
}}}}
\put(2101,-7486){\makebox(0,0)[lb]{\smash{{\SetFigFont{12}{14.4}{\rmdefault}{\mddefault}{\updefault}{\color[rgb]{0,0,0}$\pathQCTL_{i-1}$}%
}}}}
\put(1501,-7486){\makebox(0,0)[lb]{\smash{{\SetFigFont{12}{14.4}{\rmdefault}{\mddefault}{\updefault}{\color[rgb]{0,0,0}$\pathQCTL_{1}$}%
}}}}
\put(751,-7486){\makebox(0,0)[lb]{\smash{{\SetFigFont{12}{14.4}{\rmdefault}{\mddefault}{\updefault}{\color[rgb]{0,0,0}$\pathQCTL_{0}$}%
}}}}
\put(4801,-6736){\makebox(0,0)[lb]{\smash{{\SetFigFont{12}{14.4}{\rmdefault}{\mddefault}{\updefault}{\color[rgb]{0,0,0}$\neg y$}%
}}}}
\put(4051,-6736){\makebox(0,0)[lb]{\smash{{\SetFigFont{12}{14.4}{\rmdefault}{\mddefault}{\updefault}{\color[rgb]{0,0,0}$\neg y$}%
}}}}
\put(4876,-7486){\makebox(0,0)[lb]{\smash{{\SetFigFont{12}{14.4}{\rmdefault}{\mddefault}{\updefault}{\color[rgb]{0,0,0}$\pathQCTL_{i+2}$}%
}}}}
\put(751,-7111){\makebox(0,0)[lb]{\smash{{\SetFigFont{12}{14.4}{\rmdefault}{\mddefault}{\updefault}{\color[rgb]{0,0,0}$y$}%
}}}}
\put(1576,-7111){\makebox(0,0)[lb]{\smash{{\SetFigFont{12}{14.4}{\rmdefault}{\mddefault}{\updefault}{\color[rgb]{0,0,0}$y$}%
}}}}
\put(2176,-7111){\makebox(0,0)[lb]{\smash{{\SetFigFont{12}{14.4}{\rmdefault}{\mddefault}{\updefault}{\color[rgb]{0,0,0}$y$}%
}}}}
\put(3076,-7111){\makebox(0,0)[lb]{\smash{{\SetFigFont{12}{14.4}{\rmdefault}{\mddefault}{\updefault}{\color[rgb]{0,0,0}$y$}%
}}}}
\put(4051,-7111){\makebox(0,0)[lb]{\smash{{\SetFigFont{12}{14.4}{\rmdefault}{\mddefault}{\updefault}{\color[rgb]{0,0,0}$y$}%
}}}}
\put(4876,-7111){\makebox(0,0)[lb]{\smash{{\SetFigFont{12}{14.4}{\rmdefault}{\mddefault}{\updefault}{\color[rgb]{0,0,0}$y$}%
}}}}
\put(3976,-6361){\makebox(0,0)[lb]{\smash{{\SetFigFont{12}{14.4}{\rmdefault}{\mddefault}{\updefault}{\color[rgb]{0,0,0}$\pathbQCTL_{i+1}$}%
}}}}
\put(4801,-6361){\makebox(0,0)[lb]{\smash{{\SetFigFont{12}{14.4}{\rmdefault}{\mddefault}{\updefault}{\color[rgb]{0,0,0}$\pathbQCTL_{i+2}$}%
}}}}
\put(3076,-6886){\makebox(0,0)[lb]{\smash{{\SetFigFont{12}{14.4}{\rmdefault}{\mddefault}{\updefault}{\color[rgb]{0,0,0}$\pathbQCTL_{i}$}%
}}}}
\put(2176,-6886){\makebox(0,0)[lb]{\smash{{\SetFigFont{12}{14.4}{\rmdefault}{\mddefault}{\updefault}{\color[rgb]{0,0,0}$\pathbQCTL_{i-1}$}%
}}}}
\put(1576,-6886){\makebox(0,0)[lb]{\smash{{\SetFigFont{12}{14.4}{\rmdefault}{\mddefault}{\updefault}{\color[rgb]{0,0,0}$\pathbQCTL_{1}$}%
}}}}
\put(751,-6886){\makebox(0,0)[lb]{\smash{{\SetFigFont{12}{14.4}{\rmdefault}{\mddefault}{\updefault}{\color[rgb]{0,0,0}$\pathbQCTL_{0}$}%
}}}}
\end{picture}%

%% file: LTL_non_counterfree.pdftex_t
\begin{picture}(0,0)%
\includegraphics{LTL_non_counterfree}%
\end{picture}%
\setlength{\unitlength}{3947sp}%
\begingroup\makeatletter\ifx\SetFigFont\undefined%
\gdef\SetFigFont#1#2#3#4#5{%
  \reset@font\fontsize{#1}{#2pt}%
  \fontfamily{#3}\fontseries{#4}\fontshape{#5}%
  \selectfont}%
\fi\endgroup%
\begin{picture}(1203,821)(2062,-3535)
\put(3104,-3129){\makebox(0,0)[lb]{\smash{{\SetFigFont{12}{14.4}{\rmdefault}{\mddefault}{\updefault}{\color[rgb]{0,0,0}$b$}%
}}}}
\put(2620,-3457){\makebox(0,0)[lb]{\smash{{\SetFigFont{12}{14.4}{\rmdefault}{\mddefault}{\updefault}{\color[rgb]{0,0,0}$p$}%
}}}}
\put(2608,-2897){\makebox(0,0)[lb]{\smash{{\SetFigFont{12}{14.4}{\rmdefault}{\mddefault}{\updefault}{\color[rgb]{0,0,0}$p$}%
}}}}
\put(2123,-3111){\makebox(0,0)[lb]{\smash{{\SetFigFont{12}{14.4}{\rmdefault}{\mddefault}{\updefault}{\color[rgb]{0,0,0}$a$}%
}}}}
\end{picture}%

%% file: sub,devi.tex
\global\long\def\ssfo{\Var}
\global\long\def\ssfd{\Var}

In this section we will show how to construct an automaton $\autom_{\devi\phi}$
from $\autom_{\phi}$.  Where $\autom_{\phi}$ is equivalent to $\phi$,
the automaton $\autom_{\devi\phi}$ is equivalent to $\devi\phi$.
Note that the remainder of input from the current path is irrelevant
once the deviation has occurred. Thus we may define $\autom_{\devi\phi}$
as accepting finite words terminated by a state formula indicating
that a deviation has occurred, and hence define $\autom_{\devi\phi}$
as a finite automaton.
\begin{defi}
Where $\autom_{\phi}=\left\langle 2^{\ssfo},S,\initial,\transition,F\right\rangle $
is a counter-free automaton for $\phi$, we create a finite automaton
$\autom_{\devi\phi}=\left\langle 2^{\ssfd},S_{\devi},\initial,\transition_{\devi},F_{\devi}\right\rangle $
for $\devi\phi$, where\end{defi}
\begin{enumerate}
\item $\Psi=\left\{ \psi_{s}\colon\: s\in S\right\} $, where $\psi_{s}$
is the following state formula: 
\begin{eqnarray*}
E\left(\bigwedge s\wedge NNG\neg\Viol\right)
\end{eqnarray*}
$\psi_{s}$ is roughly equivalent to saying ``if we are in state
$s$, we can deviate here''.
\item We add a state $s_{F}$ indicating that there existed an accepting
deviation from this path and so we shall accept regardless of further
input. This input relates to the original path rather than the deviation
and is thus irrelevant. As such, $S_{\devi}=S\union\sF$ and $F_{\devi}=\left\{ \sF\right\} $.
\item $\transition_{\devi}$ is the relation that includes $\transition$
but at each state also gives the option to branch into $s_{F}$ when
a deviation is possible and remain in that state regardless of the
input along the current path. That is, $\transition_{\devi}$ is the
minimal relation satisfying:

\begin{enumerate}
\item \label{enu:every-run}If for every tuple $\left\langle s,e,t\right\rangle $
in $\delta$ the tuple $\left\langle s,e,t\right\rangle \in\delta_{\devi}$.
This is to ensure that wherever $g_{\ssfo}\left(\sigma\right)$ is
a run of $\autom_{\phi}$, it is also the case that $g_{\ssfd}\left(\sigma\right)$
is a run of $\autom_{\devi\phi}$.
\item \label{enu:accept-on-deviate}For each $s\in S$ and each $e_{\devi}\in2^{\ssfd}$
such that $p_{\psi_{s}}\in e_{\devi}$ we have $\left\langle s,e_{\devi},\sF\right\rangle $
in $\transition_{\devi}$.
\item For each $e_{\devi}$ in $2^{\ssfd}$ we have $\left\langle \sF,e_{\devi},\sF\right\rangle $
in $\transition_{\devi}$.
\end{enumerate}
\end{enumerate}
The translation above is broadly similar to the translation presented
in \cite{DBLP:conf/time/McCabe-DanstedFRP09}, but we translate the
$\devi$ operator instead of the $\eE$ operator so that we can use
finite automata.

%% file: sub,devi,proof.tex
\global\long\def\ssfo{\Var}
\global\long\def\ssfd{\Var}

\global\long\def\vPp{g_{\Var}}

We fix $M$ to be some structure  such that for all worlds $w$,
\formulae{} $\psi$ and all atoms labelled $p_{E\psi}$, we have
$M,w\forces p_{E\psi}$ iff there exists a path $\sigma$ starting
at $w$ such that $M,\sigma\forces\psi$. Recall that $\autom_{\phi}=\left\langle 2^{\myPhi},S,\initial,\transition,F\right\rangle $
is the translation of $\phi$ into an automaton, and $\autom_{\devi\phi}=\left\langle 2^{\Var},S_{\devi},\initial,\transition_{\devi},F_{\devi}\right\rangle $
is the automaton constructed from $\autom_{\phi}$.

Here we present a lemma demonstrating that the translation of $\devi$
is correct.
\begin{lemma}
\label{lem:devi-correct}For any  fullpath $\sigma$ and $\ATL$
formula $\phi$ it is the case that $M,\sigma\forces\autom_{\devi\phi}$
iff there exists a deviation $\pi$ from $\sigma$ such that $M,\pi\forces\phi$.\end{lemma}
\begin{proof}

$\left(\Longleftarrow\right)$ Say that there exists a deviation $\pi$
from $\sigma$ such that $M,\pi\forces\phi$;  then there exists
an integer $i$ such that $\sigma_{\leq i}=\pi_{\leq i}$ and $\pi_{\geq i+1}$
is failure-free. Since $\pi\forces\phi$ we know from \prettyref{lem:accepts,pi,i}
that $\autom_{\phi}$ accepts $\left(\pi,i\right)$, ending in some
state $s$. As $\pi_{\geq i}\forces\bigwedge s$ and $\pi_{\geq i+1}$
is failure-free we see that $\pi_{\geq i}\forces\bigwedge s\wedge NNG\neg\Viol$,
and hence $p_{\psi_{s}}\in g_{\myPhi}\left(\pi_{i}\right)$ and so
$\left\langle s,g_{\ssfo}\left(\pi_{i}\right),\sF\right\rangle \in\transition_{\devi}$.
Thus $M,\sigma\forces\autom_{\devi\phi}$.

$\left(\Longrightarrow\right)$ Say that $M,\sigma\forces\autom_{\devi\phi}$.
Thus there is an accepting run $s_{0}\overset{\vPp\left(\sigma_{0}\right)}{\rightarrow}s_{1}\overset{\vPp\left(\sigma_{1}\right)}{\rightarrow}\cdots\rightarrow\sF$
for $\autom_{\devi\phi}$.

We know from the construction of $\autom_{\devi\phi}$ above that
$p_{\psi_{s_{i}}}\in\vPp\left(\sigma_{i}\right)$. Thus $\sigma_{\geq i}\forces p_{\psi_{s_{i}}}$
and so there exists a fullpath $\pi$ such that $\pi_{\leq i}=\sigma_{\leq i}$
and $\pi_{\geq i}\forces\bigwedge s_{i}\wedge NNG\neg\Viol$. Hence
$\pi_{\geq i+1}$ is failure-free and so $\pi$ is an $i$-deviation
from $\sigma$. Since $\pi_{\geq i}\forces\bigwedge s_{i}$ and $s_{0}\overset{\vPp\left(\sigma_{0}\right)}{\rightarrow}s_{1}\overset{\vPp\left(\sigma_{1}\right)}{\rightarrow}\cdots s_{i-1}$
is a path of $\autom_{\phi}$ we see that $\autom_{\phi}$ accepts
$\left(\pi,i\right)$. From \prettyref{lem:accepts,pi,i} we know
$\pi\forces\phi$.\end{proof}

%% file: Body_Succinctness__.tex
\providecommand{\ifnoroot}[1]{#1}

\ifnoroot{\input{csmacros.tex}

\global\long\def\States{Q}

\global\long\def\state{q}
}

\global\long\def\allnodes{\allworlds}

In the previous section we showed that a satisfaction preserving translation
from RoCTL{*} to CTL{*} exists. In this section we will show that
any satisfaction preserving translation is non-elementary in the length
of the \formulae{}. 

 We will do this by taking a class of labelled trees which we will
call $\tuple{h,l}$-utrees, where $h$ represents the height $h$
and $l$ is the number of bits per label. We will show that the number
$\#(h,l)$, of pairwise non-isomorphic $\tuple{h,l}$-utrees, is non-elementary
in $h$. We will then present ``suffix'' and ``prefix'' encodings
of utrees into RoCTL-structures, and for each pair of utrees will
define $u(T,T')$ to be the structure that results when the prefix
encoding of $T$ is joined/followed by the suffix encoding of $T'$.
For each positive $h$ and $l$ we define a RoCTL{*} formula $f\left(h,l\right)$
such that for any pair of utrees $T$ and $T'$ of height $h$ it
is the case that $u(T,T')$ satisfies $f\left(h,l\right)$ iff $T,T'$
are isomorphic. For an automaton that accepts the tree-unwinding of
$u(T,T')$ iff $T$ and $T'$ are isomorphic, once the automaton has
read the prefix encoding, the state of the automaton must give us
enough information to determine which of $\#(h,l)$ isomorphic equivalence
classes $T$ fell into. As $\#(h,l)$ is non-elementary in $h$, the
number of states in the automata must also be non-elementary in $h$.
Since there are elementary translations of CTL{*} into automata, we
will conclude that there is no elementary translation of RoCTL{*}
into CTL{*}.

\begin{defi}
We define \uline{isomorphism} on finite labelled trees recursively.
We say that $T=\treestruct$ and $T'=\left(\allnodes',\ra',g'\right)$
are isomorphic if $g\left(\func{root}\left(T\right)\right)=g'\left(\func{root}\left(T'\right)\right)$
and there exist orderings $\mathcal{C}=\tuple{C_{1},\ldots,C_{\left|\mathcal{C}\right|}}$
and $\mathcal{C}'=\tuple{C'_{1},\ldots,C'_{\mathcal{\mathcal{\left|C\right|}}}}$
of the direct subtrees of $T$ and $T'$ respectively such that $C_{i}$
and $C'_{i}$ are isomorphic for all $i\in\left[1,\left|\mathcal{C}\right|\right]$.
\end{defi}
We define \emph{utrees} below such that all $\tuple{h,l}$-utrees
have the same number of direct subtrees, which are pairwise non-isomorphic.
For any pair $T,T'$ of $\tuple{h,l}$-utrees, this ensures that if
there is a direct subtree of $T$ that is not isomorphic to any subtree
of $T'$, there must also be a direct subtree of $T'$ that is not
isomorphic to any subtree of $T$. This makes it easier to test whether
a pair of utrees are isomorphic.
\begin{defi}
We define the concept of a \uline{utree} recursively. We fix an
infinite enumerated set $\Var_{\omega}=\left\{ b_{1},b_{2},\ldots\right\} $.
A tree $T=\left\langle \allworlds,\ra,\valuation\right\rangle $ consisting
of a single node $\node$ is a $\tuple{0,l}$-utree iff $g(\node)\subseteq\Var_{l}$
where $\Var_{l}=\left\{ b_{1},b_{2},\ldots b_{l}\right\} $. We let
$\#\left(h,l\right)$ be the number of pairwise non-isomorphic $\tuple{h,l}$-utrees;
then a tree $T$ is a $\tuple{h+1,l}$-utree iff $g(\func{root}\left(T\right))=\emptyset$
and $T$ has $\left\lfloor \#\left(h,l\right)/2\right\rfloor $ direct
subtrees, which are pairwise non-isomorphic $\tuple{h,l}$-utrees.\end{defi}
\begin{exam}
\label{exa:utree}Here is an example $\tuple{1,2}$- \index{utree}utree.
We use ``11'' as shorthand for $b_{1},b_{2}$ and ``01'' as shorthand
for $b_{2}$. 
\end{exam}
\begin{center}
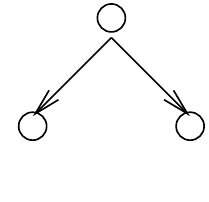
\par\end{center}
\begin{lemma}
\label{lem:hash-is-h1-exponential}The function $\#\left(h,l\right)$
is at least $\left(h+1\right)$-exponential in $l$.\end{lemma}
\begin{proof}
We see that the number of pairwise non-isomorphic $\tuple{0,l}$-utrees
is $2^{l}$. From the definition of utrees where $n=2\left\lfloor \#\left(h,l\right)/2\right\rfloor $,
\begin{align*}
\#\left(h+1,l\right) & \geq nC\left(\frac{n}{2}\right)\\
 & =\frac{n!}{\frac{n}{2}!\frac{n}{2}!}\\
 & =\frac{n.n-1\ldots\frac{n}{2}\ldots2.1}{\left(\frac{n}{2}.\ldots.2.1\right)\left(\frac{n}{2}.\ldots.2.1\right)}\\
 & =\frac{n\left(n-1\right)\ldots\left(\frac{n}{2}+1\right)}{\left(\frac{n}{2}.\ldots.2.1\right)}\\
 & \geq2^{\left(\frac{n}{2}\right)}\mathfullstop
\end{align*}
Thus when $\#\left(h,l\right)$ is $j$-exponential in $l$, it is
the case that $\#\left(h+1,l\right)$ is $\left(j+1\right)$-exponential
in $l$. As $\#\left(0,l\right)$ is singly exponential in $l$ it
follows from induction that $\#\left(h,l\right)$ is at least $\left(h+1\right)$-exponential
in $l$.
\end{proof}
 It is well known that we can describe the structure of a tree using
a string of `\{' and `\}' characters. For example, ``\{\}'' represents
a tree with a single node, and ``\{\{\}\{\}\}'' represents a tree
where the root has two root nodes as successors. Algorithms~\ref{Alg:T2prefix}~and~\ref{Alg:T2g}
for outputting the \uline{prefix encoding} $\func{prefix}\left(T\right)$
of $T$ use this principle. The function $\func{prefix}$ is from
utrees to labelled trees where each node has degree of at most one;
essentially converting the utree into a linear string of symbols.
 In addition to the atoms used to label the input tree, the prefix
encoding also uses the following atoms as labels, where $h$ is the
height of the tree and $k\in\left[0,h\right]$.
\begin{lyxlist}{00.00.0000}
\item [{$I_{\{}$}] This atom indicates that we begin the description of
a direct subtree of the tree we were describing. The current world
also encodes the label of this subtree.
\item [{$I_{\}}$}] This atom indicates that we are ending the description
of some tree.
\item [{$t_{C}$}] This indicates that the description of the subtree $C$
starts here. This is not used in function $f$ below. It is only included
to allow sections of the encoding to be easily and unambiguously referenced
in the proof of correctness.
\item [{$H_{k}$}] The current input character describes the start of a
tree of height $k$, we are at a node of height $k$. Thus $I_{\{}\wedge H_{3}$
means we are beginning the definition of a tree of height 3 and $I_{\}}\wedge H_{3}$
means we are ending the definition of a tree of height 3.
\end{lyxlist}
The final world in the prefix encoding is $w_{Z}$; the prefix encoding
is not a transition structure as $w_{Z}$ has no successor.
\begin{exam}
Below we present the prefix encoding of the \index{utree}utree $T$
from Example~\ref{exa:utree}.
\end{exam}
\begin{center}
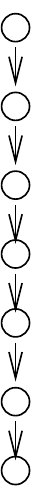\hspace{3cm}
\par\end{center}

\providecommand{\RETURNN}{\STATE{}\textbf{return}}

\begin{algorithm}[h]

\caption{T2prefix($T$)}\label{Alg:T2prefix}

\begin{algorithmic}[1]

\STATE $\tuple{\valuation,i}$:=T2g($T$,$\emptyset$,0)

\STATE $\allworlds$ := $\func{domain}\left(g\right)\union\left\{ w_{Z}\right\} $

\STATE $\rightarrow$ := $\left\{ \tuple{w_{j-1},w_{j}}:j\in[1,i)\right\} \union\tuple{w_{i},w_{Z}}$

\RETURNN $\tuple{\allworlds,\rightarrow,\valuation}$

\end{algorithmic}

\end{algorithm}

\begin{algorithm}[h]

\caption{T2g($T$,$\valuation$,$i$)}\label{Alg:T2g}

\begin{algorithmic}[1]

\STATE $\tuple{\allnodes^{T},\ra^{T},g^{T}}$:=$T$

\STATE $\valuation[w_{i}]$:=$\{I_{\{},H_{\heightT\left(T\right)},t_{T}\}\union\valuation^{T}\left(\rootT\left(T\right)\right)$;
i := i + 1

\STATE \textbf{for each} direct subtree C of T: $\tuple{\valuation,i}$:=T2g($C$,$\valuation$,$i$)

\STATE $g[w_{i}]$:=$\{I_{\}},H_{\heightT\left(T\right)}\}$; i :=
i + 1

\RETURNN $\tuple{g,i}$

\end{algorithmic}

\end{algorithm}

Strictly speaking, to be an algorithm, the \textbf{for each} in Algorithm~\ref{Alg:T2g}
must iterate over the subtrees in some order, but the ordering chosen
is unimportant and will not be defined here.

We now define the \uline{suffix encoding} $\func{suffix}\left(T\right)$
of a tree $T=\tuple{\allnodes^{T},\ra^{T},g^{T}}$. In addition to
the atoms used in the labelling of the input tree $T$, the suffix
encoding uses: the violation atom $\Viol$ from RoCTL{*}; and $H_{k}^{F}$
for $k$ in $[0,h]$ which is used to indicate the height of the current
node in the tree, much like $H_{k}$ is used in the prefix encoding.
Let $N=\left\{ \node_{1},\dots,\node_{\left|N\right|}\right\} $ be
the set of nodes in the tree $T$. Let $N'$ be a numbered set such
that $\left|N\right|=\left|N'\right|$; that is $N'=\left\{ \node'_{1},\dots,\node'_{\left|N\right|}\right\} $.
Then for all trees $T$, if $\ctlstruct=\func{suffix}\left(T\right)$
we have
\begin{enumerate}
\item $\allnodes=N\union N'\union\left\{ \node_{Z}\right\} $
\item $\ra$ is the minimal relation satisfying: $\ra\supseteq\ra^{T}$,
and 
\begin{eqnarray*}
\left\{ \tuple{\node_{i,}\node'_{i}},\tuple{\node'_{i},\node_{Z}},\tuple{\node_{Z},\node_{Z}}\right\}  & \subseteq & \ra\mathcomma
\end{eqnarray*}
`for all $i\in[1,\left|N\right|]$.
\item the valuation $g$ is the valuation satisfying $g(\node_{i})=\left\{ \Viol\right\} $;
$g\left(\node_{Z}\right)=\emptyset$ and 
\begin{eqnarray*}
g\left(\node'_{i}\right) & = & g^{T}\left(\node_{i}\right)\union\left\{ H_{\mbox{\ensuremath{\heightT}}_{\ra^{T}}\left(\node_{i}\right)}^{F}\right\} \mathfullstop
\end{eqnarray*}
\end{enumerate}
\begin{exam}
Below we present the suffix encoding of the \index{utree}utree from
Example~\ref{exa:utree}.
\end{exam}
\begin{center}
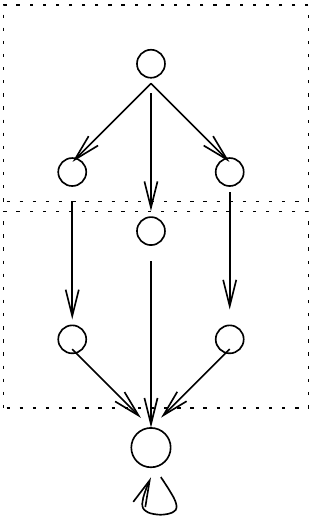
\par\end{center}
\begin{defi}
We let $u\left(T,T'\right)$ be the model that results when we join
the prefix encoding of $T$ to the suffix encoding of $T'$ by adding
$\tuple{w_{Z},\func{root}\left(T'\right)}$ to $\ra$. Formally, where $\tuple{\allnodes^{P},\ra^{P},\valuation^{P}}$ is
the prefix encoding of $T$ and $\tuple{\allnodes^{S},\ra^{S},\valuation^{S}}$
is the suffix encoding of $T'$, it is the case that $u\left(T,T'\right)=\ctlstruct$
where $\allnodes=\allnodes^{P}\union\allnodes^{S}$, $\valuation(w)=g^{S}(w)$
if $w\in\allnodes^{S}$, $\valuation\left(w\right)=\valuation^{P}\left(w\right)$
if $w\in\allnodes^{P}$, $\ra=\ra^{S}\union\ra^{P}\union\left\{ \tuple{w_{Z},\func{root}\left(T'\right)}\right\} $.
\end{defi}

\begin{defi}
Let us define a function $f$ as follows from pairs of natural numbers
to RoCTL{*} \formulae{}:
\begin{align*}
f(0,l)= & \bigwedge_{i\in\left[1,l\right]}\left(b_{i}\rightarrow F\left(H_{0}^{F}\wedge b_{i}\right)\right)\wedge\\
 & \quad\bigwedge_{i\in\left[1,l\right]}\left(\neg b_{i}\rightarrow F\left(H_{0}^{F}\wedge\neg b_{i}\right)\right)\\
f(k,l)= & \left(\left(I_{\{}\wedge H_{k-1}\right)\rightarrow\prone f(k-1,l)\right)U\left(I_{\}}\wedge H_{k}\right)\\
 & \quad\wedge FH_{k}^{F}\wedge\left(I_{\{}\wedge H_{k}\right)
\end{align*}

\end{defi}
Recall that $F\phi$ is shorthand for $\left(\top U\phi\right)$,
and as such $M,\sigma\forces F\phi\iff\exists_{i}M,\sigma_{\geq i}\forces\phi$.

The intuition behind $f$ is that a path $\sigma$ through $u\left(T,T'\right)=\left\langle \allworlds,\ra,\valuation\right\rangle $
can correspond to both a subtree of $T$ and a subtree $T'$; if $t_{C}\in g\left(\sigma_{0}\right)$
then $\sigma$ starts at the beginning of the prefix encoding of some
subtree $C$ of $T$, and if $\node'_{C'}$ is in $\sigma$ then $\sigma$
corresponds to some subtree $C'$ of $T'$. The formula $f\left(0,l\right)$
is satisfied if the labels of $C$ and $C'$ match, so $f\left(0,l\right)$
is satisfied iff $C$ and $C'$ are isomorphic leaves. A deviation
from the current path can only have one additional failure, and hence
only one additional edge. So, where $\node'_{C'}$ is in $\sigma$,
then for each subtree $D'$ of $T'$ satisfying $\heightT\left(D'\right)=\heightT\left(C\right)-1$
there exists a deviation from $\sigma$ containing $\node'_{D'}$
iff $D'$ is a direct subtree of $C'$. As such, $\prone f(0,l)$
is satisfied exactly on those paths that correspond to subtrees $C$
and $D'$ such that $C$ has a direct subtree isomorphic to $D'$.
We use this intuition and recursion to prove the following lemma.

\begin{lemma}
\label{lem:uf-iff-isomorphic}For any integers $u$ and $l$, if $T$
and $T'$ are $\tuple{u,l}$-utrees then $u\left(T,T'\right)$ satisfies
$f\left(u,l\right)$ iff $T$ and $T'$ are isomorphic.
\end{lemma}
\begin{proof}For each subtree $C$ of $T$, let $w_{C}$ be the world that is the
beginning of the suffix encoding of $C$, or more formally the world
where $t_{C}$ is true. For any path, $\sigma$ we define $\sigma_{\geq C}$
such that $\sigma_{\geq C}=\sigma_{\geq i}$ where $\sigma_{i}=w_{C}$.

$\left(\Longrightarrow\right)$ Say that $u\left(T,T'\right),\sigma^{T}\forces f\left(u,l\right)$
for some $\sigma^{T}$. We see that $\sigma_{0}^{T}=w_{0}$ as $f\left(u,l\right)\forces I_{\{}\wedge H_{u}$.
We define $\sigma^{C}$ recursively for each subtree $C$ of $T$.
Say we have defined the path $\sigma^{C}$ for some subtree $C$ such
that $u\left(T,T'\right),\sigma^{C}\forces f\left(k,l\right)$ where
$k$ is the height of $C$. Then for each direct subtree $D$ of $C$,
we see that $\sigma_{\geq D}^{C}\forces\prone f(k-1,l)$ and thus
there must exist a deviation from $\sigma_{\geq D}^{C}$ satisfying
$f(k-1,l)$, we call this deviation $\sigma^{D}$. 

We see that for each $C$ there is a unique $C'$ such that $\node'_{C'}$
is in the path $\sigma^{C}$. In the following paragraph we will show
that for each subtree $C$ and direct subtree $D$ of $C$, we can
produce $\sigma^{D}$ from $\sigma_{\geq D}^{C}$ by replacing $\node'_{C'}$
with $\node_{D'}\node'_{D'}$, and hence that $D'$ is a direct subtree
of $C'$.

Consider where $\sigma^{D}$ deviates from $\sigma_{\geq D}^{C}$.
Say $\node_{y}$ is the first world in $\sigma^{D}$ not in $\sigma^{C}$
and that $\node_{x}$ is the last world in both $\sigma^{C}$ and
$\sigma^{D}$. From the definition of deviations we see that $\sigma_{\geq\node_{y}}^{D}$
is failure-free and so the next world on $\sigma^{D}$ must be $\node'_{B}$.
Since $\sigma^{D}\forces FH_{k}^{F}$ where $k$ is the height of
$D$ it follows that $H_{k}^{F}\in\valuation\left(\node'_{B}\right)$;
from the structure of the suffix encoding it is clear that $B$ is
a direct subtree of $A$, and $\heightT\left(A\right)=k+1$ and thus
$H_{k+1}^{F}$ in $\valuation\left(\node'_{A}\right)$. As each parent
has a height greater than that of its direct subtrees, it follows
that $\node_{C'}$ is the only world in $\sigma^{C}$ such that $H_{k+1}^{F}\in\node'_{C'}$,
and hence it follows that $\node_{x}=\node_{C'}$. 

Consider $D$ of height 0. The path $\sigma^{D}$ is of the form
\begin{eqnarray*}
 &  & \left\langle w_{D},\ldots,w_{Z},\node_{T},\ldots,\node_{C'},\node_{D'},\node'_{D'},\node_{Z},\node_{Z},\ldots\right\rangle 
\end{eqnarray*}
It is easy to show that $D$ and $D'$ are isomorphic. For each $C$,
we choose $C'$ such that $\node'_{C'}$ is in the full path $\sigma^{C}$.
Say that for every $D$ of height $k$ it is the case that $D'$ and
$D$ are isomorphic. Consider $C$ of height $k+1$. We have shown
that for each direct subtree $D$ of $C$, it is the case that $D'$
is a direct subtree of $C'$. As $C$ must have the same height as
$C'$ (otherwise the requirement that $\sigma\forces FH_{k+1}^{F}$
would not be satisfied), $C'$ and $C$ have the same number of direct
subtrees, each of height $k$. We have show previously that for each
direct subtree $D$ of $C$, it is also the case that $D'$ is a direct
subtree of $C'$. By assumption, each pair $D,D'$ are isomorphic,
and so $C,C'$ are isomorphic. By induction $T$ and $T'$ are isomorphic. 

$\left(\Longleftarrow\right)$ Say that $T'$ and $T$ are isomorphic.
Clearly suffix encodings of $T'$ and $T$ will also be isomorphic,
and so $u\left(T,T'\right)$ satisfies $f\left(u,l\right)$ iff $u\left(T,T\right)$
does. Thus we can assume without loss of generality that $T=T'=\tuple{\allnodes^{T},\ra^{T},g^{T}}$.

Likewise let $\node_{C}$ be the node that is the root of the subtree
$C$. We define $\sigma^{C}$ recursively as follows: let $\sigma^{T}$
be the fullpath starting at $w_{0}$ that passes through $\node'_{0}$;
that is, $\sigma^{T}=\left\langle w_{0},\ldots,w_{Z},\node_{T},\node'_{T},\node_{Z},\node_{Z},\ldots\right\rangle $.
Say that $D$ is the direct subtree of $C$, then where 
\begin{align*}
\sigma^{C} & =\left\langle w_{C},\ldots,w_{D},\ldots,w_{Z},\node_{T},\ldots,\node_{C},\node'_{C},\node_{Z},\node_{Z},\ldots\right\rangle 
\end{align*}
we let
\begin{align*}
\sigma^{D} & =\left\langle w_{D},\ldots,w_{Z},\node_{T},\ldots,\node_{C},\node_{D},\node'_{D},\node_{Z},\node_{Z},\ldots\right\rangle \mathfullstop
\end{align*}
In other words, we produce $\sigma_{\geq D}^{C}$ from $\sigma^{C}$
by pruning everything prior to $w_{D}$, and $\sigma^{D}$ from $\sigma_{\geq D}^{C}$
by replacing $\node'_{C}$ with $\node_{D},\node'_{D}$. This remains
a full path, since $D$ is a direct subtree of $C$, and so $\node_{D}$
is a child of $\node_{C}$. Note also that $\sigma^{D}$ is a deviation
from $\sigma_{\geq D}^{C}$.

If $\heightT\left(C\right)=0$ it is easy to verify that $\sigma^{C}\forces f\left(0,l\right)$,
as $\valuation\left(w_{C}\right)\union\left\{ H_{0}^{F}\right\} =\valuation\left(\node_{C}\right)\union\left\{ H_{0},t_{C},I_{\{}\right\} $
so it is clear that $\left(\neg\right)b_{i}\rightarrow F\left(H_{0}^{F}\wedge\left(\neg\right)b_{i}\right)$.
For $C$ of height $k$, it is likewise easy to see that $\sigma^{C}\forces FH_{k}^{F}$.
Assume that $\sigma^{C}\forces f\left(k-1,l\right)$ for all $C$
of height $k-1$. Now consider $C$ of height $k$. It is easy to
show that
\begin{align*}
\sigma^{C}\forces & \left(\left(I_{\{}\wedge H_{k-1}\right)\rightarrow\bigvee_{D\mbox{ is child of C}}t_{D}\right)U\left(I_{\}}\wedge H_{k}\right)\mathfullstop
\end{align*}
By assumption $\sigma^{D}\forces f(k-1,l)$, and $\sigma^{D}$ is
a deviation from $\sigma_{\geq D}^{C}$, so $\sigma_{\geq D}^{C}\forces\prone f(k-1,l)$.
Thus 
\begin{align*}
\sigma^{C}\forces & \left(\left(I_{\{}\wedge H_{k-1}\right)\rightarrow\prone f(k-1)\right)U\left(I_{\}}\wedge H_{k}\right)\mbox{ .}
\end{align*}
Thus  $\sigma^{C}\forces f\left(k,l\right)$.  By induction $u\left(T,T'\right),\sigma^{T}\forces f\left(u,l\right)$.\end{proof}
\begin{exam}
In \prettyref{lem:uf-iff-isomorphic} above, we proved that $u\left(T,T'\right),\sigma^{T}\forces f\left(u,l\right)$
for some $\sigma^{T}$ iff $T$ and $T'$ are isomorphic.  Using $T$ as the tree in Example~\ref{exa:utree}, let
\begin{align*}
\sigma^{0} & =\left\langle w_{0},\ldots,w_{Z},\node_{1},\node'_{1},\node_{Z},\ldots\right\rangle \\
\sigma^{1} & =\left\langle w_{1},\ldots,w_{Z},\node_{1},\node_{2},\node'_{2},\node_{Z},\ldots\right\rangle \\
\sigma^{2} & =\left\langle w_{3},\ldots,w_{Z},\node_{1},\node_{3},\node'_{3},\node_{Z},\ldots\right\rangle 
\end{align*}
be paths through $u\left(T,T\right)$. We see that $\sigma^{1}$ and
$\sigma^{2}$ satisfy $f\left(0,2\right)$. As $\sigma^{1}$ and $\sigma^{2}$
are deviations from $\sigma_{\geq1}^{0}$ and $\sigma_{\geq3}^{0}$
respectively, it is the case that $\sigma_{\geq1}^{0}$ and $\sigma_{\geq3}^{0}$
satisfy $\prone f\left(0,2\right)$. Thus wherever $I_{\{}\wedge H_{0}$
is true, it is also the case that $\prone f(0,2)$ is true; hence
$u\left(T,T\right),\sigma^{0}\forces f\left(1,2\right)$.\end{exam}
\begin{defi}
We say an automaton $\autom$ accepts a structure $M$ iff the tree
unwinding of $M$ is a member of $\lang\left(\autom\right)$.
\end{defi}

\global\long\def\run{\mathfrak{R}}

\begin{lemma}
\label{lem:States-in-HAA-2exp}For any arbitrary $h,l\in\natnum$,
let $\autom=(\alphabet,\States,\States_{0},\transition,F)$ be an
SAA\index{SAA} such that for any pair $T,T'$ of $\tuple{h,l}$-utrees
$\autom$ accepts $u\left(T,T'\right)$ iff $T$ and $T'$ are isomorphic;
then $2^{\left|\States\right|}\geq\#\left(h,l\right)$.\end{lemma}
\begin{proof}
Let $\left\{ \jSeqOneTo T{\#\left(h,l\right)}\right\} $ be a set
of pairwise non-isomorphic $\left(h,l\right)$-utrees. For each $i$,
let $\run_{i}=\left\langle \allnodes_{\run_{i}},\access_{\run_{i}},\valuation_{\run_{i}}\right\rangle $
be an accepting run of $\autom$ on $u\left(T_{i},T_{i}\right)$;
let $Q_{i}$ be the set of all states that the automata is in after
reading the prefix encoding of $T_{i}$; formally let $Q_{i}\subseteq\States$
be the set of states such that for all $q\in\States$ we have $q\in Q_{i}$
iff there exists $w_{R}\in\allnodes_{\run_{i}}$ such that $\left(\func{root}\left(T_{i}\right),q\right)\in g_{\run_{i}\left(w_{R}\right)}$.
Recall that $\func{root}\left(T_{i}\right)$ is the beginning of the
suffix encoding of $u\left(T_{i},T_{i}\right)$.

Say that $Q_{i}=Q_{j}$ for some $i\ne j$. Recall from \prettyref{def:shorthand-As}
that $\autom^{q}$ is shorthand for $(\alphabet,\States,\left\{ q\right\} ,\transition,F)$.
In the next paragraph we will define a run $\run_{j}^{i}$ with the
prefix from the run $\run_{j}$ and the suffix from $\run_{i}$.

Since all infinite paths of the run $\run_{i}$ are accepting, we
see that for each $q\in Q_{i}$, the relevant subtree $\run_{i,q}^{\mbox{suffix}}$
of $\run_{i}$ is an accepting run for $\autom^{q}$ on the suffix
encoding of $T_{i}$. Let $\run_{j}^{i}$ be the tree that results
when we replace the subtree beginning at $w_{\run}$ with $\run_{i,q}^{\mbox{suffix}}$,
for each $q\in Q_{i}=Q_{j}$ and $w_{\run\in\allnodes_{\run_{i}}}$
satisfying $\valuation_{\run_{j}}$. It is easy to show that $\run_{j}^{i}$
is an accepting run of $\autom$ on $u\left(T_{j},T_{i}\right)$.
However, we have assumed that $T_{i}$ is not isomorphic to $T_{j}$,
and so $\autom$ does not accept $u\left(T_{j},T_{i}\right)$. By
contradiction $Q_{i}\ne Q_{j}$ for any $i,j\in\left[1,\#\left(h,l\right)\right]$
such that $i\ne j$. As each $Q_{i}\in2^{\States}$, we can conclude
from the pigeon hole principle that $2^{\left|\States\right|}\geq\#\left(h,l\right)$.
\end{proof}

\begin{lemma}
For all fixed $h\geq1$, there is no function $e$ which is less than
$\left(h-1\right)$-exponential, such that the length $\left|\phi_{l}\right|$
of the shortest CTL{*} formula $\phi_{l}\equiv f\left(h,l\right)$
satisfies $\left|\phi_{l}\right|<e\left(l\right)$ for all $l$. \end{lemma}
\begin{proof}
Say $e$ exists. Since $\phi_{l}\equiv f\left(h,l\right)$ then there
exists a fullpath $\sigma^{T}$ starting at $w_{0}$ through $u(T,T')$
such that $u(T,T'),\sigma^{T}\forces\phi_{l}$ iff $T$ and $T'$
are isomorphic. As $e$ is less than $(h-1)$-exponential, from \prettyref{thm:CTL2AA-single-exponential}
the size of the SAA\index{SAA} is less than $h$-exponential in $l$. 

From Lemma~\ref{lem:States-in-HAA-2exp}, we have $2^{n}\geq\#\left(h,l\right)$
where $n$ is the size of the automata, and from \prettyref{lem:hash-is-h1-exponential}
we know that $\#\left(h,l\right)$ is $\left(h+1\right)$-exponential
in $l$. Hence $2^{n}$ is at least $\left(h+1\right)$-exponential
in $l$, and so $n$ is at least $h$-exponential in $l$. By contradiction
no such $e$ exists.\end{proof}
\begin{lemma}
For all fixed $h\geq2$, there is no function $e$ which is less than
$\left(h-2\right)$-exponential such that for all RoCTL{*} \formulae{}
$\phi$ with at most $h$ nested $\prone$ (or $\eA$), the length
$\left|\psi\right|$ of the shortest CTL{*} formula $\psi$ equivalent
to $\phi$ is no more than $e\left(\left|\phi\right|\right)$. \end{lemma}
\begin{proof}
This follows from the above lemma, and the fact that $f\left(h,l\right)$
has at most $h$ nested $\prone$ and $\left|f\left(h,l\right)\right|\in\bigO\left(h+l\right)$.
\end{proof}

We can now state the main succinctness result.
\begin{theorem}
\label{thm:RoCTL*-succinct}There is no truth preserving translation
from RoCTL{*} to \index{CTL{*}}CTL{*} that is elementary in the length
of the formula.
\end{theorem}

It is easy to prove this theorem from the lemma above. We only need
to note that if there were an $i$-exponential translation of RoCTL{*}
into CTL{*} for any $i\in\natnum$ there would be an $i$-exponential
translation of RoCTL{*} \formulae{} with $i+3$ nested $\prone$
operators.

We see that the only non-classical operators in $f(h,l)$ are positively
occurring $\eE$, $U$ and $F$. Since $F\psi$ is short hand for
$\top U\psi$ we see that alternations between positively occurring
$U$ and $\prone$ are sufficient to produce non-elementary blowup.
By slightly modifying $f$, we can similarly demonstrate that alternation
between positively occurring $\robust$ and $U$ are also sufficient
to produce non-elementary blowup. For example, the following $f'$
contains only operators equivalent to negatively occurring $U$, where
$W$ is the weak until operator and $H^{F}\approx\bigvee_{i}H_{i}^{F}$
:
\begin{align*}
f'(0,l)= & \bigwedge_{i\in\left[1,l\right]}\left(b_{i}\rightarrow G\left(H^{F}\rightarrow\left(H_{0}^{F}\wedge b_{i}\right)\right)\right)\wedge\\
 & \quad\bigwedge_{i\in\left[1,l\right]}\left(\neg b_{i}\rightarrow G\left(H^{F}\rightarrow\left(H_{0}^{F}\wedge\neg b_{i}\right)\right)\right)\\
f'(k,l)= & \left(\left(I_{\{}\wedge H_{k-1}\right)\rightarrow\prone f'(k-1,l)\right)W\left(I_{\}}\wedge H_{k}\right)\\
 & \quad\wedge FH_{k}^{F}\wedge\left(I_{\{}\wedge H_{k}\right)
\end{align*}
Since there is no elementary translation of $f$ and $f'$ into CTL{*},
there is also no elementary translation of $\neg f$ and $\neg f'$
into CTL{*}.

\ifnoroot{\bibliographystyle{latex8}
\bibliography{csbibtex}

}

%% file: utree2.pdftex_t
\begin{picture}(0,0)%
\includegraphics{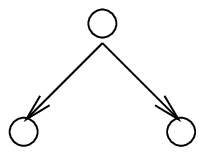}%
\end{picture}%
\setlength{\unitlength}{4144sp}%
\begingroup\makeatletter\ifx\SetFigFont\undefined%
\gdef\SetFigFont#1#2#3#4#5{%
  \reset@font\fontsize{#1}{#2pt}%
  \fontfamily{#3}\fontseries{#4}\fontshape{#5}%
  \selectfont}%
\fi\endgroup%
\begin{picture}(942,919)(1291,-1124)
\put(1936,-376){\makebox(0,0)[lb]{\smash{{\SetFigFont{12}{14.4}{\rmdefault}{\mddefault}{\updefault}{\color[rgb]{0,0,0}$\node_1$}%
}}}}
\put(2071,-1051){\makebox(0,0)[lb]{\smash{{\SetFigFont{12}{14.4}{\rmdefault}{\mddefault}{\updefault}{\color[rgb]{0,0,0}$\node_3\{11\}$}%
}}}}
\put(1306,-1051){\makebox(0,0)[lb]{\smash{{\SetFigFont{12}{14.4}{\rmdefault}{\mddefault}{\updefault}{\color[rgb]{0,0,0}$\node_2\{01\}$}%
}}}}
\end{picture}%

%% file: utree_prefix.pdftex_t
\begin{picture}(0,0)%
\includegraphics{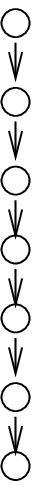}%
\end{picture}%
\setlength{\unitlength}{4144sp}%
\begingroup\makeatletter\ifx\SetFigFont\undefined%
\gdef\SetFigFont#1#2#3#4#5{%
  \reset@font\fontsize{#1}{#2pt}%
  \fontfamily{#3}\fontseries{#4}\fontshape{#5}%
  \selectfont}%
\fi\endgroup%
\begin{picture}(222,2269)(1729,-2429)
\put(1936,-331){\makebox(0,0)[lb]{\smash{{\SetFigFont{12}{14.4}{\rmdefault}{\mddefault}{\updefault}{\color[rgb]{0,0,0}$w_0\{I_\{,H_1,t_T\}$}%
}}}}
\put(1936,-691){\makebox(0,0)[lb]{\smash{{\SetFigFont{12}{14.4}{\rmdefault}{\mddefault}{\updefault}{\color[rgb]{0,0,0}$w_1\{I_\{,H_0,01,t_{\tuple{\node_{2},\emptyset,\left\{ \node_{2}\mapsto01\right\} }}\}$}%
}}}}
\put(1936,-1051){\makebox(0,0)[lb]{\smash{{\SetFigFont{12}{14.4}{\rmdefault}{\mddefault}{\updefault}{\color[rgb]{0,0,0}$w_2\{I_{\}},H_{0}\}$}%
}}}}
\put(1936,-1366){\makebox(0,0)[lb]{\smash{{\SetFigFont{12}{14.4}{\rmdefault}{\mddefault}{\updefault}{\color[rgb]{0,0,0}$w_3\{I_{\{},H_{0},11,t_{\tuple{\node_{3},\emptyset,\left\{ \node_{3}\mapsto11\right\} }}\}$}%
}}}}
\put(1936,-1681){\makebox(0,0)[lb]{\smash{{\SetFigFont{12}{14.4}{\rmdefault}{\mddefault}{\updefault}{\color[rgb]{0,0,0}$w_4\{I_{\}},H_{0}\}$}%
}}}}
\put(1936,-2041){\makebox(0,0)[lb]{\smash{{\SetFigFont{12}{14.4}{\rmdefault}{\mddefault}{\updefault}{\color[rgb]{0,0,0}$w_5\{I_{\}},H_{1}\}$}%
}}}}
\put(1936,-2356){\makebox(0,0)[lb]{\smash{{\SetFigFont{12}{14.4}{\rmdefault}{\mddefault}{\updefault}{\color[rgb]{0,0,0}$w_Z$}%
}}}}
\end{picture}%

%% file: utree_suffix.pdftex_t
\begin{picture}(0,0)%
\includegraphics{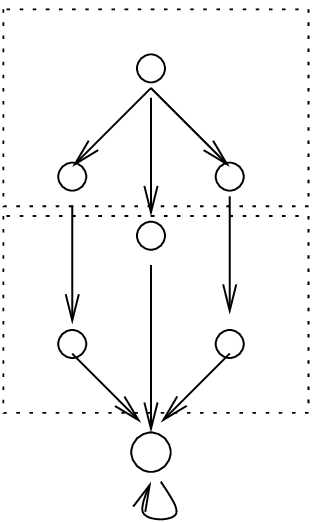}%
\end{picture}%
\setlength{\unitlength}{4144sp}%
\begingroup\makeatletter\ifx\SetFigFont\undefined%
\gdef\SetFigFont#1#2#3#4#5{%
  \reset@font\fontsize{#1}{#2pt}%
  \fontfamily{#3}\fontseries{#4}\fontshape{#5}%
  \selectfont}%
\fi\endgroup%
\begin{picture}(1422,2356)(1111,-2360)
\put(2251,-1591){\makebox(0,0)[lb]{\smash{{\SetFigFont{12}{14.4}{\rmdefault}{\mddefault}{\updefault}{\color[rgb]{0,0,0}$\node'_3\{11,H^F_0\}$}%
}}}}
\put(1936,-2131){\makebox(0,0)[lb]{\smash{{\SetFigFont{12}{14.4}{\rmdefault}{\mddefault}{\updefault}{\color[rgb]{0,0,0}$\node_Z$}%
}}}}
\put(1351,-826){\makebox(0,0)[rb]{\smash{{\SetFigFont{12}{14.4}{\rmdefault}{\mddefault}{\updefault}{\color[rgb]{0,0,0}$\node_2\{\Viol\}$}%
}}}}
\put(1936,-1096){\makebox(0,0)[lb]{\smash{{\SetFigFont{12}{14.4}{\rmdefault}{\mddefault}{\updefault}{\color[rgb]{0,0,0}$\node'_1\{H^F_1\}$}%
}}}}
\put(1351,-1591){\makebox(0,0)[rb]{\smash{{\SetFigFont{12}{14.4}{\rmdefault}{\mddefault}{\updefault}{\color[rgb]{0,0,0}$\node'_2\{01,H^F_0\}$}%
}}}}
\put(2296,-781){\makebox(0,0)[lb]{\smash{{\SetFigFont{12}{14.4}{\rmdefault}{\mddefault}{\updefault}{\color[rgb]{0,0,0}$\node_3\{\Viol\}$}%
}}}}
\put(1936,-286){\makebox(0,0)[lb]{\smash{{\SetFigFont{12}{14.4}{\rmdefault}{\mddefault}{\updefault}{\color[rgb]{0,0,0}$\node_1\{\Viol\}$}%
}}}}
\put(1126,-376){\makebox(0,0)[lb]{\smash{{\SetFigFont{12}{14.4}{\rmdefault}{\mddefault}{\updefault}{\color[rgb]{0,0,0}$\node_3\{\Viol\}$}%
}}}}
\end{picture}%

%% file: sprobConcl.tex
\newcommand{\bt}{\robust}  
We have defined a new, interesting, intuitive and expressive logic, RoCTL*, for specifying robustness in systems. The logic combines temporal and deontic notions in a way that captures the important contrary-to-duty obligations and seems free of the usual paradoxes. 

We have shown that all RoCTL* formulas can be expressed as an equivalent CTL* formula. This translation can also be used to translate RoBCTL* [McCabe-Dansted 2008] formulas into BCTL* formulas. Once translated into CTL* formula we can use any of the standard methods for model checking, so this result provides us with a model checking procedure for RoCTL*. As with CTL*, the model checking problem for RoCTL* is linear with respect to the size of the model [Clarke et al. 1999]. Classes of RoCTL* formulas with bounded $\bt$-complexity have linear translations into CTL*. Thus as with CTL* the model checking problem is also singly exponential [Clarke et al. 1999] with respect to the length of these formulas , and satisfiability is doubly exponential. Multiple nestings of $\bt$ (or $\Delta$) without any form of alternation can also be translated to CTL* without increasing the complexity of the translation over a single $\bt$ operator.

We have shown that RoCTL* is non-elementarily more succinct than CTL* for specifying some properties but we have not shown the exact complexity of the translation. However, asymptotically there is additionally one single exponential blowup per nested $\bt$ operator; never-the-less we expect model checking to be practical for some useful sub-classes of RoCTL* formulas. To verify this empirically we would need to implement the model checking procedure as a computer program. However, we have shown by hand that the given examples have translations into CTL* of reasonable length. Although a human translator can give better results than a naive computer translation, a practical model checking algorithm has the advantage that it can avoid translating automata back into CTL* and instead directly use the automaton to model check.
While in other logics non-elementary blowup is frequently the result of unbounded alternations between positive and negative occurrences of the same operator, we do not need to alternate between $\Delta$ and $\bt$ to demonstrate non-elementary blowup. Indeed, the only non-classical operators in the function $f$ were positively occurring $U$ and $\Delta$. We may modify $f$ slightly so that it only contains positively occurring $U$ and $\bt$.

RoCTL* is known to be decidable, but without a known elementary upper bound. Our succinctness result shows that a full translation into CTL* or Tree Automata cannot result in elementary decision procedures. The question still remains as to whether some other elementary decision procedure can be found for RoCTL*. The discovery of such a procedure would be interesting, as this would be the first modal logic which was elementary to decide but had only non-elementary translations into tree automata.